\documentclass[acmsmall,screen,authorversion,nonacm]{acmart}
\pdfoutput=1 
\setcopyright{none}
\usepackage{graphicx} %
\usepackage{url}
\usepackage{amsfonts,amsmath}
\usepackage{xspace}
\usepackage{color}
\usepackage[table]{xcolor}
\usepackage{mathtools}
\usepackage{soul}
\usepackage{centernot}

\usepackage{paralist}

\usepackage{tikz}
\usetikzlibrary{decorations.pathreplacing, positioning, arrows.meta}

\newcommand{\forarxiv}[1]{#1}

\newcommand{\newextmathcommand}[2]{%
    \newcommand{#1}{\ensuremath{#2}\xspace}
}

\newextmathcommand{\branch}{\mathrm{B}}
\newextmathcommand{\realvtree}{\tau}
\newextmathcommand{\branchlabelling}{\Lambda}

\usepackage{amsthm}
\usepackage{thmtools}

\theoremstyle{acmplain}

\declaretheoremstyle[
    spaceabove=.5\baselineskip\@plus.2\baselineskip\@minus.2\baselineskip,
    spacebelow=.5\baselineskip\@plus.2\baselineskip\@minus.2\baselineskip,
    bodyfont=\normalfont,
    headindent=\parindent,
    headfont=\scshape, 
    notefont=\scshape, 
    qed=$\triangleleft$
]{examplestyle}

\newcounter{myexmpcount}
\declaretheorem[
    style=examplestyle,
    name=Example,
    numberlike=myexmpcount
]{example}

\newextmathcommand{\true}{\textbf{true}}
\newextmathcommand{\false}{\textbf{false}}
\newextmathcommand{\adom}{\textbf{ADom}}
\newextmathcommand{\adomset}{\textit{ADom}}
\newextmathcommand{\rfo}{\textbf{RQ}}
\newextmathcommand{\rqfo}{\rfo}

\newextmathcommand{\rso}{\textbf{RSO}}
\newextmathcommand{\rqc}{\textbf{RQC}}
\newextmathcommand{\nip}{\textbf{NIP}}
\newextmathcommand{\udtfs}{\textbf{UDTFS}}
\newextmathcommand{\ip}{\textbf{IP}}
\newextmathcommand{\pwp}{\textbf{PWP}}
\newextmathcommand{\fptfe}{\textbf{FPTFE}}
\newextmathcommand{\ffte}{\textbf{FFTE}}
\newextmathcommand{\hsoc}{\textbf{HSOC}}
\newextmathcommand{\ehsoc}{\textbf{EHSOC}}
\newextmathcommand{\hso}{\textbf{HSO}}
\newextmathcommand{\ehso}{\textbf{EHSO}}

\newextmathcommand{\varsep}{\, \textbf{;} \,}

\newextmathcommand{\relvoc}{V}

\newextmathcommand{\mylower}{\textbf{low}}
\newextmathcommand{\myupper}{\textbf{up}}

\newcommand{\abs}[1]{\ensuremath{|#1|}\xspace}
\newcommand{\card}[1]{\ensuremath{|#1|}\xspace}

\newextmathcommand{\prespower}{\textbf{PrA}(2^{\N})}
\newextmathcommand{\trace}{\textbf{Tr}}
\newextmathcommand{\compatparamrv}{\mathrm{CompatParamRV}}
\newextmathcommand{\compatrangerv}{\mathrm{CompatRangeRV}}
\newextmathcommand{\hypoclass}{{\mathcal H}}
\newextmathcommand{\hypo}{h} %
\newextmathcommand{\concept}{{\mathcal C}}
\newextmathcommand{\classfamily}{{\mathcal F}}
\newextmathcommand{\expclass}{{\mathbb E}}
\newextmathcommand{\probspace}{{\mathcal P}}
\newextmathcommand{\randvar}{{\mathcal{RV}}}

\newcommand{\measureclass}[1]{\ensuremath{{\mathcal {D}}_{#1}}\xspace}

\newcommand{\expectationclass}[1]{\ensuremath{{\mathcal{E}}_{#1}}\xspace}

\newextmathcommand{\expectedloss}{\mathrm{ExpLoss}}
\newextmathcommand{\bestexpectedloss}{\mathrm{BestExpLoss}}
\newextmathcommand{\pshatter}{\mathrm{FatSHDim}}
\newextmathcommand{\fatshatter}{\pshatter}
\newextmathcommand{\seqfatshatter}{\pshatter^{Seq}}
\newextmathcommand{\littlestone}{\mathrm{LDim}}
\newextmathcommand{\theory}{{\rm{Th}}}
\newextmathcommand{\amodel}{\mathfrak M}
\newextmathcommand{\embmodel}{\mathcal E}
\newextmathcommand{\acmodel}{\mathfrak C}
\newextmathcommand{\glivcant}{GC}
\newextmathcommand{\conceptclass}{{\mathcal C}}
\newextmathcommand{\ptime}{\textsc{PTIME}\xspace}
\newextmathcommand{\pspace}{\textsc{PSPACE}\xspace}
\newextmathcommand{\np}{\textsc{NP}\xspace}

\newtheorem{theorem}{Theorem}
\newtheorem{proposition}{Proposition}
\newtheorem{claim}{Claim}
\newtheorem{corollary}{Corollary}
\newtheorem{lemma}{Lemma}
\newtheorem{fact}{Fact}
\newtheorem{definition}{Definition}
\newtheorem{assumption}{Assumption}
\newtheorem{remark}{Remark}
\newtheorem*{takeaway}{Takeaway}
\usepackage{thm-restate}

\newcommand{\expectation}{\mathbb{E}}
\newextmathcommand{\paramrandomized}{\mathrm{ParamRandom}}
\newextmathcommand{\setwidth}{\width}

\newcommand{\initial}{\textsc{Initial}}
\newcommand{\tupletriples}{\textsc{TupleTriples}}
\newcommand{\alignedtuples}{\textsc{AlignedTuples}}
\newcommand{\initialtuples}{\textsc{InitialTuples}}
\newcommand{\tupleof}{\textsc{TupleOf}}
\newextmathcommand{\randomelementsort}{{\mathcal K}}
\newextmathcommand{\eventsort}{{\mathcal B}}
\newextmathcommand{\randvarsort}{{\mathcal A}}
\newextmathcommand{\rademacher}{{\mathcal R}}
\newextmathcommand{\seqrademacher}{{\rademacher^{Seq}}}
\newextmathcommand{\gaussianmean}{{\mathcal G}}
\newcommand{\myeat}[1]{}
\newextmathcommand{\Conv}{\overline{\mathrm{Conv}}}
\newextmathcommand{\width}{w}
\newextmathcommand{\rangespace}{X}
\newextmathcommand{\paramspace}{Y}

\newextmathcommand{\EAvg}{\mathrm{Avg}}

\newextmathcommand{\R}{\mathbb{R}}
\newextmathcommand{\N}{\mathbb{N}}
\newextmathcommand{\X}{\mathbb{X}}

\newcommand{\nats}{\N}
\newcommand{\ints}{\mathbb{Z}}
\newcommand{\myparagraph}[1]{\par\medskip\noindent\textbf{#1.}}
\newextmathcommand{\dom}{{\text{dom}}}

\newextmathcommand{\idloss}{\ell_{\textrm{id}}}

\newcommand{\bitsize}[1]{\ensuremath{\langle{#1}\rangle}\xspace}

\newextmathcommand{\Img}{\text{Im}}

\usepackage{algorithm}
\usepackage[noend]{algpseudocode}

\algnewcommand\algorithmicforeach{\textbf{foreach}}
\algdef{S}[FOR]{ForEach}[1]{\algorithmicforeach\ #1\ \algorithmicdo}%

\title{How (and when) can you fit examples to logic-based hypothesis classes over infinite structures?}

\author{Michael Benedikt}
\orcid{0000-0003-2964-0880}
\affiliation{
  \institution{University of Oxford}
  \country{UK}
}
\email{michael.benedikt@cs.ox.ac.uk}

\author{Alessio Mansutti}
\orcid{0000-0002-1104-7299}
\affiliation{
  \institution{IMDEA Software Institute}
  \country{Spain}
}
\email{alessio.mansutti@imdea.org}

\date{\today}
\begin{document}
\begin{abstract}
We study fitting problems, sometimes called ``training problems'',  where we have a finite sample consisting of inputs and outputs,
and we want to know whether there is a function in a certain class that could produce these outputs, exactly or approximately, on the given inputs.
We focus on the computational and descriptive complexity of fitting for logically-defined classes in common decidable structures, like the real ordered field and Presburger arithmetic, and also for broader classes defined via combinatorial or model-theoretic properties. We isolate the complexity of these  fitting problems, with particular attention to cases where we can use queries  in  a natural query language over the sample to determine whether a sample is fittable.
\end{abstract}

\maketitle

\section{Introduction}

A basic computational problem in machine learning is to fit sample inputs and outputs with a function in a given class.
There are many notions of fitting possible, depending on the loss function used and whether the fit has to be exact or within some tolerance. For example, \emph{linear regression} involves fitting samples with a linear function minimizing the aggregate error.
Here we look at decision problems, asking if there is an exact fit for a set of input/output samples, or a fit with a certain bound on the error.
The case of exact fitting is often called the \emph{consistency problem}. 

Our focus will be on \emph{hypothesis classes arising from logical formulas over infinite structures}. We will first consider families of Boolean-valued functions
(a.k.a.~``concept classes'') defined by first-order formulas in a decidable infinite structure. The formula's free variables are partitioned into two parts, the ``object variables'' and the ``parameter variables''. By considering all instantiations of the parameter variables, one obtains families of Boolean-valued functions via logic.

\begin{example}[Rectangles]\label{ex:rectangleintro}
Our logic-based concept class could be the set of rectangles in the plane with sides parallel to the axes.
This class is definable by a formula~$\phi$ over the structure~$(\R,<)$. There
are four ``parameter variables'',~$y_1,\dots,y_4$, 
and two ``object variables'', $x_1$ and~$x_2$. 
The formula~$\phi$ constrains the point $(x_1,x_2)$ to lie inside the rectangle defined by the parameter variables: 
\[
\phi(x_1,x_2 \varsep y_1,y_2,y_3,y_4) \coloneqq y_1 < x_1 < y_2 \land y_3 < x_2 < y_4.
\]

In the exact fitting problem for this class, we have sequences, $p_1, \ldots, p_n$ and $q_1, \ldots, q_m$, of pairs of reals: the \emph{positive examples} and \emph{negative examples}.
We want to know if there is a rectangle
(i.e., an instantiation of the parameter variables)
that contains each $p_i$ for $i \leq n$ and excludes each $q_i$ for $i \leq m$. In the approximate fitting problem, we allow $k$ errors, where $k$ is an additional input. 
\end{example}

We also consider families of \emph{real-valued functions} definable from logical formulas. A first way of generating such families is via variable-partitioned formulas over a numerical structure, like the reals or the integers, where for each parameter the instantiated formula defines a function.

\begin{example}\label{ex:unionquadraticintro}
 Suppose our real-valued hypothesis class is the set of functions whose graph is a 
 piecewise combination of two quadratics.
 This class can be constructed from a parametric function
 \[
    f(x \varsep t, a_1, b_1, c_1, a_2, b_2 ,c_2) = 
        \begin{cases}
            a_1 \cdot x^2 + b_1 \cdot x + c_1 & \text{if } x \geq t\\ 
            a_2 \cdot x^2 + b_2 \cdot x + c_2 & \text{otherwise},
        \end{cases}
 \]
 where the parameter variables are the threshold $t$ and the coefficients of the two quadratic polynomials.
 This hypothesis class is definable by a formula over the real ordered field $(\R,+,\cdot,<)$. 
 The exact fitting problem consists of a sequence of pairs of reals $(x_1, y_1), \ldots, (x_n, y_n)$: we wish to know if there is an instantiation of the parameters such that $f(x_i)=y_i$ for each $i \leq n$.
 In the approximate fitting problem, deviations $f(x_i) \neq y_i$ are allowed, provided the total error across all $i \leq n$ stays within a given tolerance~$\epsilon$.
 \end{example}
 
A second way to obtain a real-valued hypothesis class, investigated in model theory~\cite{keislerrandomizing, itaykeisler} and in data management~\cite{sigmod22}, is to 
construct a \emph{distribution class} on top of a (concept or real-valued) class. 
The idea is to start with a logically-defined hypothesis class over an \emph{arbitrary}
structure, and consider the family of real-valued functions obtained by ``randomizing'' its parameters or its objects:

\begin{example}\label{ex:randomrectangle} Suppose our real-valued hypothesis class is parameterized by probability distributions over the points in the plane, with each function in the class mapping axis-aligned rectangles to their probability. The exact fitting problem will now consist of a sequence
of pairs $(r_1, p_1), \ldots, (r_n, p_n)$ where each $r_i$ is a rectangle (specified by endpoints) and $p_i$
is its desired probability. We wish to know if there is a probability distribution on the plane that can give rectangle $r_i$ probability $p_i$.
\end{example}

From the above we see that fitting problems have many variations. We will investigate two related questions for fitting problems over logically-defined hypothesis classes:
\begin{compactitem}
\item What is the computational complexity of exact and approximate fitting problems, in terms of the sample size, that is, the \emph{data complexity} of fitting problems?
\item When can we solve fitting problems in polynomial time~\emph{by evaluating a first-order query over the sample}? 
Informally, we search for a first-order sentence quantifying only over the sample, expressed in the signature of the underlying structure (e.g. the real ordered field), along with predicates for accessing the sample, such that evaluating the sentence on a sample tells whether it is fittable. We formalize this with the notion of \emph{restricted-quantifier formula} (Definition \ref{def:rqformula}) that arises from earlier work in database theory~\cite{anthonydiego,3belgians}.
\end{compactitem}
For instance, we will see that the fitting problem in~Example~\ref{ex:rectangleintro}
can be solved via ``first-order query evaluation over the sample''.
By extending this form of query evaluation with linear algebra, 
we can also solve fitting problems of distribution classes. 
For instance, we can show tractability of the fitting problem for Example~\ref{ex:randomrectangle} using first-order queries plus linear algebra. 

\myparagraph{Summary of Main Contributions}
We overview our contributions, leaving more fine-grained \textsc{Takeaways} to the individual sections where the relevant formal definitions are in place. 
We formalize fitting problems (exact and approximate) for logically-defined
classes over infinite structures.
For many common decidable structures, we show that these fitting problems are tractable in
the size of the sample, and can be reduced to \emph{first-order query evaluation over the sample}.
More precisely, this tractability-via-querying result holds as soon as the structure has
 \emph{restricted quantifier collapse}~\cite{hybridso}, 
a property that is known to hold for the real ordered field and Presburger arithmetic.

In the case of logically-defined concept classes, we compare the phenomenon of
fitting by first-order query evaluation over the sample with \emph{finite
VC dimension}, which characterizes efficiency in the number of
samples needed to learn to some confidence and accuracy~\cite{vctheorem}.
It is well-known that, in general, a concept class can have finite VC dimension 
but still have an intractable fitting problem (see Remark~\ref{remark:VC-only-combinatorial}).
We find that this discrepancy is more difficult to achieve for \emph{logically-defined} concept classes: we show that any such example must come from a structure without restricted quantifier collapse, and it cannot be an \emph{automatic structure}~\cite{gradel20years}.

For real-valued function classes, formed using function classes as in Example \ref{ex:unionquadraticintro} or via the randomization construction of \cite{sigmod22} as in Example \ref{ex:randomrectangle}, we show decidability (and, often, tractability) of fitting in important cases. In the tractable cases 
the algorithms use a combination of first-order querying and linear algebra. We also identify classes of fitting problems where tractability cannot be shown without resolving longstanding open problems in theoretical computer science.

In studying these fitting problems, we require a mix of techniques. In some cases we use ideas from~\emph{classical model theory}, such as results relating combinatorial dimensions of families and definability of types (see UDTFS in Section \ref{sec:fitting-definable-cc}). In other cases fitting problems connect to core computational questions concerning numerical structures like Presburger arithmetic and the real ordered field: see Section \ref{sec:functions} and Section \ref{sec:randomized}.

\myparagraph{Organization} After preliminaries in Section~\ref{sec:prelims},
Section~\ref{sec:fitting-definable-cc} establishes conditions on a structure implying that exact and approximate fitting for definable concept classes can be solved not only in \ptime, but via first-order queries over the sample.
In Section~\ref{sec:betweenptimeandrqfoconcept} we look at automatic structures, which may lack efficient fitting for all logic-based concept classes, but regain tractability when the class has finite VC dimension.
Sections~\ref{sec:functions} and~\ref{sec:randomized} move to real-valued classes. 
Section~\ref{sec:functions} extends our analysis of tractable fitting to hypothesis classes given by formulas over a numerical structure (i.e., a structure whose domain is a subset of~$\R$).
Section~\ref{sec:randomized} covers real-valued functions arising from randomizing a logically-defined class.
We close with discussion in Section~\ref{sec:conc}.

\section{Related work} \label{sec:related}
Our logic-based framework  subsumes some standard constrained optimization problems like linear regression, which are known
to be tractable. While such problems are often attacked by imposing geometric conditions on the constraints and objective function, like convexity, here we consider arbitrary parameterized families, restricting only the underlying structure.
Fitting problems have been studied for standard neural architectures \cite{trainingiscomplete, trainingfully, tighthardness}, and also in the presence of SoftMax and other exponential activation functions \cite{jonniexptraining}. Our results deal with a different set of function classes, and our classes will fix the number of parameters that can be set, unlike \cite{trainingiscomplete, trainingfully, tighthardness, jonniexptraining}. We discuss a few of the connections of our results to prior work in this area in Section~\ref{sec:functions}, see Remark~\ref{remark:neural}.

Our work has close connections with \emph{embedded finite model theory} which emerged from the database theory community \cite{anthonydiego,mbehlics,modeltheoryofstrings, generalizing, libkincdbembfmt}: in fact, it can be seen as a restricted case of the evaluation problem dealt with in embedded finite model theory, as we will explain shortly. Our analysis of fitting using queries is in the same spirit as ``in database learning'' : e.g. \cite{factregression}. Unlike those works,  we consider hypothesis classes that may already involve quantification over an infinite set. And we are interested in training using queries directly on a materialized supervised sample, not how to do composition with queries that form the sample. 

The complexity of fitting  has also been studied for logical formulas on finite structures, with the focus being conjunctive queries \cite{balderfitting}. In addition to the distinction in structures, in this prior line of work the logical formula again varies, while in our case it is fixed.  Similarly \cite{groheritzert} defines a  learning framework allowing formulas referencing either finite or infinite structures; but the focus is not learning parameters in a fixed formula over an infinite structure, and thus both the results and the techniques are incomparable to ours.

Our approach to logic-based fitting problems includes distribution classes (Section \ref{sec:randomized}), which were introduced in \cite{sigmod22}, motivated by machine learning approaches to cardinality estimation in databases. They are closely related to the ``randomization of a structure'' \cite{keislerrandomizing,ibyrandvar}, as explained in \cite{arxivwithaaron}.

\section{Preliminaries} \label{sec:prelims}

\subsection{Hypothesis classes and fitting problems} 

\begin{definition}[Hypothesis classes] 
    A hypothesis class is a set $\hypoclass=\{f_y : y \in Y\}$ of real-valued
    functions ${f_y \colon X \to \mathbb{R}}$, for some set $\rangespace$, the
    \emph{range space}. The functions in the class, also called
    \emph{hypotheses}, are indexed by elements of a \emph{parameter space}
    $\paramspace$.
\end{definition}

\begin{definition}[Concept classes]
    A concept class is a hypothesis class in which the image of all functions is
    restricted to $\{0, 1\}$. One can equivalently identify the functions in the
    class with subsets of the range space, i.e., a concept class $\conceptclass$
    can be thought of as the set ${\{\{x \in X : f(x) = 1\} : f \in
    \conceptclass\}}$.
\end{definition}

\begin{definition}[Sample]\label{def:sample}
    A \emph{sample} is a finite collection $(s_1, t_1),\ldots,(s_n, t_n)$ of
    pairs from~$\rangespace \times \R$, where the \emph{training inputs}
    $s_1,\dots,s_n$ are pairwise distinct, and $t_1,\dots,t_n$ are the
    \emph{training outputs}.
\end{definition}

Our focus is on the following computational problem:

\begin{definition}[Approximate fitting problem]\label{def:epsilonfittingproblem}
    The \emph{approximate fitting problem} for a hypothesis class $\hypoclass$
    is the decision problem of determining, given in input a sample and a
    tolerance~$\epsilon \geq 0$, whether there is a hypothesis $h \in
    \hypoclass$ that \emph{$\epsilon$-fits the sample}, that is, $h$ achieves
    ${\sum_{i=1}^n |h(s_i)-t_i| \leq \epsilon}$.
\end{definition}

We refer to the above problem as the \emph{$\epsilon$-fitting problem} when the
tolerance~$\epsilon$ is fixed rather than being part of the input.
The $0$-fitting problem requires an exact fit of the sample, and it
is thus also called the \emph{exact fitting problem}. In the case of fitting
problems for \emph{concept classes}, $\epsilon$ is a non-negative integer and
samples are of the form $(s_1, b_1) \ldots{}$ $(s_n, b_n)$, where $b_i \in
\{0,1\}$; so, they can be viewed equivalently as a collection $\{s_i : b_i=1 \}$
of \emph{positive examples}, and $\{s_i : b_i=0 \}$ of \emph{negative examples}.

\begin{example}[Rectangles, again]\label{ex:logicbasedclassesrepeat} In the
    concept class of rectangles from Example~\ref{ex:rectangleintro}, the
    \emph{range space} is $\R^2$ (all points in the plane) and the
    \emph{parameter space} is $\R^4$ (each parameter describes a
    rectangle). The approximate fitting problem for this class asks, given
    positive $S^+$ and negative $S^-$ examples and a tolerance $\epsilon$,
    if there is a rectangle $R \subseteq \R^2$ such that 
    $\card{(\R^2 \setminus R) \cap S^+} + \card{R \cap S^-} \leq \epsilon$.
\end{example}

\begin{remark}\label{remark:otherlosses}
    \!While Definition~\ref{def:epsilonfittingproblem} uses the $L^1$ norm
    ${\sum_{i=1}^n |h(s_i)-t_i|}$, all our complexity upper bounds~also hold for  
    the $L^2$ norm $\sqrt{{\sum_{i=1}^n (h(s_i){\,-\,}t_i)^2}}$ and
    the hinge loss ${\sum_{i=1}^n \max(0,1{\,-\,}{t_i {\cdot} h(s_i)})}$. 
    \mbox{See~Appendix~\ref{app:ltwoloss}.}
\end{remark}

\myparagraph{Fitting problems for logic-based concept classes}
We are interested in approximate fitting problems for hypothesis classes that
are defined in terms of first-order predicate logic on some \emph{infinite} structure.
For concept classes, this means considering~\emph{partitioned formulas}, that
is, formulas $\phi(\vec x \varsep \vec y)$ with free variables partitioned into
\emph{object variables}~$\vec x$ and \emph{parameter variables}~$\vec y$.

\begin{definition}[Definable concept class]
    Let $\amodel$ be a structure with domain $M$. Let $\phi(\vec x \varsep \vec y)$ 
    be a \emph{partitioned formula} over (the vocabulary of) $\amodel$,
    with~$j$ object variables and~$k$~parameter variables. Associated to $\phi$
    is the concept class $\conceptclass_\phi$ with range space~$M^j$ and
    parameter space~$M^k$, defined as 
    \[{\conceptclass_\phi \coloneqq \{\,\{ \vec q \in M^j : \amodel \models
    \phi(\vec q \varsep \vec p) \} : \vec p \in M^k \}}~.\]
\end{definition}

Concept classes that are definable in terms of partitioned formulas are
sometimes called \emph{definable families}. We sometimes abuse notation by
abbreviating ``partitioned formula'' by ``formula'', and we sometimes leave the
background structure $\amodel$ implicit. Given a formula~$\phi(\vec x \varsep
\vec y)$, we say that a sample with positive examples $S^+$ and negative
examples $S^-$ is \emph{$\phi$-fittable} if there is a parameter $\vec p$ such
that $\phi(\vec q \varsep  \vec p)$ holds for every $\vec q \in S^+$, and $\neg
\phi(\vec q, \vec p)$ holds for every $\vec q \in S^-$. 

\begin{remark}
    Many classical examples of concept classes arise as $\conceptclass_\phi$ for
    suitable $\phi$. For example, a binary classifier given by a feedforward
    neural network can be described by a formula of the real ordered
    field~$(\R,+,\cdot,<)$, where the inputs are objects and the weights are
    parameters.
\end{remark}

\subsection{Computational tameness assumptions on structures}
To obtain complexity bounds for approximate fitting problems of definable
families in the standard Turing machine model, we need two computational
assumptions about the structure on which the formulas are defined. We write
$\bitsize{w}$ for the length of a finite word~$w$ over some finite alphabet
$\Sigma$.

\begin{assumption}[Effective constants]
    We assume that each structure~$\amodel$ is equipped with a (infinite)
    distinguished subset ${E \subseteq M}$ of \emph{effective constants} and a
    surjection $\rho \colon D \to E$ from a set $D \subseteq \Sigma^*$ of finite
    words over some alphabet~$\Sigma$. We assume membership in $D$ is decidable
    in polynomial time: given $s \in \Sigma^*$, we can check if $s \in D$ in
    time polynomial in $\bitsize{s}$.
\end{assumption}

We tacitly identify effective constants with their \emph{representations}. For
example, in saying that a procedure ``takes as input an effective constant $c$
and runs in time polynomial in the size of $c$'', we mean that the procedure
takes as input an~$s \in D$ such that $\rho(s) = c$, and its running time is
polynomial in~$\bitsize{s}$. Representations extend to collections (tuples,
sets, etc.) of effective constants, by defining their size as the sum of the
sizes of all elements (from $D$) appearing in the collection. In fitting
problems, we assume the \emph{training inputs}~$s_i$ are tuples of (elements in
$D$ representing) effective constants and the \emph{training outputs}~$t_i$ are
rationals encoded as a pair of binary integers. Similarly, we assume the
tolerance $\epsilon$ to be a non-negative rational, again encoded as a pair of
binary integers. With these representations, we can now talk about the
complexity of fitting:

\begin{definition}[$\ptime$ approximate fitting] 
    A partitioned formula $\phi$ over a structure~$\amodel$ has
    \emph{$\ptime$ approximate fitting} whenever there is an algorithm 
    for deciding the approximate fitting problem for~$\conceptclass_\phi$ 
    that runs in polynomial time in the size of the sample.
\end{definition}

We stress that throughout the paper we always look at the complexity of fitting
problems~\emph{when formulas are fixed} (the input consists only of a sample and
a tolerance). In particular, our fitting algorithms can perform any
normalization on this formula (e.g. quantifier elimination) ``for free''.  Note
also that having $\ptime$ approximate fitting for each partitioned formula
$\phi(\vec x \varsep \vec y)$ does not imply even decidability of the
approximate fitting problem when $\phi$ is not fixed.

\begin{assumption}[fixed-formula tractable evaluation]
    We assume that for any first-order formula $\phi(x_1 \ldots x_k)$, there is an
    algorithm that takes as input a $k$-tuple of  effective constants $\vec c$, and
    determines whether $\phi(\vec c)$ holds in the structure, running in polynomial
    time in the total size of $\vec c$. We refer to this property of a structure as
    \emph{fixed-formula tractable evaluation} ($\ffte$).
\end{assumption}

\begin{example} \label{ex:standardstructures}
    Many commonly studied decidable structures have a natural set of effective
    constants where $\ffte$ holds. In particular, this property holds for any
    structure with quantifier elimination where checking of quantifier-free formulas
    is tractable. Examples include the real ordered field $(\R,+,\cdot,<)$, 
    where the effective constants are the rationals, 
    or additive natural number arithmetic $(\N,+,<)$, also known as
    \emph{Presburger arithmetic}, where the effective constants are all naturals.
\end{example}

Unless otherwise stated, 
\ul{we always assume that structures have effective constants and $\ffte$}.

\subsection{Combinatorial and definability properties of structures}

Two model theoretic properties will play a key role: finite VC
dimension and~\rqc.

\begin{definition}[VC dimension] \label{def:vcandfat} 
    We say a concept class $\conceptclass$ on range space $\rangespace$
    \emph{shatters} a finite subset $S$ of $\rangespace$ whenever, for every $E
    \subseteq S$, there is $h \in \conceptclass$ such that $h(s)=1$ for $s \in
    E$ and $h(s)=0$ for $s \in X \setminus E$. The \emph{VC-dimension of
    $\conceptclass$} is the supremum of the sizes of subsets of $\rangespace$
    shattered by $\conceptclass$. A concept class has \emph{finite VC dimension}
    whenever there is a finite bound on the sizes of sets it can shatter.
\end{definition}

\begin{example}[Rectangles have finite VC dimension]\label{ex:rectanglesvc} 
    The class of rectangles from Example~\ref{ex:rectangleintro} has VC dimension $4$:
    there is a set of size $4$ that you can shatter, namely a diamond-shaped set
    in the plane. However, it is rather simple to see that there is no such set
    of size $5$.
\end{example}

\begin{remark}\label{remark:VC-only-combinatorial} 
    Finite VC dimension does not imply a tractable fitting problem. There are
    effectively presented classes with finite VC dimension which do not have
    \ptime fitting (under standard complexity-theoretic assumptions), see for
    instance~\cite[Section~3]{learninghardnesspittvaliant},
    \cite[Section~7]{learninghardnessanthony}.
\end{remark}

\begin{remark}\label{remark:vc-and-pac}
    Finite VC dimension is necessary and sufficient for Probably Approximately Correct (PAC) learnability from random examples~\cite{vctheorem}. The definition of PAC learning will not be important in our work but
    if a class has finite
    VC dimension, any fitting algorithm suffices to PAC learn.
    Conversely, infinite VC dimension makes PAC learning impossible, even
    when not requiring an efficient algorithm.
\end{remark}

\begin{definition}[Finite VC dimension of formulas, $\nip$, and~$\ip$]\label{def:finitevcnip} 
    A partitioned formula~$\phi(\vec x \varsep \vec y)$ over a
    structure~$\amodel$ is said to have \emph{finite VC dimension} if the
    corresponding concept class~$\conceptclass_\phi$ does. We say that a
    structure $\amodel$ is $\nip$ (short for ``Not having the Independence
    Property'') when every partitioned formula over~$\amodel$ has finite VC
    dimension. Otherwise, $\amodel$ is said to be $\ip$.
\end{definition}

\begin{example}[Examples of $\nip$ structures]\label{ex:nip} 
    Both the real ordered field and Presburger arithmetic are $\nip$ \cite[App A]{nipbook}. The
    extension of Presburger arithmetic with the power predicate~$2^{\N} \coloneqq \{ x \in \N : \text{$x$ is a power of two} \}$, denoted~\prespower, is also \nip~\cite{prespowernip}.
\end{example}

We mentioned in the introduction that we are interested in whether we can solve
fitting problems with a~\emph{first-order query that quantifies only over the elements in
the sample}, not over the infinite structure. This is a special case of the
notion of \emph{restricted quantifier formula} from the database
literature~\cite{anthonydiego,3belgians} 
which we review. Below,
$\amodel$ is a structure with domain $M$ and signature~$L$.

\begin{definition}[Restricted quantifier formula] \label{def:rqformula}
    Let $\relvoc$ be a (finite) relational signature disjoint from~$L$. A
    \emph{restricted quantifier formula ($\rqfo$~formula)} in the signature $L
    \cup \relvoc$ is a formula built up inductively from (possibly quantified)
    first-order formulas in the signature $L$, via Boolean operators and
    \emph{restricted quantifiers} of the form $\exists \vec x \left( T(\vec x,
    \vec y) \wedge \phi(\vec x, \vec y) \right)$, where $T$ is a relation symbol
    from $\relvoc$, $\vec x$ and $\vec y$ are tuples of variables, and $\phi$ is
    an~$\rqfo$ formula.
\end{definition}

An \emph{embedded finite model} for $\relvoc$ over $\amodel$ is an
interpretation of each relation in $\relvoc$ by a finite set of tuples
in~$\amodel$. Let $\embmodel$ be such a model. By expanding $\amodel$ with
$\embmodel$ we obtain a structure, denoted $\amodel,\embmodel$, for the
signature $L \cup \relvoc$. We say that a formula~$\phi$ in this signature
\emph{holds in $\embmodel$} whenever the structure $\amodel, \embmodel$
satisfies $\phi$ in the usual sense, i.e.,~$\amodel,\embmodel \models \phi$. Two
formulas $\phi_1$ and $\phi_2$ in the signature $L \cup \relvoc$ are
\emph{equivalent over embedded finite models for $\relvoc$ over $\amodel$}
whenever ${\amodel, \embmodel \models \phi_1 \leftrightarrow \phi_2}$ holds for
every embedded finite model $\embmodel$ for~$\relvoc$ over~$\amodel$. We omit
``over $\amodel$'' when clear from context.

\begin{remark}[Active domain]\label{remark:activeisenough} 
    The \emph{active domain} of an embedded finite model $\embmodel$ for
    $\relvoc$ is the finite subset of~$M$ obtained by taking the
    projections of each interpretation of each relation symbol in
    $\embmodel$. Without loss of generality, we can let $\rqfo$ formulas replace restricted quantifiers with
    quantifiers over the active domain~$\adom$. For example, if $V =
    \{R(x_1),S(x_1,x_2)\}$, the formula ${\exists x \in \adom \ \phi}$ stands
    for $\exists x\,(R(x) \land \phi) \lor \exists x, y\,(S(x,y) \land \phi)
    \lor \exists x, y\,(S(y,x) \land \phi)$. We~call these \emph{active domain
    quantifiers}.
\end{remark}

\begin{definition}[Restricted Quantifier Collapse]\label{def:rqc}
We say that~$\amodel$ has Restricted Quantifier Collapse (`` $\amodel$ is $\rqc$'') whenever for every relational signature $V$ disjoint from $L$, every first-order formula in the signature $L\cup \relvoc$ is equivalent, over embedded finite models for $V$, to an~$\rqfo$ formula.
\end{definition}

\begin{remark}\label{remark:rqc-nip}
   It is known that $\rqc$ implies $\nip$, see %
   \cite[Proof of Lemma~3.9]{modeltheoryofstrings}. 
\end{remark}

\begin{example}[(Non) examples of $\rqc$ structures]\label{ex:rqc}
All structures listed in Example \ref{ex:nip}  are ~{\rqc}. %
For~$\prespower$ we give an argument in Appendix~\ref{app:fragile}. 
The other examples are well-known, and indeed follow by general model-theoretic criteria that
are known to imply $\rqc$, like o-minimality, NFCP~\cite{pseudofinite}, and distality~\cite{isolation} (see Appendix \ref{app:fragile} for a definition).

Some $\nip$ structures are not~\rqc. One example is the structure on the natural numbers with the equivalence
relation $E(n, m)$ holding when the largest power of two below $n$ is the same as the largest below~$m$. Note that
equivalence classes of this relation are finite, but the size of classes is unbounded.  
Let ${\relvoc = \{C(x)\}}$ and consider the sentence ${\exists x ~ \forall y ~ [ E(x,y) \rightarrow C(y)]}$ stating that $C$ contains an equivalence class. 
One can show that this is not equivalent to an \rqfo sentence. 
\end{example}

\section{RQ Fitting for definable classes}\label{sec:fitting-definable-cc}

\begin{takeaway} We show that \rqc implies \ptime  approximate fitting. Via the
    notion of Uniform Definability of Types over Finite Sets (\udtfs), we prove
    that~\rqc is equivalent to ``$\rqfo$ fitting''---fitting via an $\rqfo$
    sentence---even when restricting $\rqfo$ fitting to partitioned formulas
    with a single parameter.
\end{takeaway}

Throughout the section, let $\amodel$ be a structure with domain $M$ and
signature~$L$, and $\relvoc$ be a relational signature disjoint~from~$L$. Our
study of fitting problems for definable concept classes starts from the
following motivating example, which illustrates why \rqc is a very natural
notion in this context.

\begin{example}\label{example:rqc-is-nice} Consider the partitioned formula
$\phi(x \varsep y) \coloneqq x < y$ over the structure $(\R,<)$. Introducing two
unary predicates $P$ and $N$, to be interpreted as the sets of positive and
negative examples of an input sample, the $0$-fitting problem for $\phi$ can be
expressed~as the following first-order sentence:
\[ 
    \exists y \left( \forall x\,(P(x) \rightarrow \phi(x,y)) \land \forall x \,(N(x) \rightarrow \lnot \phi(x,y))\right).
\]
This general pattern works for any $\phi$, with the obvious adjustments to the
number of object variables and parameter variables. But a sample is $(x <
y)$-fittable if and only if the positive examples lie below the negative ones,
and so the problem can alternatively be encoded as the~\rqfo sentence $\forall x
\,(P(x) \rightarrow \forall z \,(N(z) \rightarrow x < z))$. This sentence
describes a \ptime algorithm for the $0$-fitting problem for $\phi$: iterate
through all positive examples $x$ and all negative ones $z$, checking $x <
z$.%
\end{example}

Example~\ref{example:rqc-is-nice} motivates the following notion: 

\begin{definition}[$\rqfo$ fitting]\label{def:rqfo-fitting}
    Consider a partitioned formula $\phi(x_1,\dots,x_j \varsep \vec y)$
    over~$\amodel$, and let $\relvoc = \{P,N\}$ be a relational signature with
    two $j$-ary relations. We say that $\phi$ has \emph{$\rqfo$ fitting} if
    there is an \rqfo sentence $\phi_{\textit{fit}}$ in the signature $L \cup
    \relvoc$, with the following property: consider a finite sample~$S$ of
    positive examples $S^+$ and negative examples $S^-$, and let $\embmodel$ be
    the embedded finite model for $\relvoc$ over $\amodel$ that interprets $P =
    S^+$ and $N = S^-$. Then, $\phi_{\textit{fit}}$ holds in $\embmodel$ if and
    only if $S$ is $\phi$-fittable.

    We say that~\amodel has \rqfo fitting whenever every partitioned formula
    over \amodel has \rqfo fitting.
\end{definition}

We remark that, as with \rqc (Definition~\ref{def:rqc}), the notion of \rqfo
fitting is semantic and not algorithmic: it must hold for all finite sets~$S$,
not only those composed of effective constants.

Recalling that we assume effective constants and $\ffte$, the problem of
evaluating a fixed~\rqfo sentence on a given embedded finite model is in \ptime:

\begin{proposition}\label{prop:rqfoimpliesptime} Fix an~\rqfo sentence $\phi$ in
$L \cup \relvoc$. There is a  $\ptime$ algorithm that given an embedded finite
model $\embmodel$ for~$\relvoc$ over~$\amodel$, decides~$\amodel,\embmodel
\models \phi$. Thus, any formula with $\rqfo$ fitting has $\ptime$ $0$-fitting.
\end{proposition}

By ``takes as input an embedded finite model'' we mean that it takes as input
lists of tuples of (representations of) effective constants from $M$, encoding
the interpretation of the relations in $\relvoc$.

We extend Proposition \ref{prop:rqfoimpliesptime} from exact to approximate
fitting, at the cost of assuming \rqfo fitting for the whole structure instead
of single formulas.

\begin{restatable}{theorem}{thmrqfoimpliesptimeapprox}\label{thm:rqfoimpliesptimeapprox}
    If~$\amodel$ has $\rqfo$ fitting, then any formula over $\amodel$ has
    $\ptime$ approximate fitting.
\end{restatable}

By definition, \rqc implies \rqfo fitting, and so from
Theorem~\ref{thm:rqfoimpliesptimeapprox} we obtain:

\begin{corollary}\label{cor:rqcimpliesrqfoandptimefitting} If~$\amodel$
    is~$\rqc$, then every formula $\phi$ over~$\amodel$ has $\rqfo$ fitting and
    \ptime approximate fitting. In particular, this holds when $\phi$ is a
    formula over any of the structures in~Example~\ref{ex:nip}.
\end{corollary}

\subsection{Ideas behind the proof of
Theorem~\ref{thm:rqfoimpliesptimeapprox}}\label{sec:ideas-rqfoimpliesptimeapprox}

The proof of Theorem~\ref{thm:rqfoimpliesptimeapprox} has two steps.  
First, we show $\rqfo$ fitting for all formulas implies $\amodel$ is $\nip$. We
 prove something stronger: \rqfo~fitting for a single parameter suffices. This
 will be useful later.

\begin{restatable}{lemma}{lemQfittingImpliesFiniteVC}\label{lem:RQfittingImpliesFiniteVC}
    If every formula~$\phi(\vec x \varsep y)$ over $\amodel$ has \rqfo fitting,
    then~$\amodel$ is $\nip$.
\end{restatable}

We assume for contradiction that~$\amodel$ has a partitioned formula, which we
can take to be of form $\phi(\vec x \varsep y)$, with infinite VC dimension.
From this we construct a formula $\gamma(\vec z \varsep y)$ (still with a single
parameter) capturing $3$-colorability: finite graphs can be encoded as samples
so that a sample is $\gamma$-fittable if and only if the graph is $3$-colorable.
We then show that $\gamma$ having \rqfo fitting implies that there is a $\rqfo$
sentence~$\gamma'$ over the simpler structure $(\mathbb{Z},+,<)$ with additional
binary relation $E$ such that, on embedded finite models, $\gamma'$ holds
exactly when the graph induced by $E$ is $3$-colorable. This gives a
contradiction: since $\gamma'$ uses only order and addition, evaluating
$\gamma'$ on an embedded finite model is in $AC^0$, but $3$-colorability is not
in $AC^0$~\cite{FurstSS84}. 

In the second step of the proof, we use the fact that $\amodel$ is \nip and has
\ptime $0$-fitting to provide a \ptime algorithm for approximate fitting. By the
Shelah-Sauer lemma (Fact~\ref{fact:sslemma}, below), \nip implies that the
number of \emph{realizable partitions} of a set $S$ of elements in the range
space (i.e., the partitions of $S$ into positive and negative examples that can
be fitted exactly) is polynomial in the size of $S$.

\begin{fact}[\cite{shelahsauer1972}]\label{fact:sslemma} Let $\conceptclass$ be
    a concept class on range space $\rangespace$ and with VC dimension $d \in
    \N$. For every set $S \subseteq X$ of size~$m$, the set $\{\{ s \in S : f(s)
    = 1\} : f \in \conceptclass \}$ has size at most $\sum_{i=0}^{d} {m \choose
    i}$.
\end{fact}

Fact~\ref{fact:sslemma} does not provide a method for constructing the
realizable partitions, but we can use PTIME 0-fitting to circumvent this
problem: by issuing polynomially-many evaluations to $\rqfo$ sentences
(in~\ptime by Proposition \ref{prop:rqfoimpliesptime}) it is possible to compute
all realizable partitions in polynomial time (see
Lemma~\ref{lem:computealltypes} in Appendix~\ref{app:fromexacttoapprox}). Then,
approximate fitting reduces to determining whether the smallest distance between
any realizable partition and the given sample is at most the
tolerance~$\epsilon$.

\subsection{The relationship between $\rqc$ and $\rqfo$ fitting} 
\rqfo fitting captures the idea of being able to solve a fitting problem by
evaluating an \rqfo sentence over the positive and negative examples. A priori,
this seems like a much weaker property than \rqc, which allows for more general
queries over (arbitrary) embedded structures. Surprisingly, this intuition turns
out to be wrong: \rqfo fitting is \emph{equivalent} to \rqc, with this
equivalence holding even when \rqfo fitting is considered \emph{only for
formulas with a single parameter variable}.

\begin{restatable}{theorem}{TheoremRQCandRQFitting}\label{thm:rqcandrqfitting} \
$\amodel$ has $\rqfo$ fitting $\iff$ Every $\phi(\vec x \varsep y)$ over
$\amodel$ has $\rqfo$ fitting $\iff$ $\amodel$ is $\rqc$.
\end{restatable}

\begin{proof}[Proof sketch]
The interesting direction is from one-parameter $\rqfo$ fitting to $\rqc$.
By~Lemma~\ref{lem:RQfittingImpliesFiniteVC}, one-parameter $\rqfo$ fitting
implies that $\amodel$ is $\nip$. We rely on a result in model theory of
Chernikov and Simon~\cite{simon-chernikov-two}, showing that $\nip$ structures
have a property called \emph{``Uniform Definability of Types over Finite Sets''
($\udtfs$)}. Formally, a partitioned formula $\gamma(\vec x \varsep \vec y)$
with $j$ object variables and $k$ parameter variables has $\udtfs$ (with respect
to $\amodel$) whenever:

\smallskip
\textit{There is an $L$ formula $\delta(\vec y \varsep \vec p)$ such that for
any finite set $S$ of $k$-tuples in $\amodel$ of cardinality at least~two, for
any $j$-tuple $\vec x_0$, there is %
$\vec p_0$ from $S$ such that $\amodel \models \forall \vec y \in S ~
\gamma(\vec x_0, \vec y) \leftrightarrow \delta(\vec y, \vec p_0)$.}
\smallskip

Let $V$ be a relational signature disjoint from the signature $L$ of $\amodel$.
By induction on the number of unrestricted quantifiers, we show that every
first-order formula in~$L \cup V$ is equivalent over embedded finite models for
$\amodel$ to an~$\rqfo$ formula. The induction step considers a formula
\[
\psi(\vec z) \coloneqq \exists y ~ Q_1(u_1) \ldots Q_k(u_k) ~ \gamma(y, \vec u, \vec z),
\]
where $\gamma$ is an $L$ formula, $\exists y$ is an unrestricted quantifier,
and each $Q_i(u_i)$ is an active domain quantifier (the induction hypothesis
ensures we can restrict our attention to these quantifiers).

For simplicity, we restrict the proof sketch to embedded finite models with
active domain of at least two elements, and assume that~$\gamma$ is an
$L$-formula. 
We apply $\udtfs$ on $\gamma$ to get a formula~$\delta(\vec u \varsep \vec p)$.
Let $\eta(\vec z, \vec p) \coloneqq \exists y\,\forall \vec u_1 \in \adom \dots
\forall u_k \in \adom \colon \gamma(y, \vec u, \vec z) \leftrightarrow
\delta(\vec u, \vec p)$.

The formula $\psi(\vec z)$ is equivalent, over embedded finite models with
active domain at least two, to 
\[
    \exists p_1 \in \adom \dots \exists p_k \in \adom \big(\eta(\vec z, \vec p) \land 
    \exists y ~ Q_1(u_1) \ldots Q_k(u_k) \colon \delta(\vec u, \vec p)\big).
\] 
The variable $y$ does not appear in $\delta$, and thus $\exists y$ be eliminated
from the rightmost conjunct. It remains to convert $\eta$ into an $\rqfo$
formula. We consider the fitting problem for the one-parameter formula
$\eta_{\textit{qf}}(\vec u, \vec z, \vec p \varsep y) \coloneqq (\gamma(y, \vec u, 
\vec z) \leftrightarrow \delta(\vec u, \vec p))$, on positive examples only.
This problem corresponds to the sentence $\exists y \,\forall \,\vec u, \vec z,
\vec p \colon P(\vec u, \vec z, \vec p) \rightarrow \eta_{\textit{qf}}(\vec u,
\vec z, \vec p \varsep y)$, for which we can find an equivalent $\rqfo$ sentence
$\chi$, by the $\rqfo$ fitting hypothesis. Wlog., we can assume that the
variables $\vec z$ and $\vec p$ do not occur in the sentence $\chi$. We replace
all occurrences of $P(\vec w_0,\vec w_1, \vec w_2)$ in $\chi$ (where $\vec w_0$
and $\vec w_1$ have the same cardinality as $\vec u$ and $\vec z$, respectively)
with the conjunction $\adom(\vec w_0) \land \vec w_1 = \vec z \land \vec w_2 =
\vec p$. The proof concludes by showing that the resulting~\rqfo formula is
equivalent to $\eta$.
\end{proof}

\section{$\ptime$ fitting beyond $\rqc$}\label{sec:betweenptimeandrqfoconcept}

\begin{takeaway} 
    We show that \ip structures do not necessarily have \ptime $0$-fitting; yet
    formulas of finite VC dimension over an automatic structure have \ptime
    approximate fitting.
\end{takeaway}

Theorem~\ref{thm:rqfoimpliesptimeapprox} shows that~\rqc (equivalently,~\rqfo
fitting) suffices for \ptime approximate fitting. Recall that the first step of
the proof of this theorem shows that \ip structures cannot be \rqc by showing an
encoding of $3$-colorability. Intuitively, this tells us that for many $\ip$
structures, there are fitting problems that are $\np$-hard (formalizing a
general statement requires further computational assumptions, see
Appendix~\ref{app:ipwitness}). A concrete example of this fact is given
by~\emph{B\"uchi Arithmetic}, the expansion of Presburger arithmetic with the
binary relation $V_2 \coloneqq \{(x,y) \in \N^2 : \text{$y$}$ is the largest
power of $2$ that divides $x\}$. This structure is $\ip$: the partitioned
formula $\textit{bit}(x \varsep y) \coloneqq V_2(x,x) \land \exists y_1, z, y_2
\,\left( y = y_1 + x + y_2 \land y_1 < x < z \land V(y_2,z) \right)$ holds
exactly when $x$ is a power of two appearing in the binary expansion of $y$, and
thus shatters the powers of two. It also has $\ffte$ and effective constants
(all naturals encoded in binary); this follows from the fact it is an automatic
structure (defined below)~\cite{Buechi1960,BruyareHMV94}. By reduction from
$3$-colorability, we see that there are formulas over B\"uchi arithmetic for
which the $0$-fitting problem is $\np$-hard. Moreover, we can easily show that all formulas
over B\"uchi arithmetic have approximate fitting problems in $\np$ (see
Appendix~\ref{app:buchinpmembership}).

\begin{restatable}{proposition}{buchiexactconcept}\label{prop:buchiexactconcept}
  The approximate and (exact fitting) problem of any formula over B\"uchi
  Arithmetic is in $\np$. There are formulas for which the $0$-fitting problem
  is $\np$-complete.
\end{restatable}

\prespower is a reduct of B\"uchi arithmetic with ~\rqc (Example~\ref{ex:rqc}),
and there are a number of other reducts of B\"uchi known to be \rqc from prior
work \cite{modeltheoryofstrings}. Thus, we already know that many
definable concept classes over B\"uchi arithmetic admit \ptime approximate fitting. We now show
that for definable concept classes over B\"uchi arithmetic \emph{finite VC
dimension implies tractable fitting}.

\myparagraph{Automatic structures} 
For a sequence $w_1,\dots,w_k$ of finite words over an alphabet $\Sigma$, their
\emph{convolution} $w_1 \otimes \dots \otimes w_k$ is the word over the alphabet
$(\Sigma \cup \{\Box\})^k$ obtained by first padding each $w_i$ with a fresh
symbol $\Box \not \in \Sigma$ to a common length, to then take the
component-wise product. For instance, for $w_1 = a$ and $w_2 = bcd$, we have
$w_1 \otimes w_2 = 
    \left[\begin{smallmatrix} \vphantom{\Box}a \\ b \end{smallmatrix}\right] {\cdot}
    \left[\begin{smallmatrix} \Box \\ \vphantom{b}c \end{smallmatrix}\right] {\cdot}
    \left[\begin{smallmatrix} \Box \\ d \end{smallmatrix}\right]$, where $\cdot$
denotes concatenation.

A relation $R \subseteq (\Sigma^*)^k$ is \emph{regular} if the language~$\{(w_1
\otimes \dots \otimes w_k) \cdot p : (w_1,\dots,w_k) \in R, {p \in
\{\Box\}^*}\}$ is regular. A relational structure $\amodel$ is \emph{automatic}
if there is a regular language $D \subseteq \Sigma^*$ and a surjection $r \colon
D \to M$ such that, for every relation $R \subseteq M^k$ in the structure, the
relation $\{{(w_1,\dots,w_k) \in D^k} : (r(w_1),\dots,r(w_k)) \in R\}$ is
regular. 
See~\cite{gradel20years} for an overview on automatic structures. For an
automatic structure, we take $M$ as the set of effective constants, with $D$
providing their encodings.

\begin{restatable}{theorem}{thmbuchi}\label{thm:buchi} Consider a partitioned
    formula $\phi(\vec x \varsep \vec y)$ over an automatic structure $\amodel$.
    If $\phi$ has finite VC dimension, then $\phi$ has \ptime approximate
    fitting.
\end{restatable}

\begin{proof}[Proof idea]
    Assuming a single object variable and a single parameter variable, let
    $\mathcal{A}$ be the automaton defining the binary relation given by $\phi(x
    \varsep y)$. Let $S = (s_1,b_1),\dots,(s_n,b_n)$ be a sample with examples
    $s_i$ padded to a common length $\ell$. For a $k \in [0..\ell]$, we say that
    a tuple $(q_1,\dots,q_n)$ of states of $\mathcal{A}$ is \emph{$S$-reachable
    at~$k$} if there is some $y$ and some position such that for each $j \in
    [1..n]$, the automaton $\mathcal{A}$ run on $(s_j,y)$ reaches state $q_j$ at
    position $k$. The algorithm for approximate fitting is a standard
    reachability algorithm that computes the $S$-reachable tuples inductively on
    $j$, and then checks if an $S$-reachable tuple at $\ell$ agrees with $S$ (in
    terms of accepting/rejecting states) on all but at most $\epsilon$ examples.
    As in the proof of Theorem~\ref{thm:rqfoimpliesptimeapprox}, we can use the
    finite VC dimension hypothesis to show a polynomial bound on the number of
    $S$-reachable tuples. See Appendix \ref{app:buchinipptime} for details.
\end{proof}
Note that the converse of Theorem \ref{thm:buchi} fails: there are definable
concept classes in B\"uchi arithmetic with infinite VC dimension but tractable
fitting: $\textit{bit}(x \varsep y)$ can easily be seen to give such an example.

Our results leave open the question of whether \nip and \ffte alone guarantee
\ptime (exact or approximate) fitting, with
Theorems~\ref{thm:rqfoimpliesptimeapprox} and~\ref{thm:buchi} telling us that a
counterexample to this would have to be a structure that is neither automatic
nor with \rqc. One way such a \nip structure might still have \ptime fitting is
if it can be expanded into an~$\rqc$ one (while retaining $\ffte$):

\begin{example}\label{ex:expandablerqc} Consider the last structure from
Example~\ref{ex:rqc}, involving an equivalence relation on the natural numbers
with classes of unbounded finite size. As remarked there, this structure is not
$\rqc$, and hence lacks $\rqfo$ fitting by Theorem \ref{thm:rqcandrqfitting}.
However, it can be expanded to $\prespower$, which is $\rqc$ and thus does have
$\ptime$ fitting.
\end{example}

Whether every \nip structure can be expanded to be \rqc  is open. We note that
the structure from Example~\ref{ex:expandablerqc} can alternatively be expanded
by a total order to get $\rqc$. However, it is known~\cite{addinglinearorders}
that there are \nip structures that cannot be expanded to have a linear order
without losing~$\nip$, thus we cannot hope to be able to add a total order in
general. In Example~\ref{ex:rqc} we recalled that there are some broad
model-theoretic properties that imply $\rqc$, most notably distality
\cite{isolation, simon-chernikov-two}. Unfortunately, it has been shown that
there are \nip structures that cannot be expanded to a distal structure
\cite{nodistalexpansion}; but these counterexamples are still $\rqc$.

\section{Fitting for logically-defined real-valued
functions}\label{sec:functions}
\begin{takeaway} 
    We move our attention to hypothesis classes given by definable real-valued
    functions. If the function is ``piecewise'', we show that the approximate
    fitting problem is tractable over many~$\rqc$ structures. Otherwise, one can
    readily obtain fitting problems that are sum-of-square-root-hard.
\end{takeaway}

We now consider fitting problems for hypothesis classes of real-valued
functions. As mentioned in the introduction, such classes can be built from
logic in several ways. In this section, we take the most direct approach:
definable families of functions over a \emph{numerical structure}, i.e., one
whose domain is a subset of~$\R$. One can represent such a family through a
partitioned formula $\phi(\vec x \varsep z \varsep \vec y)$ with variables split
into object variables~$\vec x$, a single \emph{output variable} $z$, and
parameter variables $\vec y$. The formula $\phi$ is assumed to be functional:
for every parameter~$\vec y^*$ and object~$\vec x^*$, there is a unique~$z^*$
such that $\phi(\vec x^* \varsep z^* \varsep \vec y^*)$. We write $f_\phi$ for
the function defined by $\phi$, i.e., the function satisfying $f_\phi(\vec x,
\vec y) = z \leftrightarrow \phi(\vec x \varsep z \varsep \vec y)$. The
\emph{definable function class}~$\hypoclass_{f_\phi}$ given by $f_\phi$ is the
hypothesis class containing all functions $f_{\vec p} \coloneqq (\vec x \mapsto
f_{\phi}(\vec x, \vec p))$, where $\vec p$ ranges over \mbox{the parameters.}

\begin{remark}\label{rem:exactfunction} It is easy to see that the $0$-fitting
problem for definable function classes trivially reduces to the concept class
case: given an input sample $S \coloneqq (\vec s_1,t_1),\dots,(\vec s_n,t_n)$,
we have
\[
    (\exists \vec y\, \sum\nolimits_{i=1}^n|f_\phi(\vec s_i, \vec y) - t_i| \leq 0) 
    \leftrightarrow 
    (\exists \vec y\, \bigwedge\nolimits_{i=1}^n \phi(\vec s_i \varsep t_i \varsep \vec y)),
\]
and so we can alternatively solve the $0$-fitting problem for $\phi(\vec x , z
\varsep \vec y)$, treating $S$ as the set of positive examples. In particular,
by Theorem~\ref{thm:rqfoimpliesptimeapprox} this implies that for every $\rqc$
numerical structure, like the real ordered field~$(\R,+,\cdot,<)$ or Presburger
Arithmetic, definable function classes have $\ptime$ $0$-fitting. In general, we
will not obtain an analogous result for approximate fitting. 
\end{remark}

\subsection{Approximate fitting for piecewise functions}\label{subsec:approximate-piecewise}
Consider a formula $\phi(\vec x \varsep z \varsep \vec y )$. 
We say that $f_\phi$ (and $\phi$) are \emph{piecewise}
whenever $\phi$ is of the form
\begin{equation}\label{eq:piecewise}
    \bigvee\nolimits_{i=1}^m \psi_i(\vec x, \vec y) \land z = \ell_i(\vec x, \vec y),
\end{equation}
where the \emph{guards} $\psi_1,\dots,\psi_m$ form a partition of the range and parameter spaces,
and each $\ell_i$ is a term. In this case, we say that $\hypoclass_{f_\phi}$ is
a \emph{piecewise  function class}.
We show that, for piecewise function classes, the landscape of approximate
fitting problems reduces to satisfiability in fixed dimension.

\begin{restatable}{lemma}{piecewisefunctionreduction}\label{lemma:piecewise-functions-reduction}
Let $\amodel$ be a numerical structure with domain $M$ containing the naturals,
and with signature $L$ containing $0$, $1$, $+$, and $<$, interpreted
as usual. Assume $\amodel$ is~\nip, and all its partitioned formulas have \ptime
$0$-fitting. Let $f_\phi$ be a piecewise function given by a 
formula~$\phi(\vec x \varsep z \varsep \vec y)$ over~$\amodel$. 
There is a polynomial-time algorithm that given in input a sample $S$ and 
a tolerance~$\epsilon$, computes $\phi'(\vec y)$ over $\amodel$, with the property
that $\hypoclass_{f_\phi}$ \mbox{$\epsilon$-fits} $S$ if and only if $\amodel
\models \exists \vec y\,\phi'(\vec y)$. If $\phi$ is quantifier-free,
so~is~$\phi'$.
\end{restatable}

\begin{proof}[Proof sketch]
    Let $\phi$ be as in Formula~\eqref{eq:piecewise}, and let $S = (\vec s_1,
    t_1), \dots,(\vec s_n,t_n)$ be the sample. Writing down the definition of
    the approximate fitting problem, we see that it is equivalent to deciding 
    \[  
        \exists \vec y : \bigvee\nolimits_{\mu \colon [1..n] \to [1..m]} 
            \textstyle\left(\bigwedge_{j=1}^n \psi_{\mu(j)}(\vec s_j,\vec y) 
            \, \land \, \sum_{j=1}^n \abs{\ell_{\mu(j)}(\vec s_j, \vec y) - t_j} \leq \epsilon \right).
    \]
    By applying~\nip (more precisely, Fact~\ref{fact:sslemma}), we show that
    $\gamma_\mu \coloneqq \bigwedge_{j=1}^n \psi_{\mu(j)}(\vec s_j,\vec y)$ is
    satisfiable only for polynomially many maps $\mu$. Similarly,~\nip implies
    that only polynomially many combinations of signs for the terms inside
    absolute values are realizable. By the \ptime 0-fitting hypothesis, we can
    compute the maps with satisfiable $\gamma_\mu$ and the realizable sign
    combinations in polynomial time. Therefore, the algorithm from the statement
    of the lemma computes a formula~$\phi'$ equivalent to the one above, but
    with only polynomially many disjuncts and no absolute values. See Appendix \ref{app:lemmapiecewisefunctionreduction}.
\end{proof}

Thanks to~Lemma~\ref{lemma:piecewise-functions-reduction}, we obtain the following:

\begin{restatable}{theorem}{ThmApproxFittingPiecewiseRealAndPa}\label{thm:approx-fitting-piecewise-reals-and-PA} 
    The approximate
    fitting problem is in~\ptime for every definable function class over the
    real ordered group or over Presburger arithmetic. For the real ordered
    field, the problem is in~\ptime for every piecewise function~class.
\end{restatable}

\begin{proof}[Proof sketch]
    First, we note that every definable function over $(\R,+,<)$ is a piecewise
    function~$f_\phi$ given by a quantifier-free formula $\phi$ from the
    expanded structure $(\R,+,-,0,1,\{\frac{\cdot}{d}\}_{d \geq 2}, <)$, where
    $\frac{\cdot}{d}$ denoted division by $d$. A similar result holds for
    Presburger arithmetic, with respect to the expanded structure
    $(\N,+,-,0,1,\{\frac{\cdot}{d}\}_{d \geq 2})$ where $\frac{\cdot}{d}$ is
    interpreted as \emph{integer} division by $d$ (i.e., $\frac{n}{d} = r$ if
    and only if $r \cdot d \leq n < (r+1) \cdot d$). Both results ultimately
    follow from quantifier elimination in these linear structures. Moreover,
    both structures inherit~\rqc and~\ffte. Thus, 
    we can restrict our attention to piecewise functions also in the first 
    statement of the theorem.

    From~\rqc we conclude that all considered structures are~\nip and have
    \ptime $0$-fitting (Theorem~\ref{thm:rqfoimpliesptimeapprox}). Moreover, we
    can assume the formula $\phi'$ from
    Lemma~\ref{lemma:piecewise-functions-reduction} to be quantifier-free
    (recall from Example \ref{ex:standardstructures}: the real ordered field has quantifier elimination). The
    sentence $\exists \vec y\, \phi'(\vec y)$ in that lemma has size polynomial in
    the sample and tolerance, and has a fixed number of (existential)
    quantifiers. The result follows because, for both Presburger arithmetic
    and the real ordered field, one can solve existential formulas in $\ptime$
    when the number of variables is fixed~\cite{Renegar92,Lenstra83}. See Appendix \ref{app:approximatefittingpiecewise}.
\end{proof}

\begin{remark}[Connection to training problems for Neural networks]
    \label{remark:neural} Feedforward neural networks with ReLU activation
    functions and a single output neuron can be encoded as piecewise functions
    over the real ordered field. By
    Theorem~\ref{thm:approx-fitting-piecewise-reals-and-PA}, the approximate
    fitting problem for any fixed network of this kind is in~\ptime. This result
    also applies to the $L^2$ norm and the hinge loss function
    (Remark~\ref{remark:otherlosses}). Related work~\cite{AroraBMM18,FroeseH23}
    shows that approximate fitting for networks with two layers and a single
    output is in~\ptime when the network is fixed, and \np-complete otherwise.
    For networks with multiple outputs, the $0$-fitting problem is already
    $\exists\R$-complete when the network is not
    fixed~\cite{trainingiscomplete}. Here, $\exists\R$ represents the complexity
    of the existential theory of the reals, see Definition
    \ref{def:existentialtheoryreals} below.
\end{remark}

\subsection{Approximate fitting outside piecewise functions} 
The real ordered field has definable functions
that are not piecewise. An example is the parameter-free square root function
given by formula~$\textit{sqrt}(x \varsep z) \coloneqq ({x < 0} \land {z = 0})
\lor (z \geq 0 \land z^2 = x)$. This function violates the piecewise property 
because the second disjunct contains a guard involving
the output $(z \geq 0)$, and an equality non-linear in $z$ ($z^2 = x$ instead of
$z = \ell(x)$). Showing \ptime approximate fitting for \textit{sqrt} would
constitute a breakthrough in computational geometry, as it would yield a
$\ptime$ algorithm for the \emph{sum-of-square-root problem (SSR)}. Given
non-negative integers $a_1,\dots,a_n$ encoded in binary, and a rational~$b$, 
this problem asks if $\sum_{i=1}^n \sqrt{a_i} \leq b$. The problem is not
known to be \ptime-hard, but the best known upper bound places it in the
counting hierarchy~\cite{AllenderBKM09}. The reduction to approximate fitting
for \textit{sqrt} is straightforward: simply supply the sample
$(a_1,0),\dots,(a_n,0)$ and the tolerance~$b$ to this problem.
Proposition~\ref{prop:fitting-real-field-functions} (stated below, proved in
Appendix~\ref{app:sqrtsum}) shows that, in fact, the tolerance can be fixed to
any~$\epsilon > 0$.

While we conjecture the approximate fitting problem for the real ordered field
to be in the counting hierarchy, at the moment the best known upper bound is
given by the straightforward reduction to the existential theory of~the~reals.
\begin{definition}[$\exists\R$]
\label{def:existentialtheoryreals} The existential theory of the reals
$\exists\R$ is the closure of the satisfiability problem for existential
sentences over the real ordered field under polynomial-time many-one reductions.
\end{definition}
The class $\exists\R$ is contained in~$\pspace$~\cite{Renegar92}, but strictness
of this containment  is open \cite{existentialtheory}.

\begin{restatable}{proposition}{PropFittingRealFieldFunctions}\label{prop:fitting-real-field-functions}
    The approximate fitting problem of any definable function class over the
    real ordered field is in~$\exists\R$.
    There are classes for which the $\epsilon$-fitting problem is SSR-hard,
    for every fixed~$\epsilon > 0$.
\end{restatable}

\section{Fitting for randomized classes}\label{sec:randomized}

\begin{takeaway} 
    We look at hypothesis classes formed by ``randomizing'' a base concept
    class, considering distributions over the parameter space. If a
    concept class is~\rqc, its randomized class has $\ptime$ 
    approximate fitting; the algorithm combines $\rqfo$ querying with linear
    algebra. 
\end{takeaway}
We now turn to a second natural approach to creating a real-valued hypothesis class from logic, which works on top of an arbitrary structure.
The idea, which was independently investigated in model theory~\cite{keislerrandomizing, itaykeisler} and in data management~\cite{sigmod22}, is to ``randomize a hypothesis class''. 
We use the variant from~\cite{sigmod22,arxivwithaaron}: starting from a ``base'' hypothesis class we form a new real-valued hypothesis class by looking at probability distributions.

\myparagraph{Randomizing a concept class}
Let $\Delta(A)$ denote  all probability distributions on a 
$\sigma$-algebra~$A$.

\begin{definition}[Distribution class]\label{def:distclass} 
    Consider a concept class $\conceptclass \coloneqq \{f_y : y \in \paramspace\}$ with range space~$\rangespace$, and let $Y_\Sigma$ 
    be the $\sigma$-algebra on $Y$ generated by the sets $P_x \coloneqq \{y \in Y : f_y(x) = 1\}$ for all $x \in X$.
    The \emph{distribution class} 
    of $\conceptclass$ is the (real-valued) hypothesis
    class $\measureclass{\conceptclass} \coloneqq \{h_\mu : \mu \in \Delta(Y_\Sigma)\}$ on the same range space $\rangespace$, 
    such that each $h_\mu$ maps $x \in X$ 
    to~$\mu(P_x)$. 
\end{definition}
\begin{example}\label{ex:randomrectangles}
    We return to Example~\ref{ex:randomrectangle}.
    Let $\conceptclass$ be the dual of the concept class of rectangles
    from~Example~\ref{ex:rectangleintro}: the range space is $\R^4$ (all
    rectangles) and the parameter space is $\R^2$ (points in the plane). Its
    distribution class is parameterized by distributions $\mu$ over a $\sigma$-algebra on $\R^2$. 
    Given such a probability distribution~$\mu$,
    the hypothesis $h_\mu$ maps a rectangle to its $\mu$-probability.

    In the approximate fitting problem for $\measureclass{\conceptclass}$, samples are sequences $(\vec s_1, t_1), \dots, (\vec s_n, t_n)$, 
    where each~$\vec s_i$ is a (representation of a) rectangle, and $t_i$ is a target probability. Given a tolerance $\epsilon \geq 0$, 
    we seek a distribution $\mu$ satisfying $\sum_{i=1}^n |h_\mu(\vec s_i) - t_i| \leq \epsilon$.
\end{example}

We remark that~\cite{sigmod22} defines the \emph{dual distribution class} of a
concept class by taking distributions~$\mu$ on the range space~$\rangespace$, where
each hypothesis $h_\mu$ maps a parameter $y \in \paramspace$ to~$\mu({\{x
\in X : f_y(x) = 1\}})$. For a concept class given by a partitioned
formula $\phi(\vec x \varsep \vec y)$, the dual distribution class coincides
with the distribution class of the formula obtained from $\phi$ by
swapping object and parameter variables. 
All our results on distribution classes carry
over to dual distribution~classes.

\subsection{Fitting problems for distribution classes}\label{subsec:fittingdistributionclasses}

We show that when the $0$-fitting problem for the underlying concept class is decidable, so is the fitting problem for its distribution class, via a reduction from the latter to the former.

\begin{restatable}{proposition}{ProprEasyDecidabilityDistribution}\label{prop:easydecidabilitydistribution}
Let $\conceptclass$ be a concept class given by a partitioned formula $\phi(\vec x \varsep \vec y)$. 
The approximate fitting problem for the distribution class $\measureclass{\conceptclass}$ 
reduces in non-deterministic polynomial time to deciding polynomially many instances of the $0$-fitting problem for $\conceptclass$. 
The instance of the fitting problem for~$\measureclass{\conceptclass}$ is a yes-instance if and only if all the $0$-fitting problems are yes-instances.
\end{restatable}

\begin{proof}
Let 
$(\vec s_1, t_1), \dots, (\vec s_n, t_n)$ be the input sample, and $\epsilon$ be the
tolerance. The approximate fitting problem for $\measureclass{\conceptclass}$
asks for a distribution $\mu$ such that $\sum_{i=1}^n |\mu(\{\vec y
\in Y : \phi(\vec s_n,\vec y)\}) - t_i| \leq \epsilon$. 
All that is relevant for specifying $\mu$ is  its output on
the finite Boolean algebra generated by the sets 
${\{\vec p \in Y : \phi(\vec s_i, \vec p)\}}$ 
for $i \in [1..n]$, so it suffices to know the probability of 
formulas $\bigwedge_{i=1}^n \psi_{i}(\vec s_i, \vec y)$ 
where each $\psi_{i}$ is either $\phi$ or its negation. 
There are at most $2^n$ of these formulas. Let $F$ be the set of these formulas that are satisfiable. We write $F_i$ for the formulas in $F$ in which $\phi(\vec s_i, \vec y)$ occurs positively.
Introducing a variable $x_\psi$ for every formula $\psi \in F$,
the approximate fitting problem reduces to solving the following problem over the real ordered group: 
\begin{equation}\label{eq:systemdistribution}
    \sum\nolimits_{i=1}^n |t_i - {\textstyle\sum_{\psi \in F_i}} x_\psi| \leq \epsilon 
    \land \sum\nolimits_{\psi \in F} x_\psi = 1
    \land \bigwedge\nolimits_{\psi \in F} x_\psi \geq 0.
\end{equation}
The main
observation is that Formula~\eqref{eq:systemdistribution} is severely
underconstrained, and so by appealing to Carath\'eodory's theorem (for convex cones), we are
able to conclude that if there is a solution, then there is one where at most
polynomially many variables $x_\psi$ are non-zero (see Appendix~\ref{app:npboundforpwpdistributionclasses} for details).
Note that the constraint featuring absolute 
values is unproblematic when it comes to applying Carath\'eodory's theorem:
it can be replaced with 
\[
    \sum\nolimits_{i=1}^n z_i \leq \epsilon \land \bigwedge\nolimits_{i=1}^n (y_i = t_i - \sum\nolimits_{\psi \in F_i} x_\psi \land y_i \leq z_i \land -y_i \leq z_i),
\]
where the variables $y_i,z_i$ are fresh.

The non-deterministic polynomial-time reduction 
simply guesses polynomially many formulas of the form
$\bigwedge_{i=1}^n \psi_{i}(\vec s_i, \vec y)$ described above,
and solves in polynomial time the system obtained from Formula~\eqref{eq:systemdistribution} by keeping only the variables corresponding to the guessed formulas.
Each formula $\bigwedge_{i=1}^n \psi_{i}(\vec s_i, \vec y)$ corresponds to an instance 
of the $0$-fitting problem for $\conceptclass$,
and the fitting problem for $\measureclass{\conceptclass}$ is a yes-instance if and only if the all the $0$-fitting problems are yes-instances.
\end{proof}

It is easy to see that \emph{the $0$-fitting problem for the distribution class over
a concept class $\conceptclass$ is at least as hard to fit as $\conceptclass$
itself}: for samples where the associated target probabilities  are in $\{0,1\}$,
if they can be fit by a randomized parameter, they can also be fit by a
deterministic one (see Appendix~\ref{app:easy-hardness-distributions}). Then,
the corollary below follows from
Proposition ~\ref{prop:buchiexactconcept}  
and Proposition~\ref{prop:easydecidabilitydistribution}.

\begin{corollary}\label{cor:distclassconceptbuchi} 
    Approximate fitting for distribution classes of formulas in B\"uchi arithmetic 
    is in $\np$, and there is a distribution class over  B\"uchi arithmetic with an \np-complete $0$-fitting problem.
\end{corollary}

Echoing Lemma~\ref{lem:FiniteVCandPTimeFittingImpliesPtimeApprox}, 
if we assume that the concept class has finite VC dimension and \ptime $0$-fitting, 
then its distribution class has \ptime approximate fitting.

\begin{theorem}\label{thm:finitevcandptimeimpliesdistptime} 
    Let $\phi(\vec x \varsep \vec y)$ be a partitioned formula with finite VC dimension and $\ptime$ $0$-fitting. Then, the distribution class $\measureclass{\conceptclass_\phi}$ has $\ptime$ %
    approximate fitting.
\end{theorem}

\begin{proof}
    We use the same reduction outlined in the proof of Proposition~\ref{prop:easydecidabilitydistribution}. 
    Finite VC dimension (more precisely,~Fact~\ref{fact:sslemma}) 
    and $\ptime$ $0$-fitting ensure that the set $F$ in that proof can be constructed in {\ptime}. Then Formula~\ref{eq:systemdistribution} 
    involves only a polynomial number of variables. After eliminating the absolute values, as detailed  
    in the proof of Proposition~\ref{prop:easydecidabilitydistribution}, 
    the system can be solved in polynomial time using the ellipsoid method~\cite{linearprogrammingbook}.
\end{proof}

Theorem~\ref{thm:finitevcandptimeimpliesdistptime} allows us to lift the results in Thm.~\ref{thm:rqfoimpliesptimeapprox}, Cor.~\ref{cor:rqcimpliesrqfoandptimefitting} 
and Thm.~\ref{thm:buchi} to distribution classes:

\begin{corollary}\label{cor:rqcmodelimpliesdistclassconceptptimefitting} 
    If a structure~$\amodel$ is $\rqc$ then for any partitioned formula $\phi$, the  approximate fitting problem for the distribution class~$\measureclass{\conceptclass_\phi}$ is in $\ptime$. The same result holds if ~ $\amodel$ is an automatic structure and $\phi$ has finite VC dimension.
\end{corollary}

\subsection{Randomizing a definable function class}\label{subsec:expectationclass}
Thus far, we have considered two ways of creating logic-based real-valued hypothesis classes: via formulas over a numerical structure that define functions, and by randomizing an arbitrary logically-defined hypothesis class. It is possible to combine these approaches. Starting from a definable function class over a numerical structure such as the reals, we ``randomize the function class'' to form a new real-valued class parameterized by distributions over the parameter space. Each distribution induces a hypothesis by taking the expected value of the underlying function over the parameters. We formalize this construction below.

A real-valued hypothesis class $\hypoclass \coloneqq \{f_y : y \in \paramspace\}$ with range space $\rangespace$ is said to be \emph{bounded} 
whenever, for every $x \in \rangespace$, the set $\{f_y(x) : y \in Y\}$ is included in an interval $[a,b]$ of $\R$.

\begin{restatable}[Expectation class]{definition}{DefExpectationClass}\label{def:expectationclass}
    Consider a bounded hypothesis class $\hypoclass \coloneqq \{f_y : y \in \paramspace\}$ with range space~$\rangespace$, and the $\sigma$-algebra~$Y_\Sigma$ generated by the sets ${\{y \in Y : f_y(x)  \leq r\}}$ for all $x \in \rangespace$ and $r \in \R$.
    The \emph{expectation class} 
    of $\hypoclass$ is the (real-valued) hypothesis
    class $\expectationclass{\hypoclass} \coloneqq \{h_\mu : \mu \in \Delta(Y_\Sigma)\}$ on the same range space $\rangespace$, 
    such that each $h_\mu$ maps $x \in \rangespace$ 
    to the expectation $\expectation_{y \sim \mu}[f_y(x)] \coloneqq \int_{Y} f_y(x)\,d\mu(y)$. 
\end{restatable}

Note that the hypothesis that $\hypoclass$ is bounded is both a necessary and sufficient condition for the expectation $\expectation_{y \sim \mu}[f_y(x)]$
to be well-defined (i.e., for every $x \in \rangespace$, the map~$y \mapsto f_y(x)$ is $\mu$-integrable) on all probability distributions $\mu \in \Delta(Y_\Sigma)$.

\myeat{
\begin{example}\label{ex:defunctiondist}
Let $\hypoclass = \{f_y(x) : y \in [0,1]\}$ be the bounded hypothesis class with ${f_y(x) \coloneqq |x-y|}$, having range space $\R$.
Hypotheses $h_\mu$ in the expectation class $\expectationclass{\hypoclass}$ 
are parametrized by distributions~$\mu$ on $[0,1]$, 
and map $x \in \R$ 
to the expectation ${\expectation_{y \sim \mu}[|x-y|]}$.
In general, $h_\mu$ is not linear.
For example, when $\mu$ is the uniform distribution, resolving the integral yields:
\[
    h_\mu(x) = \textstyle\textbf{if } x \not \in [0,1] \textbf{ then } |{x-\frac{1}{2}}| \textbf{ else } x^2 - x + \frac{1}{2}.
\]
The fitting problem for this class would take as input a finite set of samples consisting of 
pairs of a rational $s_i$ and a target value for the mean over $y$ of $f_y(s_i)$. 
\end{example}
}

We can show that fitting problems are decidable for expectation classes that are built over definable function classes over classical numerical structures:
\begin{restatable}{theorem}{approxfittingexpect}%
    \label{thm:approx-fitting-piecewise-reals-and-PA-expect}%
    \label{thm:distclassfunctionrcfpspace}%
    \label{thm:expectationPresburger}%
    \label{thm:expectationrealgroup}%
    The approximate fitting problem is in~\np
    for every expectation class over a bounded function class over the real ordered group 
    and also for Presburger arithmetic. 
    For expectation classes based on the real ordered field, the problem is in~$\exists\R$.
\end{restatable}

The proof of this theorem is given in Appendix \ref{app:expectation}. 
Similarly to the proof of~Proposition~\ref{prop:easydecidabilitydistribution}, 
this proof mixes techniques from linear algebra and convex geometry with our algorithms 
for fitting problems for definable function classes.

\section{Discussion}\label{sec:conc}

\begin{table}
    \scalebox{0.95}{%
    \centering
        \begin{tabular}{|l || l | l |}
            \hline 
            \rowcolor{gray!20} {\bf Concept Class} & {\bf Exact fitting}& {\bf Approximate fitting}\\
            \hline
            \rqc structure (Presburger, real field, $\ldots$)  %
                & $\rqfo$ (Corollary~\ref{cor:rqcimpliesrqfoandptimefitting})  & {\bf PTIME} (Corollary~\ref{cor:rqcimpliesrqfoandptimefitting})\\
            B\"uchi Arithmetic, finite VC dimension  &
            {\bf PTIME} (Theorem~\ref{thm:buchi}) & {\bf PTIME} (Theorem~\ref{thm:buchi})\\
            B\"uchi Arithmetic, general &
            {\bf NP}, tight  (Proposition~\ref{prop:buchiexactconcept}) & {\bf NP}, tight (Proposition~\ref{prop:buchiexactconcept})\\[1pt]
            \hline 
            \rowcolor{gray!20} {\bf Definable Function Class}&&\\
            \hline
            over $(\R,+,<)$ or $(\N,+,<)$
                & $\rqfo$ (Remark~\ref{rem:exactfunction})  & {\bf PTIME}   (Theorem~\ref{thm:approx-fitting-piecewise-reals-and-PA})\\
            Real ordered field, piecewise 
                & $\rqfo$ (Remark~\ref{rem:exactfunction}) & {\bf PTIME}   (Theorem~\ref{thm:approx-fitting-piecewise-reals-and-PA})\\
            Real ordered field, general
                &  $\rqfo$  (Remark~\ref{rem:exactfunction}) & {\bf $\exists\R$}, SSR-hard (Prop.~\ref{prop:fitting-real-field-functions})\\[1pt]
            \hline 
            \rowcolor{gray!20} {\bf Distribution Class}&&\\
            \hline
            \rqc structure  &
            \multicolumn{2}{c|}{{\bf PTIME} (Corollary~\ref{cor:rqcmodelimpliesdistclassconceptptimefitting})}\\
            B\"uchi arithmetic, finite VC dimension 
            & \multicolumn{2}{c|}{{\bf PTIME} (Corollary~\ref{cor:rqcmodelimpliesdistclassconceptptimefitting})}\\
            B\"uchi arithmetic, general 
            & \multicolumn{2}{c|}{{\bf NP}, tight (Corollary~\ref{cor:distclassconceptbuchi})}\\[1pt]
            \hline
            \rowcolor{gray!20} {\bf Expectation Class}&&\\
            \hline
            over $(\R,+,<)$ or $(\N,+,<)$
                & \multicolumn{2}{c|}{{\bf NP} (Theorem~\ref{thm:distclassfunctionrcfpspace})}\\
            Real ordered field 
                & \multicolumn{2}{c|}{{\bf $\exists\R$} (Theorem~\ref{thm:distclassfunctionrcfpspace})}\\
            \hline
        \end{tabular}
    }
    \medskip
    \caption{Logic-based fitting problems (assuming effective constants and $\ffte$): a recap.}
    \label{tab:summary}
    \vspace{-0.6cm}
\end{table}

Our work initiates a study of fitting problems for logic-based classes over infinite structures. Table~\ref{tab:summary} overviews some of~our~results. We highlight that,
for several natural examples, exact fitting can be performed with a first-order query over the active domain of the sample---$\rqfo$ entries in the table---while approximate fitting reduces to generating polynomial many such queries. For real-valued classes, we show that fitting is decidable for many natural classes;  the algorithms often involve equation solving in addition to querying, and tractability remains open in some cases. 

For general hypothesis classes, it is known that ``well-behaved sample
complexity'' (characterized by finite VC dimension) and ``well-behaved
computational complexity of fitting'' are unrelated
(Remark~\ref{remark:VC-only-combinatorial}). However, our exploration of fitting
problems reveals that these two are linked when focusing
on~\emph{logically-defined classes}. Indeed, for these classes we have found no
counterexample to the question ``does every $\nip$ structure with $\ffte$ have
$\ptime$ fitting?'', and we proved that any such counterexample would need to
come from a structure that is neither \rqc nor automatic. This is perhaps the
major question we leave open.  For an affirmative answer, it would suffice to
show that every such structure can be expanded to have $\rqc$ and $\ffte$. 

In Table~\ref{tab:summary}, all results involving the existential theory of the
reals are \emph{not} known to be tight. These problems are syntactically
restricted instances of the existential theory of the reals. The reduction from
the sum-of-square-root problem to fitting for definable functions in the real field
indicates that one should not hope to prove a bound better for this problem than a fixed level of
the counting hierarchy \cite{AllenderBKM09}. 
The $\np$ upper bounds for expectation classes over the real ordered
group and Presburger arithmetic from Section~\ref{sec:randomized} are also not
known to be tight.

All of our complexity results use standard bit complexity, assuming effective constants for the samples. It would be possible to restate some results using the abstract Blum--Shub--Smale model of computation \cite{bss1, poizat1995petits}, but in many cases this would weaken the statements (e.g. for NP bounds).

In this work we have only briefly explored the connection with fitting  problems for neural networks, as in \cite{trainingfully, jonniexptraining,trainingiscomplete, FroeseH23, AroraBMM18}. To strengthen the applicability of our techniques to practical settings, an important future direction is to examine approximate fitting for other loss functions. All our upper bounds
extend to the $L^2$ norm and hinge loss function, see Appendix~\ref{app:ltwoloss}.
Also note that we have only dealt with \emph{decision problems} here, leaving open the problem of finding a witness parameter in the case of a positive answer. The search problem would require more restrictions on the computation model,  specifying the representation for parameters.

Even for the simplest case of definable concept classes, we have not provided a detailed study of when fitting can be done in tractable query languages beyond $\rqfo$. \forarxiv{We do know that there are structures without $\rqfo$ fitting, but where fitting for all definable concept classes can be done in another $\ptime$ logic: see Appendix \ref{app:lfp} for details.}

\bibliographystyle{abbrv}
\bibliography{references}

\begin{thebibliography}{10}

\bibitem{trainingiscomplete}
M.~Abrahamsen, L.~Kleist, and T.~Miltzow.
\newblock {Training Neural Networks is ER-complete}.
\newblock In {\em NeurIPS}, 2021.

\bibitem{AllenderBKM09}
E.~Allender, P.~B{\"{u}}rgisser, J.~Kjeldgaard{-}Pedersen, and P.~B. Miltersen.
\newblock On the complexity of numerical analysis.
\newblock {\em {SIAM} J. Comput.}, 2009.

\bibitem{arxivwithaaron}
A.~Anderson and M.~Benedikt.
\newblock From learnable objects to learnable random objects, 2025.
\newblock https://arxiv.org/abs/2504.00847.

\bibitem{learninghardnessanthony}
M.~Anthony.
\newblock Some connections between learning and optimization.
\newblock {\em Discrete Applied Mathematics}, 144(1):17--26, 2004.

\bibitem{AroraBMM18}
R.~Arora, A.~Basu, P.~Mianjy, and A.~Mukherjee.
\newblock {Understanding deep neural networks with rectified linear units}.
\newblock In {\em {ICLR}}, 2018.

\bibitem{existentialtheory}
S.~Basu, R.~Pollack, and M.-F. Roy.
\newblock {\em Existential Theory of the Reals}, pages 465--492.
\newblock {Springer}, 2003.

\bibitem{isolation}
O.~V. Belegradek, A.~P. Stolboushkin, and M.~A. Taitslin.
\newblock Extended order-generic queries.
\newblock {\em Annals of Pure and Applied Logic}, 97(1):85--125, 1999.

\bibitem{ibyrandvar}
I.~{Ben Yaacov}.
\newblock On theories of random variables.
\newblock {\em Israel Journal of Mathematics}, 194:957--1012, 2013.

\bibitem{itaykeisler}
I.~{Ben Yaacov} and H.~J. Keisler.
\newblock Randomizations of models as metric structures.
\newblock {\em Confluentes Mathematici}, 1(2):197--223, 2009.

\bibitem{generalizing}
M.~Benedikt.
\newblock Generalizing finite model theory.
\newblock In V.~Stoltenberg-Hansen and J.~Väänänen, editors, {\em Logic Colloquium ’03}, page 3–24. Cambridge University Press, 2006.

\bibitem{mbehlics}
M.~Benedikt and E.~Hrushovski.
\newblock Embedded finite models beyond restricted quantifier collapse.
\newblock In {\em LICS}, 2023.

\bibitem{hybridso}
M.~Benedikt and L.~Libkin.
\newblock Relational queries over interpreted structures.
\newblock {\em J. ACM}, 47(4):644–680, July 2000.

\bibitem{modeltheoryofstrings}
M.~Benedikt, L.~Libkin, T.~Schwentick, and L.~Segoufin.
\newblock Definable relations and first-order query languages over strings.
\newblock {\em J. {ACM}}, 50(5):694--751, 2003.

\bibitem{trainingfully}
D.~Bertschinger, C.~Hertrich, P.~Jungeblut, T.~Miltzow, and S.~Weber.
\newblock {Training Fully Connected Neural Networks is $\exists \reals$ -Complete}.
\newblock In {\em NeurIPS}, 2023.

\bibitem{bss1}
L.~Blum, M.~Shub, and S.~Smale.
\newblock On a theory of computation and complexity over the real numbers: {NP}-completeness, recursive functions and universal machines.
\newblock {\em Bulletin of the American Mathematical Society}, 21(1):1--46, 1989.

\bibitem{vctheorem}
A.~Blumer, A.~Ehrenfeucht, D.~Haussler, and M.~K. Warmuth.
\newblock Learnability and the {Vapnik-Chervonenkis} dimension.
\newblock {\em J. ACM}, 36(4):929–965, 1989.

\bibitem{BruyareHMV94}
V.~Bruy{\`e}re, G.~Hansel, C.~Michaux, and R.~Villemaire.
\newblock Logic and {$p$}-recognizable sets of integers.
\newblock {\em Bull. Belg. Math. Soc. Simon Stevin}, 1(2):191--238, 1994.

\bibitem{Buechi1960}
R.~J. Büchi.
\newblock Weak second-order arithmetic and finite automata.
\newblock {\em Math. Logic Quart.}, 6(1-6):66--92, 1960.

\bibitem{changkeisler}
{C. C. Chang and H. J. Keisler}.
\newblock {\em Model theory}.
\newblock North-Holland, third edition, 1990.

\bibitem{caratheodory1907}
C.~Carath{\'e}odory.
\newblock {\"U}ber den variabilit{\"a}tsbereich der koeffizienten von potenzreihen, die gegebene werte nicht annehmen.
\newblock {\em Mathematische Annalen}, 64(1), 1907.

\bibitem{simon-chernikov-two}
A.~Chernikov and P.~Simon.
\newblock {Externally definable sets and dependent pairs II}.
\newblock {\em { Transactions of the American Mathematical Society}}, 367(7):5217--5235, 2015.

\bibitem{faginzeroone}
R.~Fagin.
\newblock Probabilities on finite models.
\newblock {\em Journal of Symbolic Logic}, 41(1):50--58, 1976.

\bibitem{anthonydiego}
D.~Figueira, A.~Jez, and A.~W. Lin.
\newblock Data path queries over embedded graph databases.
\newblock In {\em {PODS}}, 2022.

\bibitem{pseudofinite}
J.~Flum and M.~Ziegler.
\newblock Pseudo-finite homogeneïty and saturation.
\newblock {\em The Journal of Symbolic Logic}, 64(4):1689--1699, 1999.

\bibitem{FroeseH23}
V.~Froese and C.~Hertrich.
\newblock Training neural networks is {NP}-hard in fixed dimension.
\newblock In {\em NeurIPS}, 2023.

\bibitem{FurstSS84}
M.~L. Furst, J.~B. Saxe, and M.~Sipser.
\newblock Parity, circuits, and the polynomial-time hierarchy.
\newblock {\em Math. Syst. Theory}, 17(1):13--27, 1984.

\bibitem{tighthardness}
S.~Goel, A.~R. Klivans, P.~Manurangsi, and D.~Reichman.
\newblock {Tight Hardness Results for Training Depth-2 ReLU Networks}.
\newblock In {\em {ITCS}}, 2021.

\bibitem{gradel20years}
E.~Gr{\"{a}}del.
\newblock Automatic structures: Twenty years later.
\newblock In {\em {LICS}}, 2020.

\bibitem{qerandomgraph}
E.~Grandjean.
\newblock Complexity of the first-order theory of almost all finite structures.
\newblock {\em Information and Control}, 57(2):180--204, 1983.

\bibitem{groheritzert}
M.~Grohe and M.~Ritzert.
\newblock Learning first-order definable concepts over structures of small degree.
\newblock In {\em {LICS}}, 2017.

\bibitem{GLS1988}
M.~Gr{\"{o}}tschel, L.~Lov{\'{a}}sz, and A.~Schrijver.
\newblock {\em Geometric Algorithms and Combinatorial Optimization}, volume~2 of {\em Algorithms and Combinatorics}.
\newblock Springer, 1988.

\bibitem{jonniexptraining}
T.~Hankala, M.~Hannula, J.~Kontinen, and J.~Virtema.
\newblock Complexity of neural network training and {ETR:} extensions with effectively continuous functions.
\newblock In {\em {AAAI}}, 2023.

\bibitem{nodistalexpansion}
P.~Hieronymi, T.~Nell, and E.~Walsberg.
\newblock Wild theories with o-minimal open core.
\newblock {\em Annals of Pure and Applied Logic}, 169(2):146--163, 2018.

\bibitem{sigmod22}
X.~Hu, Y.~Liu, H.~Xiu, P.~K. Agarwal, D.~Panigrahi, S.~Roy, and J.~Yang.
\newblock Selectivity functions of range queries are learnable.
\newblock In {\em {SIGMOD}}, 2022.

\bibitem{immermanvardi}
N.~Immerman.
\newblock {Relational queries computable in polynomial time}.
\newblock {\em {Information and Control}}, 68:86–104, 1986.

\bibitem{meanvalue}
S.~Janković and M.~Merkle.
\newblock A mean value theorem for systems of integrals.
\newblock {\em Journal of Mathematical Analysis and Applications}, 342(1):334--339, 2008.

\bibitem{keislerrandomizing}
H.~J. Keisler.
\newblock Randomizing a model.
\newblock {\em Advances in Mathematics}, 143(1):124--158, 1999.

\bibitem{KhachiyanP00}
L.~Khachiyan and L.~Porkolab.
\newblock Integer optimization on convex semialgebraic sets.
\newblock {\em Discret. Comput. Geom.}, 23(2):207--224, 2000.

\bibitem{prespowernip}
Q.~Lambotte and F.~Point.
\newblock {On expansions of (Z, +, 0)}.
\newblock {\em {Annals of Pure and Applied Logic}}, 171(8), 2020.

\bibitem{Lenstra83}
H.~W. Lenstra.
\newblock Integer programming with a fixed number of variables.
\newblock {\em Math. Oper. Res.}, 8(4):538--548, 1983.

\bibitem{libkincdbembfmt}
L.~Libkin.
\newblock Embedded finite models and constraint databases.
\newblock In {\em Finite Model Theory and Its Applications}. 2007.

\bibitem{factregression}
D.~Olteanu and M.~Schleich.
\newblock {F:} regression models over factorized views.
\newblock {\em Proc. {VLDB} Endow.}, 9(13):1573--1576, 2016.

\bibitem{3belgians}
J.~Paredaens, J.~V. den Bussche, and D.~V. Gucht.
\newblock First-order queries on finite structures over the reals.
\newblock {\em {SIAM} J. Comput.}, 27(6):1747--1763, 1998.

\bibitem{dualvc}
{Patrick Assouad}.
\newblock {Densit\'e et dimension}.
\newblock {\em Annales de l’Institut Fourier}, 33(3):233–282, 1983.

\bibitem{learninghardnesspittvaliant}
L.~Pitt and L.~G. Valiant.
\newblock Computational limitations on learning from examples.
\newblock {\em J. ACM}, 35(4):965–984, 1988.

\bibitem{poizat1995petits}
B.~Poizat.
\newblock {\em Les Petits Cailloux: Une approche modèle-théorique de l'algorithmie}.
\newblock Aléas, Lyon, 1995.

\bibitem{Renegar92}
J.~Renegar.
\newblock On the computational complexity and geometry of the first-order theory of the reals, part {I:} introduction. preliminaries. the geometry of semi-algebraic sets. the decision problem for the existential theory of the reals.
\newblock {\em J. Symb. Comput.}, 1992.

\bibitem{linearprogrammingbook}
A.~Schrijver.
\newblock {\em Theory of Linear and Integer Programming}.
\newblock John Wiley \& Sons, 1986.

\bibitem{shelahonevar}
S.~Shelah.
\newblock Stability, the f.c.p., and superstability; model theoretic properties of formulas in first order theory.
\newblock {\em Annals of Mathematical Logic}, 3(3):271--362, 1971.

\bibitem{shelahsauer1972}
S.~Shelah.
\newblock A combinatorial problem; stability and order for models and theories in infinitary languages.
\newblock {\em Pacific Journal of Mathematics}, 41:247--261, 1972.

\bibitem{addinglinearorders}
S.~Shelah and P.~Simon.
\newblock Adding linear orders.
\newblock {\em The Journal of Symbolic Logic}, 77(2):717–725, 2012.

\bibitem{nipbook}
P.~Simon.
\newblock {\em A Guide to NIP Theories}.
\newblock Cambridge University Press, 2015.

\bibitem{balderfitting}
B.~ten Cate, M.~Funk, J.~C. Jung, and C.~Lutz.
\newblock Fitting algorithms for conjunctive queries.
\newblock {\em {SIGMOD} Rec.}, 52(4):6--18, 2023.

\bibitem{tongdistal}
M.~Tong.
\newblock {Distal Expansions of Presburger Arithmetic by a sparse predicate}.
\newblock {\em Journal of Symbolic Logic}, page 1–33, 2024.

\end{thebibliography}

\appendix
\begingroup
\fontsize{10}{12}\selectfont 
\onecolumn

\section{$\rqfo$ fitting is preserved when moving to non-standard models}\label{app:nonstandard}

Let $\amodel$ be a structure with domain $M$ and signature $L$. We recall that
the first-order theory of $\amodel$, denoted $\theory(\amodel)$, is the set of
all sentences $\phi$ in the signature $L$ such that $\amodel \models \phi$. And,
a model of $\theory(\amodel)$ is any structure~$\amodel'$ with signature $L$
such that $\theory(\amodel') = \theory(\amodel)$. 
The next lemma shows that if a formula has $\rqfo$ fitting in a structure~$\amodel$, 
then it has $\rqfo$ fitting in all models of $\theory(\amodel)$, 
with respect to the same \rqfo formula $\phi_{\textit{fit}}$ from~Definition~\ref{def:rqfo-fitting}.

\begin{proposition}\label{prop:emb-preserved-nonstandard}
    Let $\phi(x_1,\dots,x_n)$ and $\psi(x_1,\dots,x_n)$ be two formulas over $V \cup L$, where $V$ is a relational vocabulary disjoint from $L$. 
    Suppose $\phi$ and $\psi$ equivalent over embedded finite models for $V$ over $\amodel$.
    For every model $\amodel'$ of $\theory(\amodel)$, 
    $\phi$ and $\psi$ are equivalent over embedded finite models for $V$ over $\amodel'$.
\end{proposition}

\begin{proof}
    Towards a contradiction, suppose $\amodel',\embmodel' \not\models \phi \leftrightarrow \psi$ for some embedded finite model $\embmodel'$ for $V$ over $\amodel'$. 
    Consider the map $\nu \colon \{x_1,\dots,x_n\} \to M'$ witnessing $\amodel',\embmodel',\nu \models {\lnot (\phi \leftrightarrow \psi)}$, 
    and let $\{a_1,\dots,a_m\}$ be the active domain of $\embmodel'$ (see Remark~\ref{remark:activeisenough}). 
    
    We extend the map $\nu$ into ${\nu' \colon \{x_1,\dots,x_n,y_1,\dots,y_m\} \to M'}$ by adding assignments $\nu'(y_i) = a_i$ for all $i \in [1..m]$, with $y_1,\dots,y_m$ fresh. Furthermore, let $\gamma(x_1,\dots,x_n,y_1,\dots,y_m)$ be the formula 
    obtained from $\phi \leftrightarrow \psi$ 
    by replacing each atomic formula $R(z_1,\dots,z_k)$,
    where $R$ belongs to $V$, by $\bigvee_{(a_{i_1},\dots,a_{i_k}) \in R^{\embmodel'}}
    \bigwedge_{j=1}^k z_j = y_{i_j}$. By construction, $\amodel', \nu' \models \lnot \gamma$, 
    hence $\amodel' \models \exists x_1,\dots,x_n,y_1,\dots,y_m \colon \lnot \gamma$. 
    Since $\theory(\amodel') = \theory(\amodel)$,
    this implies that there is a map $\mu' \colon \{x_1,\dots,x_n,y_1,\dots,y_m\} \to M$ such that $\amodel, \mu' \models \lnot \gamma$. 
    Let $\mu$ be the projection of $\mu'$ with domain $\{x_1,\dots,x_n\}$, and let $f$ be the map from $\{a_1,\dots,a_m\}$ to $M$ given by $f(a_i) = \mu'(y_i)$. 
    Let $f(\embmodel')$ be the embedded finite model for $V$ over $\amodel$ 
    given by $R^{f(\embmodel')} \coloneqq \{(f(a_{i_1}),\dots,f(a_{i_k})) \mid (a_{i_1},\dots,a_{i_k}) \in R^{\embmodel'}\}$ for every relation $R$ in $V$ of arity $k$.
    The map $f$ is a homomorphism from $\embmodel'$ to $f(\embmodel')$.
    By construction of $\gamma$, we have $\amodel,f(\embmodel'),\mu \models \lnot (\phi \leftrightarrow \psi)$, contradicting the fact that $\phi$ and $\psi$ are equivalent over embedded finite models for $V$ over $\amodel$.
\end{proof}

\section{Proof of Lemma~\ref{lem:RQfittingImpliesFiniteVC}: $\rqfo$ fitting for
1-parameter formulas implies $\nip$}\label{app:rqfoimpliesfinitevc}

We recall the lemma, which is used in several places in the body.

\lemQfittingImpliesFiniteVC*

Before proving the lemma, let us note that the ability to move from $\rqfo$
fitting to finite VC dimension does not hold at the level of individual
partitioned formulas: a formula can have \rqfo fitting but not VC dimension. 

\begin{example}[\rqfo fitting for $\phi$ $\centernot\implies$  finite VC
    dimension]\label{example:Buchi-infiniteVC-but-RQ-fitting} Consider B\"uchi
    arithmetic (defined in Section~\ref{sec:betweenptimeandrqfoconcept}). Recall
    that in this structure we can define the formula $\textit{bit}(x \varsep y)$
    expressing that $x$ is a power of two appearing in the binary expansion of
    $y$, which has infinite VC dimension as it shatters the powers of two. Yet
    this partitioned formula has \rqfo fitting, given by the \rqfo sentence 
  
  $\forall x \colon P(x) \rightarrow V_2(x,x)$. This sentence states that a
    sample with positive examples~$S^+$ and negative examples~$S^-$ is
    $\gamma$-fittable if and only if every element of $S^+$ is a power of two
    (recall that we always assume samples to satisfy ${S^+ \cap  S^- =
    \emptyset}$, see Definition~\ref{def:sample}).
\end{example}

Although the above counterexample shows that \rqfo fitting does not imply finite
VC dimension for individual formulas, Lemma~\ref{lem:RQfittingImpliesFiniteVC}
states that finite VC dimension can be recovered when the structure $\amodel$
has \rqfo fitting.

\begin{remark} This lemma is closely related to prior results. From the lemma we
immediately derive that $\rqc$ implies $\nip$, which was previously known: see
Remark~\ref{remark:rqc-nip} in the body of the paper. And the proof of this lemma
is a variation of prior proofs that $\rqc$ implies $\nip$, which date back at
least to the complexity-theoretic proof of Lemma~3.9
in~\cite{modeltheoryofstrings}. The idea in each case is that from a partitioned
formula $\phi$ that is a counterexample to $\nip$, we get another partitioned
formula $\gamma$ that encodes $3$-coloring on arbitrarily large embedded finite
graphs. We then argue from the separation of $AC^0$ and $\ptime$ that this is
impossible. The proof we give below will go through nonstandard models that
share a theory with the initial structure, in order to normalize the collection
of counterexamples further. But it should in principle be possible to avoid this
detour, as in the original proof in \cite{modeltheoryofstrings}.
\end{remark}
We now begin the proof of Lemma~\ref{lem:RQfittingImpliesFiniteVC}. We start by
recalling the following result by~Shelah.

\begin{fact}[{\cite[Theorem
    4.6]{shelahonevar}}]\label{fact:oneparamsufficenip-body} If every
    partitioned formula of the form $\phi(\vec x \varsep y)$ has finite VC
    dimension, then the same holds for every partitioned formula: that is, the
    structure is $\nip$. \footnote{\cite[Theorem 4.6]{shelahonevar} considers
    formulas with a single object and multiple parameters. The formulation given
    here relies on VC duality~\cite{dualvc}: finite VC dimension is preserved
    under swapping objects and parameters.}
\end{fact}

So, it suffices to show that every formula with a single parameter variable has
finite VC dimension. We reason by contrapositive: assuming that a partitioned
formula $\phi(\vec x \varsep y)$ has infinite VC dimension, we show that another
formula $\gamma(\vec z \varsep y)$ does not have \rqfo fitting. This is achieved
by forcing the $0$-fitting problem for $\gamma$ to describe (in the sense of
descriptive complexity) an~$\np$-hard problem (3-coloring), and then relying on
the fact that these problems are not definable by \rqfo sentences.

Fix a formula $\phi(x_1 \ldots x_{\ell} \varsep y)$ with infinite VC dimension.
We will define a $\gamma$ that has $6 \cdot \ell$ object variables,
i.e., it is of the form $\gamma( \vec x_0, \vec x_1, \vec x_2,\vec z_0, \vec z_1, 
\vec z_2 \varsep y)$, where the $\vec x_i$ and $\vec z_i$ are vectors of
$\ell$ variables. The formula states that:
\begin{itemize}
    \item $\phi(\vec x_c \varsep y)$ is satisfied by a single $c \in [0..2]$,
    \item $\phi(\vec z_{c'} \varsep y)$ is satisfied by a single $c'\in [0..2]$,
    \item  
    the two indices $c$ and $c'$ in $[0..2]$ such that the above holds are
    distinct.
\end{itemize}
There properties are first-order definable, so $\gamma$ is a formula over
$\amodel$. Intuitively (we formalize this below), each triple of $\ell$-tuples
$(\vec x_0, \vec x_1, \vec x_2)$ represents a single vertex of the graph, as
does $(\vec z_0, \vec z_1, \vec z_2)$. In the fitting problem for $\gamma$, we
feed graph edges as positive examples (each edge is a pair of $3\ell$-tuples, or
equivalently a $6\ell$-tuple), with no negative examples. The parameter variable $y$
then selects a coloring for the vertices, via the indices~$0,1,2$.

\medskip

We will assume, towards a contradiction, that there is a $\rqfo$ sentence
$\widetilde{\gamma}(T)$ that, when interpreting $T$ as a finite set $S^+$ of $6\ell$-tuples, 
is true if and only if the sample with positive examples
$S^+$ and no negative example is $\gamma$-fittable. That is, $\widetilde{\gamma}(T)$ is
equivalent to $\exists y \, \forall \vec z (T(\vec z) \rightarrow \gamma(\vec z
\varsep y))$. Note that this equivalence holds over all embedded finite models
in our structure $\amodel$ and, by
Proposition~\ref{prop:emb-preserved-nonstandard} in
Appendix~\ref{app:nonstandard}, continues to hold in any model of
$\theory(\amodel)$. For convenience of the proof, we will move from $\amodel$ to
a richer model of $\theory(\amodel)$. We rely on the following definition:

\begin{definition}[Indiscernible set]
    Consider an infinite structure $\amodel$ with domain $M$ and language $L$,
    and an infinite linear order $(O,<)$. A set $I \coloneqq \{e_i: i \in O\}$
    of distinct elements from $M$ indexed by $O$ is said to be \emph{order
    indiscernible} (with order type ($O,<$) ) if for every first-order
    $L$-formula $\phi(x_1 \ldots x_k)$, the truth value in $\amodel$ of $\phi$
    on a $k$-tuple $(e_{j_1} \ldots e_{j_k})$ of elements from $I$ is determined
    by the equalities and orderings among the indices $j_1,\dots,j_k$.
\end{definition}

Since $\amodel$ is an infinite structure, a basic fact in model theory, proved
via Ramsey's theorem and compactness~\cite{changkeisler}, is that for any
infinite linear order $(O,<)$, there is a model $\amodel'$ of $\theory(\amodel)$
(with domain $M'$)
which contains an indiscernible set~$I$ with order type $(O,<)$. Furthermore,
since $\phi$ has infinite VC dimension, $\amodel'$ can be
constructed so that there is an infinite set $S$ of $\ell$-tuples that is
shattered by $\phi$, such each component of each tuple comes from $I$. Below we
will sometimes abuse notation by consider the ordering $<$ as being on the
domain $I$ itself, not the indices. Observe that we can take any linear
ordering, so let us take $(\ints,<)$. Call an embedded finite model
\emph{$I$-based} if each component in each tuple in the relations comes from
$I$. 

\begin{proposition}\label{prop:ibased} 
    There is an~\rqfo sentence
    $\widehat{\gamma}(T)$ using only $T$ and $<$ such that, for every $I$-based
    embedded finite model~$\embmodel$, $\amodel',\embmodel \models
    \widetilde{\gamma}$ if and only if $(M',<),\embmodel \models
    \widehat{\gamma}$.
\end{proposition}

\begin{proof} 
    We substitute for every maximal $L$-formula in $\widetilde{\gamma}$ with the
    corresponding inequalities.
\end{proof}

Since shattering is preserved when restricting to a subset of tuples, we can
apply the pigeonhole principle to shrink $S$ while additionally ensuring:
\begin{enumerate}
\item\label{pigeon-i1} All tuples in $S$ have the same $<$ order type. By
    reordering coordinates and ignoring components that are constant,
    without loss of generality we may assume that each tuple $(e_{j_1},\dots,e_{j_\ell})
    \in S$ is strictly increasing with respect to the order on $I$. That is,
    $j_1 < \dots < j_\ell$ (equivalently, $e_{j_1} < \dots < e_{j_\ell}$).
\item\label{pigeon-i2} There is an enumeration $\vec s^0, \vec s^1, \vec s^2, \dots$ of
    the elements of $S$ and an index $k \in [0..\ell]$ such that: for every $j
    \in [1..k]$, the $j$th coordinate of $\vec s^i$ is decreasing in $i$, and
    for every $j \in [k+1..\ell]$, the $j$th coordinate of $\vec s^i$ is
    increasing in $i$.
\end{enumerate}

In the following we will assume for simplicity that $k = 0$: each component is
increasing. The generalization to arbitrary $k$ is routine. This simplification
allows strengthening Items~\ref{pigeon-i1} and~\ref{pigeon-i2} to: 
\begin{enumerate}
\item[($\ast$)] There is an enumeration $\vec s^0, \vec s^1, \vec s^2, \dots$ such that
    for every $i$, given $\vec s^i = (e_{j_1},\dots,e_{j_\ell})$ and 
    $\vec s^{i+1} = (e_{j_1'},\dots,e_{j_\ell'})$, we have 
    $e_{j_1} < \dots < e_{j_\ell} < e_{j_1'} < \dots < e_{j_\ell'}$.
\end{enumerate}

Now, by further shrinking $I$ to remove all elements not contained in $S$, we
can assume that \emph{each tuple enumerates an interval in $I$ of size $\ell$},
That is, we can enumerate the elements $e_0,e_1,\dots$ of $I$ so that the
interval $e_{i \cdot \ell},\,e_{i \cdot \ell + 2},\,\dots,\,e_{i \cdot \ell +
(\ell-1)}$ of length $\ell$ is the $i$th vector $\vec s^i$ in the enumeration
of~$S$. We call this a \emph{round-robin shattered set of tuples}. 

We give further notation, fixing a natural number $n$:  

\begin{itemize}
    \item  Let $I_n \coloneqq (e_{0},\dots,e_{3 \cdot \ell \cdot n-1})$ be the
    initial segment of $I$ of size $3 \cdot \ell \cdot n$.

    \item  By an \emph{aligned $\ell$-tuple} we mean an $\ell$-tuple $(e_{i
    \cdot \ell},\dots,e_{i \cdot \ell + \ell - 1})$ for some $i$, i.e., the
    initial element of the tuple has an index that is a multiple of $\ell$. 

    We let $\alignedtuples_n$ be the set of all aligned $\ell$-tuples from
    $I_n$. We order elements of this set accordingly to the order on elements:
    two aligned $\ell$-tuples $\vec s$ and $\vec s'$ satisfy $\vec s < \vec s'$
    whenever every element in $\vec s'$ is above every element in $\vec s$. 
    We can thus associate each element of $\alignedtuples_n$ with an \emph{index}
    in $[0..3n-1]$. 

    \item We let $\initialtuples_n \subseteq \alignedtuples_n$ be the set of
        tuples with index divisible by $3$. Elements of this set are called
        \emph{initial $\ell$-tuples}.

    \item Let $\tupletriples^n$ be the set (of size $n$) of~$3\ell$-tuples
    formed by taking each $3\ell$ length interval beginning at a tuple in
    $\initialtuples_n$. Equivalently, this is a group of $3$ consecutive aligned
    tuples (according to their indexes) beginning with an initial tuple.

    \item Let $\initial_n$ be the set of $n$ elements from $I_n$ with an index
    divisible by $3 \ell$. Equivalently, these are the first elements of initial
    $\ell$-tuples.

    \item Let $\tupleof_c(v)$ for $c \in [0..2]$ and $v \in \initial_n$ be the
    $(3v+c)^{th}$ $\ell$-tuple in $\alignedtuples_n$. Note that $\tupleof_0(v) =
    v$, and the $3\ell$-tuple $(\tupleof_0(v),\tupleof_1(v),\tupleof_2(v))$
    belongs to $\tupletriples^n$.
\end{itemize}
Figure~\ref{figure:aligned-tuples} depicts the above objects for $\ell = 2$ and
$n = 2$. With this notation in place, we are now ready to reduce $3$-coloring to
fitting.

\begin{figure}
    \begin{center}
        \begin{tikzpicture}[
            element/.style={draw, circle, minimum size=6mm, inner sep=1pt},
            brace/.style={decorate, decoration={brace, amplitude=5pt}},
            label/.style={font=\footnotesize}
        ]

            \foreach \i in {0,...,11} {
                \node[element] (e\i) at (\i*0.8, 0) {\tiny $e_{\i}$};
            }

            \foreach \i in {0,1,2,3,4,5} {
                \pgfmathtruncatemacro{\start}{2*\i}
                \pgfmathtruncatemacro{\end}{2*\i+1}
                \draw[thick, blue!60, rounded corners=2pt] 
                    ([yshift=-10pt]e\end.south east) rectangle ([yshift=-10pt]e\start.south west);
            }

            \foreach \i in {0,1} {
                \pgfmathtruncatemacro{\start}{6*\i}
                \pgfmathtruncatemacro{\end}{6*\i+1}
                \draw[thick, red!70, line width=1.5pt, rounded corners=2pt] 
                    ([yshift=-13pt]e\end.south east) rectangle ([yshift=-13pt]e\start.south west);
            }

            \draw[thick, green!60!black, dashed, line width=1pt, rounded corners=3pt]
                ([yshift=-20pt, xshift=-3pt]e0.south west) rectangle ([yshift=-20pt, xshift=3pt]e5.south east);

            \draw[thick, green!60!black, dashed, line width=1pt, rounded corners=3pt]
                ([yshift=-20pt, xshift=-3pt]e6.south west) rectangle ([yshift=-20pt, xshift=3pt]e11.south east);

            \foreach \i in {0,6} {
                \node[draw=orange, circle, line width=1.5pt, minimum size=8mm] at (e\i) {};
            }

            \node[right= 0.8cm of e11, anchor = west, font=\footnotesize, align=left] (legend) {
                \begin{tabular}{rl}
                \textcolor{blue!60}{\rule{8pt}{2pt}}        & $\alignedtuples_n$ \\[2pt]
                \textcolor{red!70}{\rule[0.5pt]{8pt}{3pt}}  & $\initialtuples_n$ \\[2pt]
                \textcolor{green!60!black}{- - -}           & $\tupletriples_n$  \\[2pt]
                \textcolor{orange}{\Large$\circ$}           & $\initial_n$
                \end{tabular}
            };
        \end{tikzpicture}
    \end{center}
    \caption{Example of the initial segment $I_n$ and its subsets, for $\ell = 2$ and $n = 2$.}\label{figure:aligned-tuples}
\end{figure}
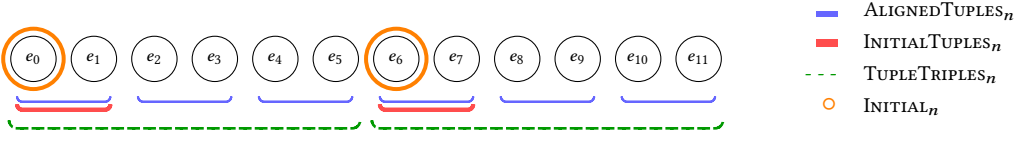

Consider a graph $G = (V,E)$ with vertices~$V \coloneqq [0..n-1]$, for some $n
\geq 2$, and edges $E$ without self-loops. We identify each vertex $i \in
[0..n-1]$ with the $i^{th}$ element of $\initial_n$. Note that, then, tuples in
$\tupletriples^n$ are triples of the form 
\[ 
    \tupleof_0(i), \tupleof_1(i), \tupleof_2(i)
\]
where $i$ is a vertex of $G$. We encode $G$ with the $I$-based embedded finite
model $T_G^I$ that interpret the relation $T$ as follows: $T$ contains a
$6\ell$-tuple $\vec w$ if and only if there is an edge $(i,j) \in E$ such that 
\[ 
    \vec w = (\tupleof_0(i), \tupleof_1(i), \tupleof_2(i), \tupleof_0(j), \tupleof_1(j), \tupleof_2(j)).
\]

The following claim will follow by just definition chasing:

\begin{claim}\label{clm:colorableandfittable} For any $G$ as above, $G$ is
    $3$-colorable exactly when $T_G^I$ is $\gamma$-fittable (w.r.t.~$(M',<)$).
\end{claim}
\begin{proof}Suppose $T_G^I$ is $\gamma$-fittable with witness $y^*$ for the
parameter variable. Then for each vertex $v \in G$ we color $v$ with the unique
$c \in [0..2]$ such that $\phi(\tupleof_c(v),y^*)$ holds. By the definition of
$\gamma$, there is exactly one such $c$ for each $v$.

Consider an edge~$(v, v')$ in $G$, and suppose $v$ and $v'$ are given the same
color $c$ according to the recipe above. Then consider the following $6$-tuple
of $\ell$-tuples, which can also be considered as a single $6\ell$-tuple,
concatenating:
\begin{align*}
    \tupleof_0(v) , \tupleof_1(v) , \tupleof_2(v),
    \tupleof_0(v'), \tupleof_1(v'), \tupleof_2(v')
\end{align*}
This tuple is in the training set. But we cannot have $\gamma(\vec w'; y^*)$ for
this tuple as required, since the tuple violates the last condition in the
definition of $\gamma$. Therefore, we have produced a $3$-coloring of $G$, as
required.

Conversely, suppose $G$ is $3$-colorable with partition of the vertices~$V$ into
$V_0, V_1, V_2$. Partition the tuples $\{\tupleof_0(v): v \in V\}$ into tuples
with $v \in V_0$ and those with $v \in V-V_0$, calling the former \emph{positive
$\tupleof_0$ values}. Similarly, partition $\{\tupleof_1(v): v \in V\}$ into
tuples with $v \in V_1$ and those with $v \in V-V_1$, and partition
$\{\tupleof_2(v): v \in V\}$ into tuples with $v \in V_2$ and those with $v \in
V-V_2$. Since every vertex obtains exactly one of the colors, we know that for
every $v$, there is exactly one $c$ such that the value of $\tupleof_c(v)$ is
positive. Consider the set of all $\ell$-tuples of the form $\tupleof_c(v)$ that
are in some positive partition, and its complement. This partitions the set of
$\alignedtuples_n$ into subsets $T^+$ and $T^-$. Since $\alignedtuples_n$ is
$\phi$-shattered, there is $y^*$ that is $\phi$-related to exactly  the tuples
in $T^+$. Note that:
\begin{itemize}
    \item 
since each $v$ is in exactly one of $V_0, V_1, V_2$, exactly
one of $\tupleof_c(v)$ for $c \in [0..2]$ is in $T^+$
\item if $(v,v')$ is an edge of $G$, and $c$ and $c'$ are as above for $v$ and $v'$,
then $c$ and $c'$ must be distinct.
\end{itemize}

We claim that this $y^*$ is a witness to $\gamma$-fitting of $T^I_G$.
We first need to check that for each $v$ $\phi(\tupleof_c(v),y^*)$ holds for exactly one $c \in [0..2]$: this holds by the first item above.
We then need to check that for every edge $(v,v')$ in $G$, the colors $c$ and $c'$ for $v$ and $v'$, respectively, are distinct.
This follows from the second item above.
This completes the proof of the claim.
\end{proof}

We now continue to the proof of the lemma.

By definition of $\widetilde{\gamma}$, and from Claim~\ref{clm:colorableandfittable} and Proposition~\ref{prop:ibased}, 
a finite graph $G$ as above is $3$-colorable exactly when $(M',<),T_G \models \widehat{\gamma}$, 
where $\widehat{\gamma}$ is an~\rqfo sentence constructed from $\widetilde{\gamma}$.

Now given graph $G$ with domain an initial segment of the natural numbers, let $T^\ints_G$ be built the same way as $T^I_G$, but using the natural number $i$ instead of the
element of $I$ having index~$i$. 
Therefore, $T^\ints_G$ with the natural order on $\ints$ is isomorphic to $T^I_G$ with the indiscernible ordering.
Since $\widehat{\gamma}$ uses only $T$ and the ordering relation, we conclude that

\smallskip

\textit{For any finite graph $G$ whose domain is an initial segment of the natural numbers, having at least two vertices and no self-loops, 
$G$ is $3$-colorable exactly when $(\ints,<),T^\ints_G \models \widehat{\gamma}$.}

\smallskip

Now we can observe:
\begin{proposition} 
    The mapping from $G$ to $T^\ints_G$ is definable by a quantifier-free formula over $(\ints,+,<)$.
\end{proposition}
\begin{proof}
    This is just the mapping $i \mapsto (3 \cdot \ell \cdot i,\,3 \cdot \ell \cdot i + 1,\, \dots,\, 3 \cdot \ell \cdot i + 3 \cdot \ell - 1)$.
\end{proof}

Composing the formula that builds $T^\ints_G$ with the formula $\widehat{\gamma}$, we conclude that:

\begin{proposition}\label{cor:rqforcolorabilty} 
    There is an \rqfo sentence $\chi(G)$ over $(\ints, +,<)$ that, given a graph $G$ whose vertices are an initial segment of the natural numbers, of size at least $2$, holds exactly when $G$ is $3$-colorable.
\end{proposition}
We have thus found a first-order formula $\chi(G)$ quantifying only over the domain of $G$, using ordering
and arithmetic, that determines if $G$ is $3$-colorable.
This contradicts the fact that \mbox{$3$-colorability} is not in $AC^0$ \cite{FurstSS84}. 
Therefore, $\widetilde{\gamma}$ does not exist, i.e., $\gamma$ does not have \rqfo fitting.  
This completes the proof of Lemma~\ref{lem:RQfittingImpliesFiniteVC}.

\section{Proof of Theorem \ref{thm:rqfoimpliesptimeapprox}: $\rqfo$ approximate fitting in \ptime}
\label{app:fromexacttoapprox}

We recall the statement of Theorem~\ref{thm:rqfoimpliesptimeapprox}:

\thmrqfoimpliesptimeapprox*

Recall that the approximate fitting problem for $\phi(\vec x
\varsep \vec y)$ asks, for an input tolerance $\epsilon \in \N$ and sample $(\vec s_1, b_1),\dots,(\vec s_n,
b_n)$, whether there is a parameter $\vec p$
such that $\phi(\vec s_i,\vec p)$ disagrees with $b_i$ in at most $\epsilon$ of
the~$n$ examples. Let us assume, for the moment, that $\phi$ has VC dimension~$d$. 
Then, the Shelah-Sauer lemma
(Fact~\ref{fact:sslemma}) guarantees that there are 
$O(n^d)$ ways to break the examples up into positive and negative. The lemma does not provide a method for constructing this family, but we can use \ptime $0$-fitting to circumvent this problem:

\begin{restatable}[Computing realizable partitions]{lemma}{LemComputeAllTypes}\label{lem:computealltypes} 
    Consider a concept class $\conceptclass$  with finite VC dimension and $\ptime$
    $0$-fitting. There is a polynomial-time algorithm that, given as input a
    finite set $S$ of elements of the range space, determines all the
    \emph{realizable partitions} of $S$, that is, the partitions of $S$ into
    positive and negative examples that are $0$-fitted by some
    hypothesis~in~$\conceptclass$.
\end{restatable}

\begin{proof}
    Let  $S= \{s_1 \ldots s_n\}$. The algorithm proceeds in $n$ rounds, maintaining
    a working set~$W$. Initially, $W$ is empty, and after the $j^{th}$ iteration, it
    contains all realizable partitions of $s_1 \ldots s_j$. In the $(j+1)^{th}$
    iteration, the algorithm takes each realizable partition in $W$, say $(S^+,
    S^-)$, and queries the $0$-fitting algorithm twice: first, to check if
    there is a parameter that fits~${(S^+ \cup \{ s_{j+1} \},S^-)}$
    exactly, and second, to check if there is a parameter that fits $(S^+,S^- \cup
    \{s_{j+1}\})$ exactly. Successful queries add the corresponding partition to $W$.

    By~Fact~\ref{fact:sslemma}, the total number of realizable partitions is $O(n^d)$, where $d$ is
    the VC dimension of $\conceptclass$. Since the concept class is fixed, 
    this ensures that the working set remains of polynomial size throughout the procedure.
    Then, polynomial runtime follows from \ptime $0$-fitting.
\end{proof}

Thanks to Lemma~\ref{lem:computealltypes}, approximate fitting reduces to
determining whether the smallest distance between any realizable
partition and the given sample is at most $\epsilon$.

\begin{lemma}\label{lem:FiniteVCandPTimeFittingImpliesPtimeApprox}
    Let $\conceptclass$ be a concept class with finite VC dimension and $\ptime$
    $0$-fitting. Then, $\conceptclass$ has \ptime approximate fitting.
\end{lemma}

\begin{proof}
    The polynomial-time algorithm for approximate fitting 
    proceeds as follows. Given an input sample 
    $S = (s_1, b_1),\dots,(s_n,b_n)$, it computes (in polynomial time, by
    Lemma~\ref{lem:computealltypes}) all realizable partitions of
    $\{s_1,\dots,s_n\}$. For each such partition of positive examples~$S^+$ and
    negative examples~$S^-$, the algorithm computes its distance from the
    sample $S$, defined as the number of indices $i \in [1..n]$ such that $s_i \in S^+$
    but $b_i \neq 1$, or $s_i \in S^-$ but $b_i \neq 0$. If the minimum distance
    across all partitions is at most $\epsilon$, the algorithm 
    accepts;
    otherwise, the algorithm rejects.
\end{proof}

For the proof of Theorem~\ref{thm:rqfoimpliesptimeapprox}, consider now a
structure $\amodel$ with $\rqfo$ fitting. By
Lemma~\ref{lem:RQfittingImpliesFiniteVC}, the structure is~\nip. By
Proposition~\ref{prop:rqfoimpliesptime}, every formula over $\amodel$ has \ptime
$0$-fitting. Then, by Lemma~\ref{lem:FiniteVCandPTimeFittingImpliesPtimeApprox},
$\amodel$ has \ptime approximate fitting.

\section{Proof of Theorem \ref{thm:rqcandrqfitting}: $1$-parameter $\rqfo$ fitting implies \rqc} 
\label{app:oneparameterfittingimpliesrqc}

Recall the statement of Theorem \ref{thm:rqcandrqfitting}:

\TheoremRQCandRQFitting*

\begin{proof}
It suffices to show that if every formula $\phi(\vec x \varsep y)$ over $\amodel$ has $\rqfo$ fitting, then $\amodel$ is $\rqc$. 
Our proof uses \emph{uniform definability of types over finite sets} (UDTFS). A
partitioned formula $\gamma(\vec x \varsep \vec y)$ with $j$ object variables
and $k$ parameter variables is 
$\udtfs$ (with respect to $\amodel$) whenever:

\smallskip

\textit{There is an $L$ formula $\delta(\vec y \varsep \vec p)$ such that:
for any finite set $S$ of $k$-tuples  in $\amodel$ of cardinality at least~2, for any $j$-tuple $\vec x_0$, there is %
$\vec p_0$ from $S$ such that
$\amodel \models \forall \vec y \in S ~ \gamma(\vec x_0, \vec y) \leftrightarrow \delta(\vec y, \vec p_0)$.}

\smallskip

The structure~$\amodel$ is said to be $\udtfs$ when each partitioned formula over it is.
We will make use of the following  result of Chernikov and Simon:

\begin{fact}[\cite{simon-chernikov-two}] \label{fact:nipimpliesudtfs} Every $\nip$ structure has $\udtfs$.
\end{fact}

By Lemma~\ref{lem:RQfittingImpliesFiniteVC}, $\amodel$ is \nip. To show the
theorem, we 
consider a relational signature $V$ disjoint from the
signature $L$ of $\amodel$, and show that every first-order formula in~$L \cup V$ is equivalent over embedded finite models for $\amodel$ to
an~$\rqfo$ formula. By induction on the number of unrestricted quantifiers in
the formula, it suffices to eliminate (bottom-up) a single
unrestricted quantifier.
So, consider a formula $\psi(\vec z)$, in the signature $L \cup V$, of the form
\[
\psi(\vec z) \coloneqq \exists y ~ Q_1(u_1) \ldots Q_k(u_k) ~ \gamma(y, \vec u, \vec z),
\]
where $\gamma$ is an $L$ formula, and each $Q_i(u_i)$ is an active domain quantifier (i.e., it is either~$\exists u_i \in \adom$ or $\forall u_i \in \adom$). The fact that we can restrict to these quantifiers is part of the induction hypothesis. The assumption that $\gamma$ is an $L$ formula follows because in $\rqfo$ formulas, one can assume that the $V$ predicates occur only in restricted quantifiers.

We restrict our attention to embedded finite models with active domain $\adom$
of at least two elements. The corner cases of zero and one elements are covered
in Appendix~\ref{subsec:rqfitting-to-rqc:corners}. Additionally, we assume that
$y$ does not occur in a relation from $V$. If it does, replace every occurrence
$R(\vec w_0,y,\vec w_1)$ of such a relation with $\exists u \in \adom \,(u = y
\land R(\vec w_0,u,\vec w_1))$, where $u$ is a fresh variable. Then restore the
formula to prenex form.

Let $\{\gamma_1(y, \vec u, \vec z), \ldots, \gamma_m(y, \vec u, \vec z)\}$ 
be the set of $L$-formulas appearing in~$\gamma$ that are maximal
among $L$-formulas with respect to the subformula relation.
Seeing each $\gamma_i$ as a partitioned formula $\gamma_i(y,\vec z \varsep \vec u)$, 
we apply UDTFS to get a formula~$\delta_i(\vec u \varsep \vec p)$.
Let $\eta(\vec z, \vec p,A)$, where $A$ is a new $k$-ary relation symbol (to be removed later), be the formula:
\begin{equation}
    \label{equation:etaudtfs}
    \eta(\vec z, \vec p,A) \coloneqq
    \exists y\,\forall \vec u \in A\,
    \bigwedge\nolimits_{i=1}^m [\gamma_i(y, \vec u, \vec z) \leftrightarrow \delta_i(\vec u, \vec p)].
\end{equation}

Directly from the definition of UDTFS, 
it follows that $\psi(\vec z)$ is equivalent, over embedded finite models with active domain 
at least two, to the formula 
\[
    \exists \vec p \in \adom^k \,\big( \eta(\vec z, \vec p, \adom^k) \land \exists y ~ Q_1(u_1) \ldots Q_n(u_n) \, \chi\big),
\] 
where $\chi(y,\vec u, \vec z, \vec p)$ is obtained from $\gamma$ by replacing every maximal subformula~$\gamma_i$ by $\delta_i$, 
for every $i \in [1..m]$, 
and $\adom^k$ stands for the $k$-ary Cartesian product of 
the active domain. Since we were assuming that $y$ appears in $\gamma$ only in $L$-formulas, it does not appear in $\chi$. 
We can therefore reduce the formula $\exists y ~ Q_1(u_1) \ldots Q_n(u_n) \, \chi$ above to $Q_1(u_1) \ldots Q_n(u_n) \, \chi$.

To conclude the proof, it now remains to convert the formula~$\eta$
from~Equation~\eqref{equation:etaudtfs} into an~$\rqfo$ formula.
Like $\psi$, the formula~$\eta$ has a single unrestricted quantifier; 
unlike $\psi$, however, it has only one block of universal restricted quantifiers and its quantifier-free part is over the signature~$L$. 
For the conversion we rely on the~$\rqfo$ fitting hypothesis, 
and we will obtain equivalence over each finite interpretation of $A$.

Define $\eta_{\textit{qf}}(\vec u, \vec z, \vec p \varsep y) \coloneqq \bigwedge_i
[\gamma_i \leftrightarrow \delta_i]$, that is, the quantifier-free part of $\eta$
with $y$ as the only parameter variable. 
\rqfo~fitting can be easily seen to
give an $\rqfo$ \emph{sentence} $\eta^*(P)$, where $P$ is a relation of arity
the number~$\ell$ of variables in $\vec w \coloneqq (\vec u,\vec z, \vec p)$,
with the following property: 

\medskip

for every finite interpretation of~$P$, $\eta^*(P)$ is
true if and only if $\exists y \,\forall \,\vec w (P(\vec w) \rightarrow
\eta_{\textit{qf}}(\vec w \varsep y))$ holds.

\medskip

Indeed, this follows from considering $P$ as the set of positive examples for fitting $\eta_{\textit{qf}}(\vec w \varsep y)$.

Informally, we want now to restrict $P$ to relations of the form $A \times
\{(\vec z_0,\vec p_0)\}$, where $(\vec z_0,\vec p_0)$ is a possible
instantiation of the variables $\vec z,\vec p$ that occur free in the starting
formula~$\eta$. Without loss of generality, we can assume that the variables $\vec z$ and $\vec p$ 
do not occur in the sentence $\eta^*$. We replace all occurrences of
$P(\vec w_0,\vec w_1, \vec w_2)$ in $\eta^*$ (where $\vec w_0$ and $\vec w_1$
have the same cardinality as $\vec u$ and $\vec z$, respectively) with the
conjunction $A(\vec w_0) \land \vec w_1 = \vec z \land \vec w_2 = \vec p$.
This results in an~\rqfo formula $\eta'(\vec z, \vec p, A)$
equivalent, for every finite interpretation of $A$ and every 
sequence $(\vec z_0,\vec p_0)$ of elements in the domain of~$\amodel$, 
to
\begin{align*}
    \exists y \, \forall \vec w_0, \vec w_1, \vec w_2 \,\big(A(\vec w_0) \land \vec w_1 = \vec z_0 \land \vec w_2 = \vec z_0 \rightarrow \eta_{\textit{qf}}(\vec w_0,\vec w_1, \vec w_2 \varsep y)\big).
\end{align*}
Thus, $\eta'$ is equivalent to $\eta$, completing the proof.
\end{proof}

\forarxiv{We note that the variation of Theorem \ref{thm:rqcandrqfitting} for \emph{object variables} 
that is, to get \rqfo~fitting in a structure, it suffices to verify it for formulas $\phi(x \varsep \vec y)$) does not hold: see Appendix~\ref{app:oneobjvarinsufficient}. This is not particularly surprising, since the role of object variables and parameter variables in fitting problems is very different. }

\subsection{Returning to some corner cases: what if the active domain is small}\label{subsec:rqfitting-to-rqc:corners}

Let us consider  the proof of the more difficult direction in the proof of Theorem \ref{thm:rqcandrqfitting}: establishing $\rqc$ while assuming only $\rqfo$-fitting for one parameter formulas. In that argument, we sometimes assumed that the
domain of embedded finite models had size at least two.
We now show how to lift the assumption on the cardinality of the active domain used during the proof of Theorem ~\ref{thm:rqcandrqfitting}. 

Recall that we are considering a formula $\psi(\vec z)$, in the signature $L \cup \relvoc$, of the form
\[
\psi(\vec z) \coloneqq \exists y ~ Q_1(u_1) \ldots Q_k(u_k) ~ \gamma(y, \vec u, \vec z),
\]
where $Q_i$ are restricted quantifiers and $\gamma$ is quantifier-free.
We want to translate this formula into a $\rqfo$ formula. We can ``explore'' whether the active domain $\adom$ has $0$, $1$ or at least two elements by rewriting $\psi$ into: 
\[ 
    (\psi \land \forall x \in \adom \colon \false) 
    \lor (\psi \land \forall x,y \in \adom \colon x = y)
    \lor (\psi \land \exists x,y \in \adom \colon x \neq y).
\]
The three disjuncts handle empty active domains, singletons, and domains of size
at least two, respectively. In the latter case, the occurrence of $\psi$ can be translated into an $\rqfo$ formula following the proof given in the body of the paper. Below, we show how to handle the remaining two cases:
\begin{description}
\item[Empty active domain:]
We discuss how to translate $\psi$ into an \rqfo formula under the assumption $\forall x \in \adom \colon \false$. If the number $k$ of restricted quantifiers is at least one, then observe that $Q_1(u_1) \ldots Q_k(u_k) ~ \gamma(y, \vec u, \vec z)$ (and therefore also $\psi$) reduces to $\false$ whenever $Q_1$ is an existential quantifier, and to $\true$ whenever $Q_1$ is a universal quantifier. If $k = 0$ instead, we can replace in $\psi$ all occurrences of relations from the signature $\relvoc$ 
with $\false$, obtaining an $L$-formula. Since any $L$-formula is an $\rqfo$ formula, we are done. 

\item[Singleton active domain:]
To translate $\psi$ under the assumption $\forall x,y \in \adom \colon x = y$, given a subset $S \subseteq \relvoc$ of the relations in $\relvoc$, 
let us write $\gamma_S$ for the sentence stating that exactly the relations in $S$ are non-empty (and are thus singletons): 
\[
    \gamma_S \coloneqq (\bigwedge_{R \in S} \exists x \colon R(x,\dots,x)) \land \bigwedge_{R \in V \setminus S} \forall x (R(x,\dots,x) \rightarrow \false).
\]
We rewrite $\psi \land \forall x,y \in \adom \colon x = y$ as $\bigvee_{S \subseteq V}(\psi \land \gamma_S \land \forall x,y \in \adom \colon x = y)$, and translate %
each disjunct. 
Fix $S \subseteq \relvoc$, and let $x^*$ be a fresh variable that will be bound to the only element in the active domain.
Define $\psi_{S}(x^*,\vec z)$ by replacing in $\psi$ 
all occurrences $R(\vec w)$ of a relation in $\relvoc$ 
(including in restricted quantifiers)
with $\false$ if $R \not\in S$, 
and with $\vec w = (x^*,\dots,x^*)$ otherwise. 
Observe that $\psi_{S}(x^*,\vec z)$ is an $L$-formula, 
and it is thus an $\rqfo$ formula.
Lastly, replacing each $\psi \land \gamma_S \land \forall x,y \in \adom \colon x = y$ with $(\exists x^* \in \adom :\psi_S) \land \gamma_S \land \forall x,y \in \adom \colon x = y$ yields the desired $\rqfo$ formula.
\end{description}

\forarxiv{

\section{\rqfo fitting: one object variable does not suffice}\label{app:oneobjvarinsufficient}

Theorem~\ref{thm:rqcandrqfitting} shows that, in order for a structure to have
$\rqfo$ fitting, it suffices that all formulas with one
parameter variable have $\rqfo$ fitting. 
It is then natural to ask whether the analogous result holds when assuming $\rqfo$ fitting on  all formulas with one \emph{object variable}. We will give a negative answer to this question.
Here is our counterexample:

\begin{example}\label{ex:onevardoesnotsuffice} There are structures %
for which some fitting problems are $\np$-hard
but where every $\phi(x \varsep \vec y)$ has $\rqfo$ fitting.
This happens in the random graph. Recall \cite{faginzeroone} that this structure can be defined, uniquely up to isomorphism, as the countable binary edge relation $E(x,y)$ that is anti-reflexive, symmetric, and satisfies
the following ``extension axiom'':

\medskip 

\textit{For every
finite graph $G_1$ and every induced subgraph $G_0$ of $G_1$,
every isomorphic copy of $G_0$ in $E$ extends to a copy of $G_1$ in $E$.}

\medskip

The structure is well-known to be IP: indeed, it is easily seen that the edge relation $E(x,y)$ shatters arbitrarily large sets. From this, we can easily encode coloring as a fitting problem, and thus some fitting problems are $\np$-hard. It is also easy (and well-known) that the random graph admits
quantifier-elimination~\cite{qerandomgraph}. Therefore, over the random graph a formula $\phi(x, \vec y)$ asserts a Boolean combination of equalities and edge relations among $x$ and $\vec y$.  

By an \emph{atomic type} in variables $\vec w$, we mean a conjunction of literals giving a complete syntactically consistent description of the equalities and edge descriptions on the variables~$\vec w$.
Let us write $AT(\vec w)$ for the set of all atomic types in variables $\vec w$.

We can rewrite $\phi(x, \vec y)$ as
\begin{equation}\label{eq:random-types}
    \bigvee_{\substack{\gamma \in AT(x,\vec y)\\\gamma \models \phi}} \gamma.
\end{equation}
Let $k$ be the number of parameter variables in $\phi$.

For the sake of the discussion, let us fix a sample of positive examples $S^+$ and negative examples~$S^-$.
To define the \rqfo sentence showing \rqfo fitting for $\phi$, we reason by cases:
\begin{description}
    \item[Case 1: every disjunct in Formula~\ref{eq:random-types} asserts~$x = y_i$ for some $i \in {[1..k]}$.] In this case, whenever $S^+$ is non-empty, the parameters must be picked from elements in this set; and so it suffices to quantify the parameter variables over the active domain. When $S^+$ is empty however, it suffices to pick as parameters arbitrary elements that do not occur in $S^-$. Therefore, the \rqfo sentence we seek is: 
        \[
            (\exists \vec y \in \adom^k \,(\forall x (P(x) \rightarrow \phi) \land \forall x (N(x) \rightarrow \lnot \phi))) \lor (\forall x (P(x) \rightarrow \bot)).
        \]
    \item[Case 2:]\textbf{The set $D$ of disjuncts from Formula~\ref{eq:random-types} that feature $x \neq y_i$ for every $i \in [1..k]$ is non-empty. There is a combination of literals on $E(x,y_1),\dots,E(x,y_k)$ that does not appear in any of the disjuncts in the set $D$.}
        In this case, we claim that the \rqfo formula we seek is simply
        $\top$. Let $\eta^-(x,\vec y)$ be a combination of literals on $E(x,y_1),\dots,E(x,y_k)$ that does not appear in the set $D$. Let $\eta^+(x,\vec y)$ be a combination of literals on $E(x,y_1),\dots,E(x,y_k)$ that does appear in $D$ (the set is non-empty).
        Consider the formula 
        \[
            {\bigwedge_{p \in S^+} \left(\eta^+(p,\vec y) \land \bigwedge\nolimits_{i=1}^k p \neq y_i\right) \land \bigwedge_{q \in S^-} \left(\eta^-(q,\vec y) \land \bigwedge\nolimits_{i=1}^k q \neq y_i\right)}.
        \]
        The sample is clearly $\phi$-fittable when this formula is true. But, in the random graph, this formula is always true: by the extension axiom, we can always find a set of vertices of the graph that are distinct from the vertices in the sample and realize any consistent relation on edges.
    \item[Case 3:]\textbf{for every combination of literals on $E(x,y_1),\dots,E(x,y_k)$, there is a disjunct in Formula~\ref{eq:random-types} featuring that combination and asserting~$x \neq y_i$ for every $i \in {[1..k]}$.}  
    In this case, by the extension axiom, the sample is $\phi$-fittable if and only if it is fittable with the formula $(\bigwedge_{i=1}^k x \neq y_i) \lor \psi$, where $\psi$ is the formula featuring all the disjuncts from Formula~\ref{eq:random-types} that have some equality. 
    The fitting problem thus asks: 
    \[
        \exists \vec y \forall x ((P(x) \rightarrow (\bigwedge_{i=1}^k x \neq y_i) \lor \psi) \land (N(x) \rightarrow (\bigvee_{i=1}^k x = y_i) \land \lnot \psi)).
    \]
    If the set of negative examples~$S^-$ is empty, then we can trivially make this formula true by picking parameters that do not occur in $S^+$. 
    Otherwise, the parameters must be picked from elements in $S^-$. 
    Therefore, the \rqfo sentence we seek is: 
    \[
        (\forall x (N(x) \rightarrow \bot)) \lor \exists \vec y \in \adom^k \,(\forall x (P(x) \rightarrow \phi) \land \forall x (N(x) \rightarrow \lnot \phi)).
    \]
\end{description}
Therefore, by case analysis on the structure of Formula~\ref{eq:random-types}, 
we are able to construct the \rqfo sentence for \rqfo fitting.
\end{example}
}

\section{$\np$-hard fitting problems for \ip structures} \label{app:ipwitness}

In Section \ref{sec:betweenptimeandrqfoconcept} we mentioned that B\"uchi
arithmetic, which is \ip, has $\np$-hard fitting problems. This if formally
stated in~Proposition~\ref{prop:buchiexactconcept}, which will be proven in
Appendix~\ref{app:buchinpmembership}. In fact  other decidable \ip structures,
like atomless Boolean algebras and the random graph, also have hard fitting
problems. In this appendix we give a general property that implies that such
problems exist, which will apply to B\"uchi arithmetic.

\begin{definition}{Tractable IP witnesses}
    An \ip structure $\amodel$ is said to have \emph{tractable IP witnesses}
    whenever there is a formula~$\gamma(\vec x \varsep  \vec y)$ from $\amodel$
    and a polynomial-time algorithm that given $i \in \N$ (in unary) produces a
    set $S_i$ of $i$ (representations of) effective constants
    that~$\conceptclass_\gamma$ shatters.
\end{definition}

Through tractable IP witnesses, we can make
Lemma~\ref{lem:RQfittingImpliesFiniteVC} ``algorithmic'', and show that these
structures will have \np-hard $0$-fitting problems. Notably, every \ip structure
with \ffte we are aware of has tractable IP witnesses.

\begin{proposition}\label{prop:ipimplieshard} If~$\amodel$ has tractable IP
witnesses, then some formula $\gamma$ over $\amodel$ has an~\np-hard $0$-fitting
problem.
\end{proposition}

\begin{proof}
    Let~$\phi(\vec x \varsep \vec y)$ be a witness to tractable IP witnesses,
    with $\ell$ object variables. We rely on the $3$-coloring encoding from the
    proof of Lemma~\ref{lem:RQfittingImpliesFiniteVC}
    (Appendix~\ref{app:rqfoimpliesfinitevc}), defining $\gamma$ from $\phi$ as
    described in the proof of the lemma. That is, $\gamma( \vec x_1, \vec x_2,
    \vec x_3,\vec z_1, \vec z_2, \vec z_3 \varsep y)$ has $6 \cdot \ell$ object
    variables.
    
    Let $G = (V,E)$ be an undirected graph with vertices $V = [1..n]$. Thanks to
    tractable IP witnesses, we can compute in polynomial time a set of $3 \cdot
    n$ many $\ell$-tuples $\vec s_1, \dots \vec s_{3 \cdot n}$ that is shattered
    by $\phi$. Let $\mu$ be the map sending $j \in [1..n]$ to the $3 \cdot
    \ell$-tuple $(\vec s_{3(j-1)+1},\vec s_{3(j-1)+2}, \vec s_{3(j-1)+3})$.
    Consider the sample $S$ with no negative examples and positive examples
    consisting of the set of $6 \cdot \ell$-tuples $\{(\mu(j),\mu(k)) : (j,k)
    \in E\}$. Following the semantics of $\gamma$ given in the proof of
    Lemma~\ref{lem:RQfittingImpliesFiniteVC}, we see that $S$ is
    $\gamma$-fittable if and only if $G$ is $3$-colorable. 
\end{proof}

\section{Proof of Proposition~\ref{prop:buchiexactconcept}: $\np$ fitting for
B\"uchi Arithmetic}\label{app:buchinpmembership}

We start out by proving Proposition~\ref{prop:buchiexactconcept}, which we now
recall:

\buchiexactconcept*

The hardness follows from Proposition~\ref{prop:ipimplieshard}
(Appendix~\ref{app:ipwitness}), as the formula~$\textit{bit}(x \varsep y)$ from
Section~\ref{sec:betweenptimeandrqfoconcept} has tractable IP witnesses (the
powers of two, encoded in binary). Therefore, we focus on the $\np$ upper bound.
For this, we rely on a further property of structures: 

\begin{definition}\label{def:pwp} A structure~\amodel has the \emph{polynomial
witnessing property} ($\pwp$) if for every first-order formula $\phi(\vec x,
\vec y)$ there is a polynomial~$p$ with the following property: for every
sequence of tuples $\vec c_1 \ldots \vec c_n$ of effective constants in
$\amodel$, if $\exists \vec y\, \bigwedge_{i=1}^n \phi(\vec c_i, \vec y)$ holds,
then there is a witness tuple $\vec y$ of effective constants such that
$\bigwedge_{i=1}^n \phi(\vec c_i, \vec y)$ holds and $\bitsize{\vec y} \leq
p(\sum_{i=1}^n \bitsize{\vec c_i})$.
\end{definition} 
Recall from Section \ref{sec:prelims} that $\bitsize{w}$ denotes the length of a word $w$.

First, we note that $\pwp$ (along with $\ffte$, which we always assume) implies
the~$\np$ bound:
\begin{restatable}{proposition}{PropPWPImpliesNPFitting}\label{prop:pwpimpliesnpfitting}
For a definable concept class over a structure $\amodel$ with $\pwp$,
approximate (and exact) fitting is in $\np$. 
\end{restatable}
\begin{proof} 
    Let $S$ be an input sample, and $\epsilon \in \N$ be the input tolerance.
    Observe that the approximate fitting problem for concept classes reduces in
    non-deterministic polynomial time to the exact fitting problem: it suffices
    to guess a subset $J \subseteq [1..n]$ with at most $\epsilon$ indices, and
    then solve the $0$-fitting problem for the sample obtained from the one in
    input by replacing, for every $j \in J$, the Boolean $b_{j}$ by its
    negation. Below we focus on exact fitting.

    $\pwp$ immediately gives an $\np$ bound for samples with positive (or
    negative) examples only. We can reduce to this case by adding two object
    variables. Given $\phi(\vec x \varsep \vec y)$ define 
    \[ 
        \phi'(\vec x, x',x'' \varsep \vec y) \coloneqq 
        (x'= x'' \wedge \phi(\vec x, \vec y)) \vee (x' \neq x'' \wedge \neg \phi(\vec x, \vec y)).
    \]
    Consider the positive examples $S^+$ and negative examples $S^-$ of the
    input sample $S$, and $c_0$ and $c_1$ be two distinct effective constants
    from $S$. Let $S'$ be the sample with positive examples $\{(\vec q,c_0,c_0)
    : \vec q \in S^+\} \cup \{(\vec q,c_0,c_1) : \vec q \in S^-\}$ and no
    negative example. Then, $S$ is $\phi$-fittable if and only if $S'$ is
    $\phi'$-fittable, and by $\pwp$ applied to the latter we get our $\np$
    bound.
\end{proof}

To conclude the proof of Proposition~\ref{prop:buchiexactconcept}, it then
suffices to show that any automatic structure (like B\"uchi arithmetic) has
$\pwp$:

\begin{restatable}{lemma}{LemmaAutomaticPWP}\label{lem:automatic-pwp} If
    $\amodel$ is an automatic structure, then it has~\pwp.
\end{restatable}

This result follows from~\cite[Proposition 4.8]{modeltheoryofstrings}. We include a proof here for completeness.
\begin{proof}
    For simplicity, we consider a formula $\phi(x \varsep y)$ with a single
    object variable and a single parameter variable; the generalization to $\vec x$ 
    and $\vec y$ will be clear. Suppose $\phi$ is a formula over the automatic
    structure $\amodel$. We can thus convert $\phi(x \varsep y)$ into a
    deterministic finite automaton $A$ over the alphabet $\Sigma_\Box \times
    \Sigma_\Box$ where $\Sigma_\Box \coloneqq \Sigma \cup \{\Box\}$. The
    automaton reads convolutions of two words, one for $x$ and one for $y$. We
    assume $A$ to be minimal, and let $Q$ denote its states. Given~$q \in Q$, we
    write $A_q(x, y)$ for the state reached by running $A$ from $q$ on the
    convolution $x \otimes y$. When $q$ is the initial state~$q_0$, we simply
    write $A(x,y)$. We encode acceptance through a function $F \colon Q \to
    \mathbb{B}$: the automaton accepts $x \otimes y$ whenever $F(A(x,y)) =
    \textbf{true}$.

    To establish \pwp, we must show that there is a polynomial $f$ such that the
    following holds: for every sequence $x_1,\dots,x_k$, if some $y$ witnesses
    that $F((A(x_i,y))) = \textbf{true}$ for all $i \in [1..k]$, then at least
    one such witness $y$ has binary length at most~$f(\sum_i\bitsize{x_i})$. We
    proceed by fixing a sequence $x_1,\dots,x_k$ for which a witness exists. We
    can assume that the witness has length at least $\max_i\bitsize{x_i}$; if
    such a witness does not exist (no matter the sequence), then taking $f$ to
    be the identity function suffices. Without loss of generality, we also assume the $x_1,\dots,x_k$
    to be padded with $\Box$ to a common length~$\ell$. Consider a word $y \in
    \Sigma_{\Box}^\ell$ that can be extended into a witness. Let
    $(q_1,\dots,q_k) \coloneqq (A(x_1,y),\dots,A(x_k,y))$. Since $x_1,\dots,x_k$
    have length $\ell$, in order to complete $y$ into a witness it suffices to
    find a short word $z \in (\{\Box\} \times \Sigma_\Box)^*$ such that $F(A_{q_1}(z)) = \dots = F(A_{q_k}(z)) = \textbf{true}$. 
    
    Suppose $z = z_1 \dots z_m$ to be a word with the above property, of minimal length.
    Consider the (only) run of $z$ on the elements of the tuple $(q_1,\dots,q_k)$: we get a sequence of tuples of states, starting from $(q_1,\dots,q_k)$ and ending in a tuple of accepting states:
    \[ 
    (q_1,\dots,q_k) \xrightarrow{z_1} (q_{11},\dots,q_{k1}) \xrightarrow{z_2} \dots \xrightarrow{z_m} (q_{1m},\dots,q_{km}),
    \]
    We show that $m < 2^h$, where $h$ is the number of states in the automaton~$A$.
    For a contradiction, suppose $m \geq 2^n$. 
    Abstract each tuple in the run to the set of its states: 
    \[ 
    \{q_1,\dots,q_k\} \xrightarrow{z_1} \{q_{11},\dots,q_{k1}\} \xrightarrow{z_2} \dots \xrightarrow{z_m} \{q_{1m},\dots,q_{km}\}.
    \]
    Since we are assuming $m \geq 2^h$, here are two positions $i < j$ 
    such that the sets $\{q_{1i},\dots,q_{ki}\}$ and $\{q_{1j},\dots,q_{kj}\}$ are the same. Consider then the word $z' = z_1 \dots z_i z_{j+1} \dots z_m$: it can be verified that $z'$ is also a word that extends $y$ into a witness, contradicting the minimality of $z$.
    
    Since $\phi$ is fixed so is $h$, and taking $f(\ell) \coloneqq \ell + 2^h$
    concludes the proof. 
\end{proof}

\section{Proof of Theorem \ref{thm:buchi}: for automatic structures, finite VC dimension implies tractable approximate fitting} 
\label{app:buchinipptime}

Recall that in the body we outlined a proof of Theorem~\ref{thm:buchi}, saying
that for partitioned formulas with finite VC dimension and over an automatic
structure, there is a $\ptime$ algorithm for both exact and approximate fitting.
In this appendix we fill in the details of the proof. In the proof of
Theorem~\ref{thm:buchi} we will also use a basic result about separating
examples for a family of finite sets.

\begin{lemma}\label{lem:small-representative-set}
    Consider $\ell+1$ pairwise different sets $S_1,\dots,S_{\ell+1}$. At most
    $\ell$ representatives $r_1,\dots,r_\ell$ are needed to distinguish these
    sets: for every $i \neq j \in [1..\ell+1]$, there is $k \in [1..\ell]$ such
    that $r_k$ belongs to the symmetric difference $S_i \triangle S_j$ of $S_i$
    and $S_j$.
\end{lemma}

\begin{proof}
    By induction on $\ell$. 
    \begin{description}
        \item[Base case: $\ell = 0$.] For a single set, no representative is needed.
        \item[Inductive step: $\ell \geq 1$.] Since $S_1 \neq S_2$, the
            symmetric difference $S_1 \triangle S_2$ is non-empty. Pick $r \in
            S_1 \triangle S_2$. Define $P \coloneqq \{S_i : r \in S_i\}$ and $N
            \coloneqq \{S_i : r \not\in S_i\}$. Let $p \coloneqq |P|$ and $n
            \coloneqq |N|$. Observe that $P$ and $N$ partition
            $\{S_1,\dots,S_{\ell+1}\}$, and since $r$ distinguishes $S_1$ from
            $S_2$, we have $p \geq 1$ and $n \geq 1$. By the induction
            hypothesis, $p-1$ representatives suffice to distinguish all sets in
            $P$, and $n-1$ representatives suffice to distinguish all sets in
            $N$. Together with $r$, we have $1 + (p-1) + (n-1) = p + n - 1 =
            \ell$ representatives. These suffice to distinguish all sets
            $S_1,\dots,S_{\ell+1}$.
            \qedhere
    \end{description}
\end{proof}
We are now ready to prove Theorem \ref{thm:buchi}, which we recall:

\thmbuchi*

\begin{proof} 
Following Lemma~\ref{lem:FiniteVCandPTimeFittingImpliesPtimeApprox}, it suffices to prove \ptime \mbox{$0$-fitting}.
For simplicity, we consider a partitioned formula $\phi(x \varsep y)$ with a single object variable
and a single parameter variable; the generalization to $\vec x$ and $\vec y$ will be
clear.

Since $\amodel$ is automatic, we can convert $\phi(x \varsep y)$ into a DFA
$A_\phi$ over the alphabet $\Sigma_\Box \times \Sigma_\Box$ where $\Sigma_\Box
\coloneqq \Sigma \cup \{\Box\}$. The automaton reads convolutions of two words,
one for $x$ and one for $y$. We assume $A_\phi$ to be minimal, and let $Q_\phi$
be its states, $i_\phi$ its initial state, and $\delta_\phi \colon Q_\phi
\times (\Sigma_\Box \times \Sigma_\Box) \to Q_\phi$ its transition relation. We
write ${\delta_{\phi}}\!^* \colon Q_\phi \times (\Sigma_\Box \times
\Sigma_\Box)^* \to Q_\phi$ for its extended transition relation. We encode
acceptance through a function $F_\phi \colon Q_\phi \to \mathbb{B}$: the
automaton accepts $x \otimes y$ whenever $F_\phi({\delta_{\phi}}\!^*(i_\phi,x
\otimes y)) = \textbf{true}$.

Parameters range over the regular language~$D \cdot \{\Box\}^*$. To enforce this, the algorithm relies on the minimal DFA $A_D$ for this language. Let
$Q_D$ denote its states, $i_D$ its initial state, $\delta_D$ its transition
relation, and~$F_D$ its acceptance function.

Let $S \coloneqq (s_1,b_1),\dots,(s_n,b_n)$ be an input sample, 
where each $(s_i,b_i)$ belongs to $D \times \mathbb{B}$.
Denote by $m$ the maximum length of the words $s_1,\dots,s_n$.
Any automatic structure has \pwp (see~Lemma~\ref{lem:automatic-pwp} in Appendix~\ref{app:buchinpmembership}), so it suffices to
consider parameters $y$ of length at most $f(n \cdot m)$, for some polynomial $f$ independent of
 $S$. %
Moreover, since $\amodel$ is automatic, 
we may pad all words $s_1,\dots,s_n$ with~$\Box$, 
to a common length $\ell \geq \max(m,f(n \cdot m))$, 
without affecting the set of parameters we need to consider.

A tuple $(q_1,\dots,q_n)$ of states is \emph{$S$-reachable~at~$k \in [0..\ell]$}
whenever there is $y \in (\Sigma_\Box)^k$ such that $q_j =
{\delta_\phi}\!^*(i_\phi,(s_j^k,y))$ for all $j \in [1..n]$, where $s_j^k$ denotes the
prefix of $s_j$ of length~$k$. 

The algorithm for the $0$-fitting problem of $\phi$ is simple (below, $s_{ij}$
stands for the $j^{th}$ letter of $s_i$):
\begin{center}
\begin{minipage}{\linewidth}
\begin{algorithmic}[1]
    \State $V \gets \{(i_D,i_\phi,\dots,i_\phi)\}$
    \Comment{$(i_\phi,\dots,i_\phi)$ only tuple $S$-reachable at $0$}
    \For{$j = 1$ to $\ell$}\label{aut:main-loop}
        \State $V \gets \begin{aligned}[t]
                            \{ \big(\delta_D(p,y),\delta_\phi(q_1,(s_{1j},y)),\dots,\delta_\phi(q_n,(s_{nj},y))\big)\\ :
                             y \in \Sigma_\Box \text{ and } (p,q_1,\dots,q_n) \in V \}
                        \end{aligned}$
    \EndFor
    \Statex \Comment{postcondition: when projecting away its first component, $V$ becomes the set of all tuples that are $S$-reachable at $\ell$}
    \For{$(p,q_1,\dots,q_n) \in V$ with $F_D(p) = \textbf{true}$}
        \If{$F_{\phi}(q_i) = b_i$ for all $i \in [1..n]$}
            \textbf{return true}
        \EndIf
    \EndFor
    \State \textbf{return false}
\end{algorithmic}
\end{minipage}
\end{center}
Briefly, the algorithm first builds a set $V$ that, after the $j^{th}$ iteration of the
loop in line~\ref{aut:main-loop}, contains all the tuples $(p,q_1,\dots,q_n)$
with $(q_1,\dots,q_n)$ $S$-reachable at $j$ via parameter ${y \in
(\Sigma_\Box)^j}$ with $\delta_D^*(i_D,y) = p$. Then for each such tuple
with $p$ a final state, it checks whether the sample has been fit exactly; 
this happens if and only if $F_\phi(q_i) = b_i$ for every $i \in [1..n]$.

The correctness of the algorithm is straightforward. 
For polynomial runtime, it suffices to show that the set $V$ remains of polynomial size at each iteration of the
for loop in line~\ref{aut:main-loop}. 
Since the number of states of the automaton $A_D$ is fixed, 
this follows directly from the following claim.

\begin{restatable}{claim}{ClmPolyReachable}\label{clm:polyreachable} There is a polynomial $g$ independent of $S$ such that the number of $S$-reachable sequences at any $k$ at most~$g(n)$.
\end{restatable}
For a proof of this claim, fix $k \in [0..\ell]$ and let~$(t_1,\dots,t_n) \coloneqq
(s_1^k,\dots,s_n^k)$. 
Consider the function $\alpha \colon (\Sigma_\Box)^k \to Q^n_\phi$ given by
\[
    \alpha(y) \coloneqq \big({\delta_\phi}\!^*(i_\phi,(t_1,y)),\dots,{\delta_\phi}\!^*(i_\phi,(t_n,y))\big). 
\]    
We show that the image $\Img(\alpha)$ of $\alpha$ has at most polynomially
many elements with respect to $n$. 

Fix a set $W \coloneqq \{w_1,\dots,w_h\}$ of convolutions
$w_i \coloneqq u_i \otimes v_i$ that are sufficient to discriminate all states
in the automaton $A$: for every two states $p,q \in Q_\phi$, there is $j \in [1..h]$
such that $F({\delta_\phi}\!^*(p,w_j)) \neq F({\delta_\phi}\!^*(q,w_j))$. By~Lemma~\ref{lem:small-representative-set}, 
such a set can be taken so that~${h < \card{Q_\phi}}$. Without loss of generality,~we
assume $w_1,\dots,w_h$ are padded with~$\Box$ to a common length.

We define $\tau \colon Q_\phi^n \to \mathbb{B}^{n \cdot h \cdot h}$:
\[ 
    \tau(q_1,\dots,q_n) \coloneqq \langle F({\delta_\phi}\!^*(q_i,(u_j, v_u))) \colon i \in [1..n], j,u \in [1..h] \rangle.
\]
From the definition of $W$, it follows that $\tau$ is an injection.
Hence, $\Img(\alpha)$ has the same cardinality as $\Img(\tau \circ \alpha)$. 
For $y \in (\Sigma_\Box)^k$, observe that $(\tau \circ \alpha)(y)$ 
can be alternatively written as the product (i.e., the concatenation) of tuples
\begin{align}
    \!\!\prod\nolimits_{u = 1}^h \!\langle F({\delta_\phi}\!^*(i_\phi,(t_i \cdot u_j, y \cdot v_u))) \colon i \in [1..n], j \in [1..h] \rangle. 
    \label{ali:tau-tuples}
\end{align}
Let $T \coloneqq \{t_i \cdot u_j : i \in [1..n], j \in [1..h]\}$. The $u^{th}$ tuple in Equation~\ref{ali:tau-tuples} asserts, for every $t \in T$, 
whether the state ${{\delta_\phi}\!^*(i_\phi,(t_i \cdot u_j, y \cdot v_u))}$ 
is final. So, it can equivalently be represented as 
the set $\{t \in T : F({\delta_\phi}\!^*(i_\phi,(t,  y \cdot v_u))) = \textbf{true}\}$.
Then, $\Img(\tau \circ \alpha)$ can be alternatively seen as a subset of  
\[
    \left(\{\{t \in T : F({\delta_\phi}\!^*(i_\phi,(t, z))) =
    \textbf{true}\} : z \in (\Sigma_\Box)^* \}\right)^h.
\]
The fact that the formula
 $\phi(x \varsep y)$ has finite VC dimension implies we can fix $d$ as the VC-dimension of the concept class $\{\{t \in
(\Sigma_\Box)^* : F({\delta_\phi}\!^*(i_\phi,(t, z))) = \textbf{true}\} : z \in (\Sigma_\Box)^* \}$. From~Fact~\ref{fact:sslemma}, the cardinality of
$\Img(\tau \circ \alpha)$ is thus bounded by $\big(\sum_{i=0}^{d} {{n \cdot h}
\choose i}\big)^h \leq (3 \cdot n \cdot h)^{d \cdot h}$. Since the formula~$\phi$ is fixed,
so are the integers $h$ and $d$. 
Then, the polynomial~$g(n) \coloneqq (3 \cdot n \cdot
h)^{d \cdot h}$ shows the claim.
\end{proof}

\section{More details on Example \ref{ex:expandablerqc}: finite VC dimension but not $\rqfo$ fitting}\label{app:fragile}

In Example~\ref{ex:expandablerqc}, we mentioned that there are structures that
lack $\rqfo$ fitting, but where this could be repaired by expanding the
structure. In particular, we mentioned that this can be done for the equivalence
relation structure from Example \ref{ex:rqc}. Similar statements are made
in~\cite{mbehlics}, but not about fitting, and without full proofs. 
Below, we expand on Example~\ref{ex:expandablerqc}.

Let us consider~$\prespower$ from~Example~\ref{ex:nip}, i.e., the structure
$(\nats, +, <, 2^{\N}(x))$, where $(\nats,+,<)$ is Presburger arithmetic and
$2^{\N}(x)$ is the predicate that holds on powers of two. This structure is
known to have quantifier elimination in the expanded signature that adds the
function $\lambda$ mapping natural number to the largest power of two below
it~\cite{prespowernip}: $\lambda(n) \coloneqq 2^k$ such that $2^k \leq n <
2^{k+1}$, and $\lambda(0) = 0$. Together with the fact that it is possible to
evaluate in polynomial time the truth of any quantifier-formula (for a given
assignment of the variables), quantifier elimination implies $\ffte$. 

We claim that this structure is $\rqc$ hence has $\rqfo$ fitting by
Corollary~\ref{cor:rqcimpliesrqfoandptimefitting}. We proceed using the notion
of \emph{distality}, which dates to \cite{isolation}, where it is termed
``isolation''. 

Given a finite set $F$ of elements from the domain of a
structure~$\amodel$, and a variable-partitioned formula $\phi(x_1 \ldots x_j
\varsep  y_1 \ldots y_k)$ over~$\amodel$, a \emph{$\phi$-type} is a function
assigning to each $k$-tuple $\vec f$ from $F$ exactly one of the assertions
$\phi(\vec x, \vec f)$ and $\neg \phi(\vec x, \vec f)$. We will abuse notation
by identifying a type with the conjunction of formulas in its range. A
$\phi$-type over $F$ is \emph{consistent} if it is the type of some $\vec x_0$.
Given a $j$-tuple $\vec x_0$ in $\amodel$ its $\phi$-type over $F$ is the set of
formulas of the above form that hold of $\vec x_0$ in $\amodel$. An $L$ formula
$\eta(\vec x, \vec f)$ with parameters $\vec f$ from $F$ is said to
\emph{isolate} $\phi$-type $t( \vec x)$ over $F$ whenever $\amodel \models
\eta(\vec x, \vec f) \rightarrow t(\vec x)$. Such a formula $\eta$ is sometimes
called a \emph{strong honest definition} \cite{simon-chernikov-two,tongdistal}.

A structure $\amodel$ is \emph{distal} if for every $\phi(\vec x \varsep \vec y)$ 
over $\amodel$ there is $\eta(\vec x \varsep \vec u)$ over $\amodel$ such
that, for every finite set $F$, for every $\phi$-type $t(\vec x)$ over $F$ there
is $\vec f \in F$ such that $\eta(\vec x, \vec f)$ isolates $t(\vec x)$. 

The following facts are known:

\begin{fact}[\cite{isolation}] If a structure is distal, then it is $\rqc$.
\end{fact}

\begin{fact}[\cite{tongdistal}] $\prespower$ is distal.
\end{fact}

Therefore,~$\prespower$ is~\rqc, and has~\rqfo fitting (by
Corollary~\ref{cor:rqcimpliesrqfoandptimefitting}).

Consider now the structure on the natural numbers with only a binary relation
$E(x,y)$, interpreted as the equivalence relation that holds of $n, n'$ exactly
when $\lambda(n)=\lambda(n')$. This is an equivalence relation with classes of
unbounded size. Using this we can easily see that the formula $\phi(x \varsep
y)$ given by $E(x,y) \wedge y \neq x$ does not have an $\rqfo$ fitting problem.
The reason is that, given a sample of positive examples~$S^+$ only, the exact
fitting problem for $\phi$ asks to decide whether $S^+$ is a complete
equivalence class. However, $\rqfo$ sentences can quantify only over the sample
itself, making it impossible to assert the existence of elements outside the
sample ---at least,  not when the only relation of the structure is $E$. This
can be formally shown using a standard Ehrenfeucht-Fra\"iss\'e game argument,
proving that no \rqfo sentence can distinguish positive examples containing an
entire $E$ class from the ones that do not. But as explained in
Example~\ref{ex:expandablerqc}, this structure can be expanded to~\prespower,
and thus it has $\ptime$ approximate fitting.

\forarxiv{\section{Variant of Example \ref{ex:expandablerqc}: Fitting for logics between
$\ptime$ and $\rqfo$} \label{app:lfp}

Recall that we studied $\rqfo$ fitting, which implies $\ptime$ fitting, but is
strictly stronger: in Example~\ref{ex:expandablerqc} we presented a structure
that had $\ptime$ fitting but not $\rqfo$ fitting. 
One could try to address this gap by  capturing fitting methods in logics over
the sample that are richer than $\rqfo$,  logics that are contained in $\ptime$,
such as transitive closure logic or least fixpoint logic.  This will not help
with the last example from Example \ref{ex:rqc}, since it is not expressible in
any traditional logic over the active domain. 

There are other examples where one can use larger
$\ptime$ logics to solve the fitting problem, but where we still do not have
$\rqfo$ fitting. 

\begin{proposition} There is a structure~$\amodel$ with $\nip$ and a variable
partitioned  formula $\varphi(\vec x \varsep \vec y)$ that does not have $\rqfo$
fitting, but where the fitting problem for $\varphi$ is not only in $\ptime$,
but is expressible as a fixpoint logic formula over the sample (using
$L$-formulas and relations $S^+$, $S^-$ as atoms).
\end{proposition}
 
\begin{proof} We use a variation of an example from \cite{mbehlics}. The only
thing we will use about least fixpoint logic is that it can compare the
cardinality of two unary predicates in an ordered structure. This follows from
the Immerman-Vardi theorem \cite{immermanvardi}, but it can also be shown
through simple programming.

We now describe the structure $\amodel$. It includes two equivalence relations
$E$ and $F$, a partial order $\preceq$, and unary predicates $Min_F$ and
$Max_F$. The equivalence relation~$E$ has classes $E_n$ of unbounded finite size
as $n \geq 1$ increases. Each equivalence class for $E$ is subdivided into
classes of the second equivalence relation $F$. The class $E_n$ is subdivided
into $F$ classes of size $1 \ldots n$. Therefore, the total size of $E_n$ is
$\Sigma_{i \leq n} i$. Every equivalence class in $F$ is linear ordered by
$\preceq$. Given $n \in \N$ positive, and $j \in [1..n]$, we let $F_{nj}$ denote
the unique $F$-class contained in the $E$-class $E_n$ with cardinality $j$. The
predicate $Min_F$ holds of elements that are in the smallest cardinality
$F$-class within their $E$-class: that is, the elements in $F_{n1}$ for some $n
\geq 1$. And similarly $Max_F$ holds of elements that are in $F_{nn}$ (the
largest $F$-class in $E_n$), for some $n \geq 1$.

Now consider $\phi(x \varsep y)= E(y, x) \wedge y \neq x$. Exactly as for the
structure described at the end of Appendix~\ref{app:fragile}, a standard
Ehrenfeucht-Fra\"iss\'e game argument shows that $\phi$ does not have $\rqfo$
fitting. Intuitively, it is indeed again not possible to write an \rqfo sentence
that distinguishes samples with positive examples containing an entire $E$-class
from the ones that do not.

We next claim that the fitting problem for $\phi$ can be expressed in least
fixedpoint logic. Given a sample consisting of positive examples~$S^+$ and
negative examples~$S^-$, let us characterize when it has a fitting. For there to
be a fitting, it must be that $S^+$ contains only examples from one $E$-class.
The set $S^-$ must be disjoint from $S^+$, and can contain at most one element
of the $E$-class that includes $S^+$. All these properties can be checked with
an $\rqfo$ sentence. Assuming they hold, then there is a way to fit the sample
if and only if the positive examples do not cover an entire $E$-class. If this
holds because there is a (single) element in $S^-$ from that $E$-class, then it
can be detected with an \rqfo formula. The same is true if the positive examples
contain only part of an $F$-class: the \rqfo formula asks if there is an $x \in
S^+$ such that the $\preceq$ successor of $x$ exists and is not in $S^+$, or
such that the $\preceq$ predecessor of $x$ exists and is not in $S^+$.
Therefore, we can assume that $S^-$ does not contain any element from the
$E$-class of $S^+$, and that $S^+$ is \emph{$F$-saturated}: it contains either
all of an $F$-class or none of it.

The remaining possibility is more problematic: there is an entire $F$-class
missing. If it is the min or the max class, we can check this using $Min_F$ or
$Max_F$. Otherwise, for some $n \geq 1$, we have the entire class $F_{nj}$, the
entire class $F_{nk}$ for $k>j+1$, but we are missing some intermediate class
$\ell$ with $j<\ell<k$. Note that $F_{nj}$ has size $j$ and, by construction,
$F_{nk}$ has size $k>j+1$. Thus it suffices to write a formula
$\phi_{\textbf{gap}} (x,y, P)$ (where $P$ is to be interpreted as $S^+$) that
holds on $x,y \in S^+$ exactly when the intermediate cardinalities for the
$F$-classes that are represented in $S^+$ never lie strictly between the
cardinality of the $F$-classes of $x$ and the $F$-classes of $y$. Since the
$F$-classes of $S^+$ are $F$-saturated, this holds of $x, y \in S^+$ exactly
when:
\begin{align*}
\forall z \in S^+  ~ \neg (  ~  | ~ \{  ~ z' \in S^+ ~ | ~ F(x,z') ~ \}  ~ | <  |\{ ~ z' \in S^+ ~| ~ F(z,z') ~ \} ~ | < | ~ \{  ~ z' \in S^+ ~ | ~ F(y,z') ~ \} ~ | ~)
\end{align*}

But since the $F$-classes are ordered by $\preceq$ we can perform this
computation in fixed point logic.
\end{proof}

Note that, as with the example from Appendix \ref{app:fragile}, this structure
is expandable to one with $\rqfo$ fitting, just by extending $\preceq$ to a
global order in which every $E$ class is an interval. We also acknowledge that
the structure in the example is constructed expressly for this purpose. We are
not aware of any natural structure where the exact fitting problem for definable
concept classes can be solved by reducing it to a sentence in fixpoint logic,
except in cases where the reduction follows trivially from $\rqfo$ fitting.
}

\section{Proof of Lemma~\ref{lemma:piecewise-functions-reduction}: Reducing
approximate fitting for piecewise functions to deciding an existential
formula}\label{app:lemmapiecewisefunctionreduction}

We recall the statement of Lemma~\ref{lemma:piecewise-functions-reduction} from
Section~\ref{sec:functions}:

\piecewisefunctionreduction*

\begin{proof}
    Let $\phi(\vec x \varsep z \varsep \vec y) \coloneqq {\bigvee_{i=1}^m
    (\psi_i(\vec x, \vec y) \land z = \ell_i(\vec x, \vec y))}$, 
    and $S = (\vec s_1,t_1),\dots,(\vec s_n,t_n)$.

    Let $\gamma(\vec x,z,u,v \varsep \vec y)$ be the partitioned formula 
    \begin{equation}\label{formula:for-nip-step}
        \bigvee\nolimits_{i=1}^m \bigvee\nolimits_{j \in \{0,1\}} (u = i \land v = j \land \psi_i(\vec x, \vec y) \land (\ell_i(\vec x, \vec y) \sim_j z)).
    \end{equation}
    over $\amodel$, 
    where $\sim_0$ denotes $\geq$ and $\sim_1$ denotes $<$. As $\amodel$ is~\nip
    and $\gamma$ has \ptime $0$-fitting, by Lemma~\ref{lem:computealltypes}, we
    can compute in polynomial time the set of all the $\gamma$ realizable
    partitions for~$\{((\vec s_k,t_k,i,j) : k \in [1..n], i \in [1..m], j \in
    \{0,1\})\}$. Each realizable partition splits the set into  positive
    examples and  negative examples. Let $\mathcal{T}$ be the set of all sets of
    positive examples of these partitions.

    Consider a set $T$ of positive examples from $\mathcal{T}$. Note that, for every
    $k \in [1..n]$, there is exactly one $i \in [1..m]$ and $j \in \{0,1\}$ such
    that $(\vec s_k,t_k,i,j)$ belongs to $T$. Define $\gamma^*$ as
    \[ 
        \gamma^*(\vec y) \coloneqq \bigvee\nolimits_{T \in \mathcal{T}} \bigwedge\nolimits_{(\vec s,t,i,j) \in T} (\psi_i(\vec s, \vec y) \land (\ell_i(\vec s, \vec y) \sim_j t)).
    \]
    By definition of $\mathcal{T}$, we have $\amodel \models \forall \vec y
    \,\gamma^*(\vec y)$. Therefore, the hypothesis class $\hypoclass_{f_\phi}$ \mbox{$\epsilon$-fits} $S$
    if and only if $\exists \vec y\, (\gamma^*(\vec y) \land \sum_{i=1}^n
    |f_\phi(\vec s_i,\vec y) - t_i| \leq \epsilon)$. By distributing the
    summation over the disjunctions of $\gamma^*$, and using the signs given by
    the constraints~${\ell_i(\vec s, \vec y) \sim_j t}$, we can remove all
    absolute values, obtaining the equivalent formula $\exists \vec y \phi'(\vec y)$, where
    \begin{align*}
        \phi'(\vec y) \coloneqq \bigvee\nolimits_{T \in \mathcal{T}} \big(&\,\epsilon \geq {\textstyle\sum\nolimits_{(\vec s,t,i,j) \in T} (-1)^j \cdot (\ell_i(\vec s, \vec y) - t)} \land \bigwedge\nolimits_{(\vec s, t, i,j) \in T}(\psi_i(\vec s, \vec y) \land \,\ell_i(\vec s, \vec y) \sim_j t\,)\big).
    \end{align*}
    The algorithm returns~$\phi'$.
\end{proof}

\section{Proof of Theorem~\ref{thm:approx-fitting-piecewise-reals-and-PA}:
\ptime approximate fitting for piecewise function classes}
\label{app:approximatefittingpiecewise} %
Recall that a formula $\phi(\vec x \varsep z \varsep \vec y)$ is called \emph{piecewise} 
whenever it defines a piecewise function~$f_\phi$, as defined in Section~\ref{subsec:approximate-piecewise}. 
We recall the statement of
Theorem~\ref{thm:approx-fitting-piecewise-reals-and-PA}:

\ThmApproxFittingPiecewiseRealAndPa*

\myparagraph{Case: Real ordered field}
We start by proving the last statement of the theorem. Let $\hypoclass_{f_\phi}$
be a piecewise hypothesis class, where $\phi(\vec x \varsep z \varsep \vec y )$
is a formula over the real ordered field of the form
\begin{equation}
    \bigvee\nolimits_{i=1}^m \psi_i(\vec x, \vec y) \land z = \ell_i(\vec x, \vec y).
\end{equation}
Since the real ordered field has quantifier elimination (in the
structure~$(\R,0,1,+,\cdot,<)$), we can assume without loss of generality that
$\phi$ is quantifier-free. 

For a given sample $S$ and a tolerance $\epsilon$, by
Lemma~\ref{lemma:piecewise-functions-reduction}, we can compute in polynomial
time a formula $\phi'(\vec y)$ such that $\hypoclass_{f_\phi}$
\mbox{$\epsilon$-fits} $S$ if and only if $\exists \vec y\,\phi'(\vec y)$ is
true over the real ordered field. Since $\phi$ is quantifier-free, so is
$\phi'$. Moreover, $\phi'$ has a fixed number of variables. The statement then
follows from the fact that satisfiability of existential formulas with a fixed
number of variables can be decided in polynomial time over the real ordered
field. More precisely, Theorem 1.1 of~\cite{Renegar92} shows that deciding
whether a sentence from the real ordered field is valid can be done in time
$L^{O(1)} \cdot (md)^{2^{O(\omega)}\cdot \Pi_k n_k}$ where:
\begin{itemize} 
\item $L$ is the maximum bit size of coefficients, $m$ is the number of atomic
formulas, and $d$ is the maximum degree of the polynomials (in our case, these
quantities are polynomial in the size of the sample and on the tolerance),
\item $\omega$ corresponds to the number of quantifier alternations (in our
case, $\omega = 1$), and
\item $n_k$ is the number of variables in the $k^{th}$ quantifier block (in our
case, $n_1$ is the number of variables in $\vec y$, which is fixed).
\end{itemize}

\myparagraph{Case: Real ordered group}
Consider now a function class $\hypoclass_{f_\phi}$, where $\phi$ is a
formula over the real ordered group. The real ordered group has quantifier
elimination, in the signature for $0,1,+,<$, %
so we can assume without loss of
generality that $\phi$ is quantifier-free, and of the form 
\[ 
    \bigvee\nolimits_{i=1}^m \gamma_i(\vec x \varsep z \varsep \vec y), 
    \quad \text{where}  \quad \gamma_i(\vec x \varsep z \varsep \vec y) \coloneqq 
    \Big(
        {\textstyle 
        A_i \cdot 
        \begin{bmatrix}
            \vec x \\
            z \\
            \vec y
        \end{bmatrix} < \vec b_i 
        \land 
        C_i \cdot 
        \begin{bmatrix}
            \vec x \\
            z \\
            \vec y
        \end{bmatrix} \leq \vec d_i}
    \Big),
\]
and for every $i \neq j \in [1..m]$ the formula $\gamma_i \land \gamma_j$ is
unsatisfiable. Here, $A_i$ and $C_i$ are matrices of integer coefficients, and
$\vec b_i$ and $\vec d_i$ are vectors of integer coefficients.

By hypothesis, $\phi$ defines a function from the variables $\vec x$ and $\vec y$ 
to the variable $z$. This implies that, for every $i \in [1..m]$, the formula
$\gamma_i$ implies an equality of the form $\vec g_i \cdot (\vec x, \vec y)
+ h_i \cdot z = k_i$, where $\vec g_i$ is a vector of integer coefficients,
and $h_i > 0$ and $k_i$ are integer coefficients. Adding divisibility by
positive integers to the language (which does not change the class of definable
sets), we can rewrite $\gamma_i$ as 
\[
    {\textstyle 
    A_i' \cdot 
    \begin{bmatrix}
        \vec x \\
        \vec y
    \end{bmatrix} < \vec b_i' 
    \land 
    C_i' \cdot 
    \begin{bmatrix}
        \vec x \\
        \vec y
    \end{bmatrix} \leq \vec d_i' 
    \land z = \frac{\vec g_i \cdot (\vec x, \vec y) + k_i}{h_i}
    },
\]
where $A_i'$, $C_i'$, $\vec b_i'$ and $\vec d_i'$ are obtained from $A_i$,
$C_i$, $\vec b_i$ and $\vec d_i$ by substituting $z$ with $\frac{\vec g_i \cdot (\vec x, 
\vec y) + k_i}{h_i}$ and suitably multiplying the resulting entries to remove
denominators. We thus have translated $\phi$ into an equivalent piecewise
formula. The proof then proceeds as in the case of the real ordered field,
observing that the divisibility operations we added can be removed from the formula
$\phi'$ obtained through Lemma~\ref{lemma:piecewise-functions-reduction} by
simply multiplying terms featuring them by a suitable positive integer.

\myparagraph{Case: Presburger arithmetic}
Consider a function class $\hypoclass_{f_\phi}$, where $\phi$
is a formula over Presburger arithmetic. Presburger arithmetic has quantifier
elimination in the structure ${(\mathbb{Z}, 0, 1, +, \{d \mid \cdot\}_{d \geq 1},
<)}$, where $d \mid x$ is a predicate that holds if and only if $x$ is divisible
by $d$. Without loss of generality, we can assume that $\phi$ is
quantifier-free, and of the form 
\[ 
    \bigvee\nolimits_{i=1}^m \gamma_i(\vec x \varsep z \varsep \vec y), 
    \quad \text{where}  \quad \gamma_i(\vec x \varsep z \varsep \vec y) \coloneqq 
    \Big(
        {
        A_i \cdot 
        \begin{bmatrix}
            \vec x \\
            z \\
            \vec y
        \end{bmatrix} \leq \vec b_i 
        \land 
        \bigwedge_{j=1}^{n} d_{i,j} \mid \vec c_{i,j} \cdot (\vec x, z, \vec y)
        } + r_{i,j}
    \Big),
\]
where all variable coefficients and constants are integers, and each $d_{i,j}$
is positive. Furthermore, by taking the least common multiple $M$ across all
constraints involving divisibility operations, and by guessing a remainder modulo $M$ for each
variable, we can assume without loss of generality that each divisibility
constraint $d_{i,j} \mid \vec c_{i,j} \cdot (\vec x, z, \vec y) + r_{i,j}$ is of
the form $M \mid w + r$, where $w$ is a variable and $r \in [0..M-1]$, and that
so that each formula $\gamma_i$ has exactly one divisibility constraint for each
variable in $\phi$.

By hypothesis, $\phi$ defines a function from the variables $\vec x$ and $\vec y$ 
to the variable $z$. This implies that, for every $i \in [1..m]$, $\gamma_i$
implies a (finite) union of equalities of the form~$\bigvee_{j=1}^{t_i} (\vec g_{i,j}
\cdot (\vec x, \vec y) + h_{i,j} \cdot z = k_{i,j})$, where $\vec g_{i,j}$ is a
vector of integer coefficients, and $h_{i,j} > 0$ and $k_{i,j}$ are integer
coefficients.
Observe that each equality $\vec g_{i,j} \cdot (\vec x, \vec y) +
h_{i,j} \cdot z = k_{i,j}$ implies the divisibility constraint $h_{i,j} \mid
\vec g_{i,j} \cdot (\vec x, \vec y) - k_{i,j}$. Again without loss of generality
(as we can further increase $M$ and guess new remainders for each variable), \emph{we
can assume that $h_{i,j}$ is a divisor of $M$}, and so that this divisibility
constraint is implied by the divisibility constraints $\gamma_i$ features on
$\vec x$ and $\vec y$. %

Putting this all together, we can assume $\phi$ to be of the form 
\[ 
    \bigvee\nolimits_{i=1}^{m'} \gamma_i'(\vec x \varsep z \varsep \vec y), 
    \ \ \text{where} \ \ \gamma_i'(\vec x \varsep z \varsep \vec y) \coloneqq 
    \Big(
        A_i' \cdot 
        \begin{bmatrix}
            \vec x \\
            z \\
            \vec y
        \end{bmatrix} \leq \vec b_i' 
        \land 
        \bigwedge_{j=1}^{\ell} M \mid w_j + r_{i,j}'
        \land 
        h_i \cdot z = \vec g_{i} \cdot (\vec x, \vec y) + k_i
    \Big),
\]
where $(w_1,\dots,w_\ell) = (\vec x, \vec y, z)$, and each $h_i$ is a positive
integer that divides $M$.

We now add to the signature an integer divisibility function $\frac{x}{d}$, with
$d$ positive integer, defined as $\frac{x}{d} = q$ if and only if $q \cdot d \leq x <
(q+1) \cdot d$, and rewrite each $\gamma_i'$ as 
\[
    A_i' \cdot 
    \begin{bmatrix}
        \vec x \\
        \frac{\vec g_{i} \cdot (\vec x, \vec y) + k_i}{h_i} \\
        \vec y
    \end{bmatrix} \leq \vec b_i'
    \land 
    \bigwedge_{j=1}^{\ell-1} M \mid w_j + r_{i,j}'
    \land 
    z = \frac{\vec g_{i} \cdot (\vec x, \vec y) + k_i}{h_i}
\]
(or as $\bot$ if $M \mid \vec g_{i} \cdot (\vec x, \vec y) + k_i$ is
inconsistent with the divisibility constraints on $\vec x$ and $\vec y$; in this
case we simply remove $\gamma_i'$ from $\phi$). We thus have translated $\phi$
into an equivalent piecewise formula. Let us stress that we are using
the divisibility function only in cases where we know that the numerator is
divisible by the denominator. 

Let us now move to the approximate fitting problem. Consider a sample $S$ and a
tolerance $\epsilon$. By Lemma~\ref{lemma:piecewise-functions-reduction}, we can
compute in polynomial time a formula $\phi'(\vec y)$ such that
$\hypoclass_{f_\phi}$ \mbox{$\epsilon$-fits} $S$ if and only if $\exists \vec y
\,\phi'(\vec y)$ is true over $(\mathbb{Z}, 0, 1, +, \{d \mid x\}_{d \geq 1}, \{\frac{x}{d}\}_{d \geq 1},
<)$. Following the proof of Lemma~\ref{lemma:piecewise-functions-reduction}, we
see that $\phi'$ a formula of the form
\[
 \phi'(\vec y) \coloneqq \bigvee\nolimits_{T \in \mathcal{T}} 
    \big(\,\epsilon \geq {\textstyle\sum\nolimits_{(\vec s,t,i,j) \in T} (-1)^j \cdot (\ell_i(\vec s, \vec y) - t)} 
    \land \bigwedge\nolimits_{(\vec s, t, i,j) \in T}(\psi_i(\vec s, \vec y) \land \,\ell_i(\vec s, \vec y) \sim_j t\,)\big).
\]
where each $\ell_i(\vec x, \vec y)$ is of the form $\frac{\vec g \cdot (\vec x,
\vec y) + k}{h}$ and each $\psi_i(\vec x, \vec y)$ is of the form 
\[
    A' \cdot 
    \begin{bmatrix}
        \vec x \\
        \frac{\vec g \cdot (\vec x, \vec y) + k}{h} \\
        \vec y
    \end{bmatrix} \leq \vec b'
    \land 
    \bigwedge_{j=1}^{\ell-1} M \mid w_j + r_{j}'.
\]
Since divisibility functions are only used under the assumption that the
numerator is divisible by the denominator, we can eliminate them (in polynomial
time) by multiplying both sides of inequalities by denominators. We next
eliminate the divisibility predicates: for every variable $y$ in $\vec y$, we
consider a fresh variable $q_y$ and replace every divisibility constraint $M
\mid y + r$ with the equality~$q_y \cdot M = y + r$. Let $\vec q$ denote the
sequence all newly introduced variables, and $\phi''$ be the resulting formula.
Note that $\phi''$ is a disjunction where each disjunct is a conjunction of
atomic formulas from~$(\mathbb{Z}, 0, 1, +, \leq)$, i.e., each disjunct is an
integer linear program. 

We have that $\hypoclass_{f_\phi}$ \mbox{$\epsilon$-fits} $S$ if and only if
$\exists \vec y \exists \vec q \,\phi''(\vec y, \vec q)$ is true over
$(\mathbb{Z}, 0, 1, +, <)$. The existential formula $\exists \vec y \exists \vec q 
\,\phi''(\vec y, \vec q)$ has only a fixed number of variables. To solve the
approximate fitting problem, it suffices to find a satisfiable disjunct of
$\phi''$. Since each disjunct is an integer linear program with a fixed number
of variables, we can apply~Lenstra's algorithm~\cite{Lenstra83} to solve this
problem in~\ptime.
 
\section{Proof of Proposition \ref{prop:fitting-real-field-functions}: Sum-of-square-root-hardness for fitting functions over the real ordered field}\label{app:sqrtsum}
\PropFittingRealFieldFunctions*
    
\begin{proof}
    Recall from the body that $\textit{sqrt}(x \varsep z)$ is the formula defining the function that returns the square root of $x$. 
    We show that the reduction from the body of the paper 
    can be modified to use any fixed tolerance $\epsilon > 0$. 
    Let $a_1,\dots,a_n$ and $b$ be the input to the problem. 
    In polynomial time, we can compute a rational approximation~$q_i$ of $\sqrt{a_i}$ 
    from below, with an error within $\frac{\epsilon}{n}$.
    Hence, $\sum_{i=1}^n|\sqrt{a_i}-q_i| = \sum_{i=1}^n \sqrt{a_i} - \sum_{i=1}^n q_i \leq \epsilon$. We can assume that $q \coloneqq \sum_{i=1}^n q_i$ is between $b-\epsilon$ and $b$, as otherwise this instance sum-of-square-root is solved in polynomial time: it is a positive instance if $q < b-\epsilon$, and a negative one if $b < q$.
    We redefine $q_1$ to be equal to $b - \epsilon - \sum_{i=2}^n q_i$. 
    Observe that this decreases $q_1$. 
    Then, $\sum_{i=1}^n|\sqrt{a_i}-q_i| \leq \epsilon$ still holds
    if and only if $\sum_{i=1}^n \sqrt{a_i} \leq b$;
    and to solve the sum-of-square-root problem it
    suffices to check if the square root function given by the formula \textit{sqrt} $\epsilon$-fits the sample~$(a_1,q_1),\dots,(a_n,q_n)$.
\end{proof}

\section{Proof of Proposition~\ref{prop:easydecidabilitydistribution}: decidability for distribution classes over decidable 
structures, and reduction to exact fitting }\label{app:npboundforpwpdistributionclasses} 

We recall that in Section \ref{sec:randomized} we considered the distribution class over a concept class: that is, the real-valued hypothesis class
formed by randomizing a concept class. We stated a proposition showing that the approximate fitting problem for such a distribution class reduces to the $0$-fitting problem for the underlying concept class.

The formal statement was:

\ProprEasyDecidabilityDistribution*

In Section~\ref{sec:randomized}, we gave a proof of this statement, appealing to the fact that a certain system of inequalities with exponentially many variables has a solution with only polynomially many non-zero variables. We now prove this fact in detail, thus completing the proof of Proposition~\ref{prop:easydecidabilitydistribution}.

\begin{proof}
    In the body of the paper, we have shown that the problem reduces to 
    solving a system of inequalities over the reals, of the form:
    \[
        \begin{aligned}
            &\textstyle\sum_{i=1}^n |t_i - \sum_{\psi \in F_i} x_\psi| \leq \epsilon\\
            &\textstyle\sum_{\psi \in F} x_\psi = 1\\
            &x_\psi \geq 0 & \psi \in F
        \end{aligned}
    \]
    This system involves at most $2^n$ variables, indexed by formulas in $F$. To
    complete the proof from the body of the paper, it suffices to show that if
    this system has a solution, then it has one where at most poly($n$) many
    variables are non-zero. Then we can guess these variables, equivalently guessing the corresponding
    non-empty Boolean combinations. We can use a call to exact fitting to verify that each of these guesses
    is correct.

    We start by performing a series of transformations on the system. First, we
    introduce fresh variables $y_i,z_i$ for each $i \in [1..n]$ to eliminate the
    absolute values. Specifically, we replace the constraint $\sum_{i=1}^n |t_i
    - \sum_{\psi \in F_i} x_\psi| \leq \epsilon$ by ${\sum_{i=1}^n z_i \leq
    \epsilon}$ and add, for all~$i \in [1..n]$, the constraints ${(y_i = t_i -
    \sum_{\psi \in F_i} x_\psi)} \land {(y_i \leq z_i)} \land (-y_i \leq z_i)$.
    Next, we convert all inequalities that are not of the form $x \geq 0$ ($x$
    variable) into equalities by introducing slack variables. For example, we
    rewrite $\sum_{i=1}^n z_i \leq \epsilon$ as $\sum_{i=1}^n z_i + \delta =
    \epsilon \land \delta \geq 0$. Finally, we replace all variables $y_i$ (the
    only variables that can take negative values) by differences $y_i^+ -
    y_i^+$, where $y_i^+,y_i^-$ are fresh variables constrained to be
    non-negative. These transformations yield a system of the form 
    $A \cdot \vec x = \vec b \land \vec x \geq 0$, where $\vec b$ is a vector of 
    $3 \cdot n + 1$ entries. This system asks whether 
    $\vec b \in \{A \cdot \vec x : \vec x \geq 0\}$. 
    The result then follows from the Carath\'eodory's theorem for cones: 

    \begin{fact}\label{fact:caratheodory}
        Let $A \in \R^{d \times m}$.
        Every point in ${\{A \cdot \vec x : \vec x \geq 0\}}$ 
        is generated by a vector $\vec x$ with at most $d+1$ non-zero entries.
        \qedhere
    \end{fact}
\end{proof}

\section{Concept fitting reduces to distribution fitting}\label{app:easy-hardness-distributions}

In Section \ref{sec:randomized} in the body of the paper, we mentioned
that for any concept class, the fitting problem for its distribution class
is at least as hard as the fitting problem for the class itself: all lower
bounds are inherited. We prove this here:

\begin{lemma}
    Let $\conceptclass$ be a concept class on range $X$ and parameters $Y$. 
    Let $S \coloneqq (s_1,b_1),\dots,(s_n,b_n)$ be a sample, with $b_1,\dots,b_n \in \{0,1\}$.
    There is a hypothesis $f \in \conceptclass$ that $0$-fits $S$ if and only if 
    there is a hypothesis $h \in \measureclass{\conceptclass}$ that $0$-fits $S$.
\end{lemma}

\begin{proof}
    Let $\conceptclass = \{f_y : y \in Y\}$ and $\measureclass{\conceptclass} = \{h_\mu : \mu \in \Delta(Y_\Sigma)\}$. We remark that $b_1,\dots,b_n$ are probabilities when solving the fitting problem for $\measureclass{\conceptclass}$.
    For the left to right direction, if there is a $y \in Y$ such that $f_y$ 
    $0$-fits $S$, then the distribution assigning probability $1$ to $y$ also $0$-fits $S$.
    For the right to left direction, suppose there is a distribution $\mu \in \Delta(Y_\Sigma)$ 
    such that $h_\mu$ $0$-fits $S$. By definition, this means that 
    for every $i \in [1..n]$, $\mu(\{y \in Y : f_y(s_i) = 1\}) = b_i$.
    Observe that $\mu(\{y \in Y : f_y(s_i) = 1\}) = b_i$ holds if and only if so does 
    $\mu(\{y \in Y : f_y(s_i) = b_i\}) = 1$: 
    for $b_i = 1$ this is trivial, and for $b_i = 0$ we have $\mu(\{y \in Y : f_y(s_i) = 1\}) = 0 \leftrightarrow \mu(\{y \in Y : f_y(s_i) = 0\}) = 1$, since $\conceptclass$ is a concept class. 
    Thus, $\mu$ is such that for every $i \in [1..n]$, $\mu(\{y \in Y : f_y(s_i) = b_i\}) = 1$.
    But then, the set $\{y \in Y : f_y(s_i) = b_i\}$ must be non-empty, 
    and so there is a function in $\conceptclass$ that $0$-fits~$S$.
\end{proof}

\section{Proof of Theorem~\ref{thm:distclassfunctionrcfpspace}: Fitting problems for expectation classes}\label{app:expectation}

In the body of the paper we focused on two ways of defining real-valued function classes using logic: looking at a numerical structure and
using standard definable formulas, and looking at an arbitrary structure and randomizing parameters or objects. In Section~\ref{subsec:expectationclass}, we briefly explain that it is possible to combine these two methods, thus ``randomizing real-valued function classes''. We also mention that many of our results  extend to this setting.  We explain this in detail now.

We recall the definitions of \emph{bounded} real-valued hypothesis class and of \emph{expectation class} from the body of the paper.
A real-valued hypothesis class $\hypoclass \coloneqq \{f_y : y \in \paramspace\}$ with range space $\rangespace$ is said to be \emph{bounded} 
when, for every $x \in \rangespace$, the set $\{f_y(x) : y \in Y\}$ is included in an interval $[a,b] \subseteq \R$.

\DefExpectationClass*

Note that the hypothesis that $\hypoclass$ is bounded is both a necessary and sufficient condition for the expectation $\expectation_{y \sim \mu}[f_y(x)]$
to be well-defined (i.e., for every $x \in \rangespace$, the map~$y \mapsto f_y(x)$ is $\mu$-integrable) on all probability distributions $\mu \in \Delta(Y_\Sigma)$.

We give a reformulation of the approximate fitting problem for
expectation classes, which we  will use to analyze the complexity of the 
problem over the real ordered group and field, and over Presburger arithmetic.
\begin{restatable}{lemma}{LemmaReformulationExpectation}\label{lemma:reformulation-expectation}
    Let $\hypoclass \coloneqq \{f_y : y \in Y\}$ be a hypothesis class with range space $\rangespace$, $S \coloneqq (s_1,t_1),\dots,(s_n,t_n)$ be a sample of pairs from $X \times \R$, 
    and $\epsilon$ be a tolerance. For $i \in [1..n]$, let $f_i$ denote the map $y \mapsto f_y(s_i)$. Then, $\expectationclass{\hypoclass}$
    \mbox{$\epsilon$-fits} $S$ if and only if there are ${y_1,\dots,y_{n+1} \in Y}$ and non-negative $p_1,\dots,p_{n+1} \in \R$ such that 
    \begin{equation}\label{eq:systemExpectReals}
        \sum\nolimits_{i=1}^n |t_i - \sum\nolimits_{j=1}^{n+1} p_j \cdot f_i(y_j)| \leq \epsilon
        \ \land \ \sum\nolimits_{i=1}^{n+1} p_i = 1.
    \end{equation}
\end{restatable}

\begin{proof}
By definition, \mbox{$\epsilon$-fits} $S$ iff $\sum_{i=1}^n\big|t_i-\int_{Y} f_i(y)\,d\mu(y)\big| \leq \epsilon$ for some
distribution $\mu \in \Delta(Y_\Sigma)$.

Consider a distribution $\mu \in \Delta(Y_\Sigma)$.
Since $\hypoclass$ is bounded, every $f_i$ is $\mu$-integrable.
We now rely on the following mean value theorem for integrals: 

\begin{fact}[\cite{meanvalue}]\label{fact:expect-in-convex-hull}
  Let $f_1,\dots,f_n \colon Y \to \R$ be $\mu$-integrable functions. 
  Then, the vector
  \[
    \big(\int_{Y} f_1(y)\,d\mu(y),\dots,\int_{Y} f_n(y)\,d\mu(y)\big)
   \] 
   is in the convex hull of the set $\{(f_1(y),\dots,f_n(y)) : y \in Y\} \subseteq \R^n$.
\end{fact}

From Carath\'eodory's theorem~\cite{caratheodory1907}, any vector $(v_1,\dots,v_n) \in \R^n$ in the convex hull of the set $C \coloneqq \{(f_1(y),\dots,f_n(y)) : y \in Y\}$ can be expressed as a convex combination of at most $n+1$ elements from $C$, that is, $(v_1,\dots,v_n) = \sum_{i=1}^{n+1}p_i \cdot (f_1(y_i),\dots,f_n(y_i))$,
for some $y_1,\dots,y_n \in Y$ and non-negative $p_1,\dots,p_{n+1} \in \R$ with~$\sum_{i=1}^{n+1} p_i = 1$.
Observe that this means that the converse of Fact~\ref{fact:expect-in-convex-hull} also holds: 
there is a probability distribution $\mu^*$ such that 
\[
(v_1,\dots,v_n) = \Big(\int_{Y} f_1(y)\,d\mu^*(y),\dots,\int_{Y} f_n(y)\,d\mu^*(y)\Big).
\]
It is the distribution defined as $\mu^*(y_i) = p_i$ for ${i \in [1..n+1]}$, and $\mu^*(y) = 0$ elsewhere.
By Fact~\ref{fact:expect-in-convex-hull}, we conclude 
that the fitting problem reduces to solving System~\ref{eq:systemExpectReals} 
in the statement of the lemma.
\end{proof}

We are now ready to prove Theorem \ref{thm:approx-fitting-piecewise-reals-and-PA-expect}, which we recall

\approxfittingexpect*

\begin{proof}
The theorem refers to three structures, and we handle each separately.

    \myparagraph{Case: real ordered field}
    For the \textit{real ordered field}, it is easy to encode System~\ref{eq:systemExpectReals} an existential formula over this structure.

    \myparagraph{Case: Presburger arithmetic}
    The $\np$ bound for \textit{Presburger arithmetic} relies on the fact that bounded function families on the integers have finite range, so they are not far from concept classes.

    Let $\hypoclass$ be a bounded function class definable 
    in Presburger arithmetic. Recall that we are looking at the approximate 
    fitting problem for $\expectationclass{\hypoclass}$.
    Since $\hypoclass$ is bounded, the image~$\Img(f_i) \subseteq \N$ of every function in System~\ref{eq:systemExpectReals} is finite. Note 
    that $f_i$ only depends on $s_i$, 
    and by applying well-known bounds on the solutions to integer linear programs (see~\cite[Chapter~17]{linearprogrammingbook}) one can show that the bit length 
    of every natural in $\Img(f_i)$ is polynomial in 
    the bit length of $s_i$.

    The~\np algorithm guesses 
    a polynomial-size natural $k_{ij}$ for each term $f_i(y_j)$ occurring in System~\ref{eq:systemExpectReals}. For each $j \in [1..n+1]$, 
    the algorithm solves the $0$-fitting problem for $\hypoclass$ (in~\ptime, by~Theorem~\ref{thm:approx-fitting-piecewise-reals-and-PA})
    with input the sample $(s_1,k_{1j}),\dots,(s_n,k_{nj})$. 
    This ensures the existence of a $y_j \in Y$ 
    such that $f_i(y_j) = k_{ij}$ for all $i \in [1..n]$. Afterwards, 
    the algorithm updates System~\ref{eq:systemExpectReals}, by replacing each $f_i(y_j)$ by~$k_{ij}$, 
    and removing the absolute value as explained in the proof of Proposition~\ref{prop:easydecidabilitydistribution}. 
    The result is a linear program that 
    can be solved in polynomial time.
    
    \myparagraph{Case: Real ordered group}
    For obtaining an {\np} bound for the \textit{Real ordered group}, the main challenge lies in handling the non-linear terms
    $p_j \cdot f_i(y_j)$ from System~\ref{eq:systemExpectReals}, 
    given that $\Img(f_i)$ can be infinite.
    Fix a partitioned formula $\phi(\vec x \varsep z \varsep \vec y)$ from the real ordered group. 
    We consider the approximate fitting problem for $\expectationclass{\hypoclass_{f_{\phi}}}$.
    In Appendix~\ref{app:approximatefittingpiecewise}, we have noted that every function definable 
    in the real ordered group is piecewise. 
    In fact, wlog., we can assume 
    \[\phi(\vec x \varsep z \varsep \vec y) \coloneqq {\bigvee_{i=1}^m (\psi_i(\vec x,
    \vec y) \land z = \ell_i(\vec x, \vec y))}
    \]
    to be such that, for~$i \in [1..m]$,
    $\psi_i(\vec x \varsep \vec y) \coloneqq A_i \cdot 
        \begin{bmatrix}\vec x\\ \vec y\end{bmatrix} \leq \vec b_i
        \land 
        C_i \cdot 
        \begin{bmatrix}\vec x\\ \vec y\end{bmatrix} < \vec d_i,$
    where $A_i,C_i$ are rational matrices, and $\vec b_i,\vec d_i$ rational vectors, 
    and each $\ell_i$ is a linear polynomial with rational coefficients 
    (note: in Appendix~\ref{app:approximatefittingpiecewise} we used integer coefficients but added a function $\frac{\cdot}{d}$ with $d \geq 1$ integer to the signature; the two approaches are equivalent). 

    Let $(\vec s_1,t_1),\dots,(\vec s_n,t_n)$ be a sample and $\epsilon$ a
    tolerance. By guessing an index $k_{ij} \in [1..m]$ for all $(i,j) \in [1..n] \times [1..n+1]$, the algorithm non-deterministically rewrites
    System~\ref{eq:systemExpectReals}~as
    \begin{equation}\label{eq:systemExpectLinearReals}
        \begin{aligned}
            &\sum\nolimits_{i=1}^n |t_i - \sum\nolimits_{j=1}^{n+1} p_j \cdot \ell_{k_{ij}}(\vec s_i, \vec y_j)| \leq \epsilon\\
            &\sum\nolimits_{j=1}^{n+1} p_j = 1\\
            &p_j \geq 0 & j \in [1..n+1]\\ 
            &\psi_{k_{ij}}(\vec s_i, \vec y_j)&\hspace{-2cm}i \in [1..n], j \in [1..n+1]
        \end{aligned}
    \end{equation}
    I.e., the algorithm guesses which formula among $\psi_1,\dots,\psi_m$ the parameter $\vec y_j$ satisfies, for every object $\vec s_i$ in the sample. 
    Observe that this is not 
    a convex system of constraints, 
    in particular because of the non-linear terms $p_j \cdot \ell_{k_{ij}}(\vec s_i, \vec y_j)$.
    So, we need an ad-hoc argument to show that it can be solved in non-deterministic polynomial time
    instead of in $\exists\R$.

    The next step eliminates all occurrences of~$\psi_{k_{ij}}(\vec s_i, \vec y_j)$ 
    from System~\ref{eq:systemExpectLinearReals}, and simplifies the terms~$\ell_{k_{ij}}$.
    The resulting system is the forthcoming System~\ref{eq:systemExpectLinearRealsTwo}.
    For every~$j \in [1..n+1]$, define $S_j \coloneqq \{\vec y \in \R^d : \bigwedge_{i=1}^n \psi_{k_{ij}}(\vec s_i,\vec y)\}$, 
    and write $\overline{S_j}$ for its \emph{topological closure}, that is,  
    the set obtained by converting all strict inequalities $<$ of the system
    defining $S_j$ to non-strict inequalities~$\leq$. Since $S_j$ is given by a system 
    of strict and non-strict inequalities, $S_j$ is a \emph{semi-open convex polyhedron}, that is, a convex polyhedron that might ``lack'' some of its faces. The set $\overline{S_j}$ is a convex polyhedron.

    Consider the affine transformation $h_j \colon \R^d \to \R^{n+1}$, where $d$ is the number of parameters in $\phi$:
    \[
        h_j(\vec y) \coloneqq \big(\ell_{k_{1j}}(\vec s_1,\vec y), \dots, \ell_{k_{nj}}(\vec s_n, \vec y)\big).
    \]
    Since $\hypoclass$ is bounded, 
    the images $I_j \coloneqq h_j(S_j)$ and $\overline{I_j} \coloneqq h_j(\overline{S_j})$ 
    are bounded sets. Moreover, because $h_j$ is an affine transformation, $I_j$ is a \emph{semi-open convex polytope}, while $\overline{I_j}$ is a \emph{convex polytope}. (A polytope is a polyhedron with no extreme ray.) 
    
    Three observations are in order: 
    \begin{compactenum}
        \item\label{observation1} As $S_j \subseteq \overline{S_j}$, every point in $S_j$ can be obtained by convex combination of some of the vertices of $\overline{S_j}$
        (with some convex combination yielding however points in the difference ${\overline{S_j} \setminus S_j}$). A similar observation holds for $I_j \subseteq \overline{I_j}$ (in fact, $\overline{I_j}$ is the topological closure~of~$I_j$).
        \item\label{observation2} 
        Since $h_j$ is an affine transformation, for every $\vec v \in \overline{I_j}$
        there is an $\vec r \in \overline{S_j}$ such that $h_j(\vec r) = \vec v$.
        \item\label{observation3} 
        Every vertex in both $\overline{S_j}$ and $\overline{I_j}$ is a rational vector (which we encode as a pair of binary integers).
        Since the system defining $\overline{S_j}$ is of size polynomial in the size of the sample, each of its vertices has length polynomial in the sample.
        See~\cite[Theorem~10.2]{linearprogrammingbook} for the precise bounds.
        From the previous Observation~\ref{observation2} and by definition of $h_j$, 
        we see that the same is true for the vertices of $\overline{I_j}$.
    \end{compactenum}
    Using Observation~\ref{observation1} and Carath\'eodory's theorem,
    every point ${\vec w \in S_j}$ is a convex combination of
    \emph{at most} $d+1$ vertices of $\overline{S_j}$. Wlog., this combination
    can be taken  \emph{positive}: there are vertices 
    $\vec v_1,\dots,\vec v_r \in \overline{S_j}$ (with $r \leq d+1$) 
    such that ${\vec w = q_1 \cdot \vec v_1 + \dots + q_r \cdot \vec v_r}$ 
    for  $q_1,\dots,q_r  > 0$ with $\sum_{i=1}^r q_i = 1$.

    For every $j \in [1..n+1]$, the algorithm guesses a linear term
    $\sum_{i=1}^{r_j}q_{ij} \cdot \vec v_{ij}$, where $r_j \leq d+1$, each 
    $\vec v_{ij}$ is a vertex of~$\overline{S_j}$, and each $q_{ij}$ is a fresh
    variable. By Observation~\ref{observation3}, these guesses can be performed
    in non-deterministic polynomial time.

    Let~$T_j\,{\coloneqq}\,\{ \sum_{i=1}^{r_j}q_{ij} {\cdot} \vec v_{ij} :
    \sum_{i=1}^{r_j} q_{ij} = 1, \text{ and } {q_{ij} \,{>}\, 0}
    \text{ for }i
    \in [1..r_j]\}$.
     This set either contains a single point (this happens
    exactly when $r_j = 1$), or it is the relative interior of the convex hull
    generated by $\vec v_{1j},\dots,\vec v_{r_j j}$. Since each $\vec v_{ij}$ is
    a vertex of $\overline{S_j}$, $T_j$ is contained in a face
    of~$\overline{S_j}$. Hence, as $S_j$ is a semi-open convex polyhedron,
    either $T_j \subseteq S_j$ or $T_j \cap S_j = \emptyset$. The algorithm
    verifies $T_j \subseteq S_j$ by testing (in \ptime) feasibility of
    the linear system of inequalities defining $T_j \cap S_j$. If the system is
    infeasible, the algorithm~rejects.
    
    Since $h_j$ is an affine transformation, $h_j(T_j)$ is equal to the set $\{ \sum_{i=1}^{r_j}q_{ij} \cdot h_j(\vec v_{ij}) : \sum_{i=1}^r q_{ij} = 1, \text{and } {q_{ij} > 0}$ $\text{for all $i \in
    [1..r]$}\}$. This allows us
    to rewrite System~\ref{eq:systemExpectLinearReals} as: 
    \begin{equation}\label{eq:systemExpectLinearRealsTwo}
        \begin{aligned}
            &\sum\nolimits_{i=1}^n |t_i - \sum\nolimits_{j=1}^{n+1} p_j \cdot (\sum\nolimits_{k=1}^{r_j} q_{kj} \cdot [h_j(\vec v_{kj})]_i)| \leq \epsilon\\
            &\sum\nolimits_{j=1}^{n+1} p_j = 1\\
            &p_j \geq 0 &\hspace{-2.9cm} j \in [1..n+1]\\            
            &\sum\nolimits_{k=1}^{r_j} q_{kj} = 1 &\hspace{-2.9cm} j \in [1..n+1]\\            
            &q_{kj} > 0 &\hspace{-2.9cm} j \in [1..n+1],\ k \in [1..r_j]\\ 
        \end{aligned}
    \end{equation}
    where $[h_j(\vec v_{kj})]_i$ is the $i^{th}$ entry of the vector $h_j(\vec v_{kj}) \in \R^n$.
    Note that System~\ref{eq:systemExpectLinearRealsTwo} excludes 
    the formulas $\psi_{k_{ij}}(\vec s_i, \vec y_j)$:
    their satisfaction is ensured by $T_j \subseteq S_j$.

    The key observation is that, in~System~\ref{eq:systemExpectLinearRealsTwo}, 
    multiplications among variables are limited to~$p_j \cdot q_{kj}$, where $p_j$ and $q_{kj}$ 
    belong to different convex combinations. These non-linear terms can be linearized 
    by variable substitution.
    For every $j \in [1..n+1]$ and $k \in [1..r_j]$, 
    let $x_{jk}$ be a fresh variable used to proxy the multiplication $p_j \cdot q_{kj}$. One can show that System~\ref{eq:systemExpectLinearRealsTwo} is feasible if and only if the following formula over the real ordered group is satisfiable:
    \begin{equation}\label{eq:systemExpectLinearRealsThree}
        \begin{aligned}
            &\sum\nolimits_{i=1}^n |t_i - \sum\nolimits_{j=1}^{n+1} (\sum\nolimits_{k=1}^{r_j} x_{jk} \cdot [h_j(\vec v_{kj})]_i)| \leq \epsilon\\
            \land&(\sum\nolimits_{j=1}^{n+1}\sum\nolimits_{k=1}^{r_j} x_{jk} = 1)
            \land \bigwedge\nolimits_{j = 1}^{n+1} \bigwedge\nolimits_{k = 1}^{r_j} x_{jk} \geq 0
            \land \bigwedge\nolimits_{j=1}^{n+1}\left((\textstyle\bigvee_{k = 1}^{r_j} x_{jk} = 0) \rightarrow \sum\nolimits_{k=1}^{r_j} x_{jk} = 0\right).
        \end{aligned}
    \end{equation}
    The final constraint in Formula~\ref{eq:systemExpectLinearRealsThree} captures the condition that the variables $q_{1j},\dots,q_{r_j j}$ in System~\ref{eq:systemExpectLinearRealsTwo} form a \emph{positive} convex combination. 
    A solution to System~\ref{eq:systemExpectLinearRealsTwo} where a variable $x_{jk}$ set to zero corresponds to a solution to System~\ref{eq:systemExpectLinearRealsTwo} with $p_j = 0$.

    The algorithm terminates by finding (in non-deterministic polynomial time) whether 
    Formula~\ref{eq:systemExpectLinearRealsThree} has a solution.
\end{proof}

\section{Other loss functions for approximate fitting}\label{app:ltwoloss}

In the body of the paper we have focused on fitting problems defined by using the $L^1$ norm (a.k.a.~the absolute value norm) as the loss function for approximate fitting. Following up on Remark~\ref{remark:otherlosses}, in this appendix we briefly discuss how our results can be adapted to the $L^2$ norm 
$\sqrt{{\sum_{i=1}^n (h(s_i)-t_i)^2}}$ and the hinge loss 
${\sum_{i=1}^n \max(0,1-{t_i {\,\cdot\,} h(s_i)})}$.

\subsection{Approximate fitting problems for concept classes}\label{app:ltwoloss:concepts}

It is rather easy to see that all our results on approximate fitting for concept classes 
(Sections~\ref{sec:fitting-definable-cc} and~\ref{sec:betweenptimeandrqfoconcept}) 
hold also for $L^2$ and hinge loss. Indeed, one can reduce the approximate fitting problem for a concept class under these loss functions to the approximate fitting problem under the $L^1$ norm, and vice versa. Consider a sample $(s_1, b_1),\ldots,(s_n, b_n)$ with $s_i$ from a range space $\rangespace$ and $b_i \in \{0,1\}$, a tolerance $\epsilon \geq 0$, and a hypothesis $h \colon X \to \{0,1\}$. For the $L^2$ norm, we have
\begin{align*}
    \textstyle\sqrt{{\sum_{i=1}^n (h(s_i)-b_i)^2}} \leq \epsilon
        \iff & \textstyle\sum_{i=1}^n |h(s_i)-b_i| \leq \epsilon^2.
\end{align*}
For the hinge loss, without loss of generality assume $b_1,\dots,b_\ell = 0$ and $b_{\ell+1},\dots,b_n = 1$, for some $\ell \in [0..n]$. Then,
\begin{align*}
    \textstyle\sum_{i=1}^n \max(0,1-{b_i {\,\cdot\,} h(s_i)}) \leq \epsilon
    & \iff \textstyle\sum_{i=1}^\ell \max(0,1-{0 {\,\cdot\,} h(s_i)}) + \sum_{i=\ell+1}^n \max(0,1-{1 {\,\cdot\,} h(s_i)}) \leq \epsilon\\
    & \iff \textstyle\sum_{i=\ell+1}^n \max(0,b_i - h(s_i)) \leq \epsilon - \ell\\
    & \iff \textstyle\sum_{i=\ell+1}^n \abs{h(s_i) - b_i} \leq \epsilon - \ell.
\end{align*}

Because of these identities, the notion of \rqfo fitting does not change when considering the $L^2$ norm or the hinge loss. 
Moreover, we can immediately conclude the following: 

\begin{proposition}
    Theorem~\ref{thm:rqfoimpliesptimeapprox}, Theorem~\ref{thm:rqcandrqfitting}, Proposition~\ref{prop:buchiexactconcept} 
    and Theorem~\ref{thm:buchi} hold 
    also for the $L^2$ norm and the hinge loss function.
\end{proposition}

\subsection{Approximate fitting problems for logically-defined real-valued functions}\label{app:ltwoloss:definable-functions}

We now discuss Theorem~\ref{thm:approx-fitting-piecewise-reals-and-PA} in the context of the $L^2$ norm and hinge loss. 
Let us recall the statement.

\ThmApproxFittingPiecewiseRealAndPa*

\begin{remark}\label{remark:what-happened-in-theorem-four}
    Let $\hypoclass_{f_\phi}$ be a hypothesis class, where $\phi(\vec x \varsep z \varsep \vec y )$
    is a formula over a structure from Theorem~\ref{thm:approx-fitting-piecewise-reals-and-PA}. 
    In the proof of the theorem, we have established the following:
    \begin{itemize}
        \item When the structure is the real ordered group, $\phi$ is equivalent to a quantifier-free piecewise formula over~$(\R,0,1,+,\{\frac{x}{d}\}_{d \geq 1},<)$ where $\frac{x}{d}$ stands for division by a positive integer $d$.
        \item When the structure is Presburger arithmetic, $\phi$ can be taken to be a piecewise function over~$(\mathbb{Z}, 0, 1, +, \{d \mid x\}_{d \geq 1}, \{\frac{x}{d}\}_{d \geq 1}, <)$, where a divisibility constraint $d \mid x$ states that $x$ is a multiple of $d$, while $\frac{x}{d}$ denotes \emph{integer} division by a positive integer $d$. 
        By construction of $\phi$, the divisibility operations have two properties. 
        First, every application of~$\frac{t}{d}$ occurs only for terms~$t$ that are guaranteed (by divisibility constraints) to be divisible by $d$. 
        Second, all divisibility constraints are restricted to the form ${M \mid x + r}$, where $x$ is a variable, $M \geq 1$ is a single integer (the same across all divisibility constraints) and is a common multiple of all divisors $d$ appearing in expressions~$\frac{\cdot}{d}$ within the formula, and $r \in [0..M-1]$.
        \item When the structure is the real ordered field \emph{and $\phi$ is piecewise}, $\phi$ is equivalent to a piecewise quantifier-free formula over the language of $(\R,0,1,+,\cdot,<)$. 
    \end{itemize}
\end{remark}

\begin{theorem}\label{thm:approx-fitting-piecewise-reals-and-PA-l2-hinge}
    Theorem~\ref{thm:approx-fitting-piecewise-reals-and-PA} holds also for the $L^2$ norm and the hinge loss function.
\end{theorem}

\begin{proof}[Proof for the $L^2$ norm.]
    Let $\phi(\vec x \varsep z \varsep y) \coloneqq \bigvee\nolimits_{i=1}^m \psi_i(\vec x, \vec y) \land z = \ell_i(\vec x, \vec y)$ be a piecewise formula satisfying the properties described in~Remark~\ref{remark:what-happened-in-theorem-four} (depending on the structure considered). 
    Following the proof of Lemma~\ref{lemma:piecewise-functions-reduction} and relying on~\nip (all structures in the statement of Theorem~\ref{thm:approx-fitting-piecewise-reals-and-PA} are~\nip)
    one can show that for a given sample $S$ and tolerance $\epsilon$, 
    we can compute in polynomial time a formula $\phi'(\vec y)$ of the form 
    \begin{align*}
        \bigvee\nolimits_{T \in \mathcal{T}} \big(&\,\epsilon^2 \geq {\textstyle\sum\nolimits_{(\vec s,t,i,j) \in T} (\ell_i(\vec s, \vec y) - t)^2} \land \bigwedge\nolimits_{(\vec s, t, i,j) \in T} \psi_i(\vec s, \vec y)\big),
    \end{align*}
    with the property that
    $\hypoclass_{f_\phi}$ $\epsilon$-fits $S$ if and only if $\exists \vec y\, \phi'(\vec y)$ is satisfiable. 
    In the cases of the real ordered group and Presburger arithmetic, by multiplying both sides of inequalities by the square of the least common multiple of the divisors appearing in functions
    $\frac{\cdot}{d}$, we can eliminate all such functions from $\phi'$. 
    
    \myparagraph{Case: Real ordered group and field} In these cases, $\ptime$ approximate fitting then follows as in the proof of Theorem~\ref{thm:approx-fitting-piecewise-reals-and-PA}, 
    using the fact that the existential theory of the reals is decidable in $\ptime$ when the number of variables is fixed~\cite{Renegar92}, as explained in Appendix \ref{app:approximatefittingpiecewise}.

    \myparagraph{Case: Presburger arithmetic} 
    Suppose that the structure is Presburger arithmetic.  
    Exactly as in the proof of Theorem~\ref{thm:approx-fitting-piecewise-reals-and-PA}, we further 
    rewrite $\phi'(\vec y)$, adding a variable $q_y$ for every variable $y$, 
    in a way that transforms $\psi_i(\vec s, \vec y)$ into an integer linear program in variables $\vec y$ and $\vec q$ (where $\vec q$ is the vector of all variables $q_y$).
    Checking whether $\phi'(\vec y)$ is satisfiable then boils down to the problem 
    of finding a solution to an integer program of the form 
    \begin{align}\label{eq:convex-quadratic-integer-program}
        \beta \geq \sum_{i=1}^r (q_i(\vec y))^2 \land A \cdot 
            \begin{bmatrix}
                \vec y\\
                \vec q\\
            \end{bmatrix}
            \leq \vec b
    \end{align}
    where $\beta$ is a positive integer, each $q_i$ is a linear polynomial with integer coefficients, $A$ is an integer matrix and $\vec b$ is an integer vector. 
    Observe that the constraint $\sum_{i=1}^r (q_i(\vec y))^2 \leq \beta$ is convex, 
    that is, for every two solutions $\vec y_1, \vec y_2$ to this constraint, 
    and for every $\lambda \in [0,1]$, we have that $\sum_{i=1}^r (q_i(\lambda \vec y_1 + (1-\lambda) \vec y_2))^2 \leq \beta$. Indeed:
    \begin{align*}
        & \sum_{i=1}^r (q_i(\lambda \vec y_1 + (1-\lambda) \vec y_2))^2\\
        {}={}& \sum_{i=1}^r (\lambda q_i(\vec y_1) + (1-\lambda) q_i(\vec y_2))^2 
        & \text{since each $q_i$ is linear}\\
        {}\leq{}& \sum_{i=1}^r (\lambda (q_i(\vec y_1))^2 + (1-\lambda) (q_i(\vec y_2))^2) 
        & \text{by convexity of $x \mapsto x^2$}\\
        {}={}& \lambda \sum_{i=1}^r (q_i(\vec y_1))^2 + (1-\lambda) \sum_{i=1}^r (q_i(\vec y_2))^2\\
        {}\leq{}& \lambda \beta + (1-\lambda) \beta  
        & \text{as $\vec y_1$ and $\vec y_2$ are solutions to the constraint}\\
        {}={}& \beta.
    \end{align*}
    Therefore, the integer program in~\eqref{eq:convex-quadratic-integer-program} is a convex quadratic integer program, which can be solved in $\ptime$ when the number of variables is fixed~\cite[Theorem~1.2]{KhachiyanP00}.
\end{proof}

\begin{proof}[Proof for the hinge loss function.]
    This proof is a straightforward adaptation of the proof for the $L^1$ norm. 
    Let $\phi(\vec x \varsep z \varsep y) \coloneqq \bigvee\nolimits_{i=1}^m \psi_i(\vec x, \vec y) \land z = \ell_i(\vec x, \vec y)$ be a piecewise formula satisfying the properties described in~Remark~\ref{remark:what-happened-in-theorem-four} (depending on the structure considered). 
    We follow the proof of Lemma~\ref{lemma:piecewise-functions-reduction}. The formula $\gamma(\vec x,z,u,v \varsep \vec y)$ in that proof is replaced by the formula: 
    \[
        \bigvee\nolimits_{i=1}^m \bigvee\nolimits_{j \in \{0,1\}} (u = i \land v = j \land \psi_i(\vec x, \vec y) \land (0 \sim_j 1 - z \cdot \ell_i(\vec x, \vec y))).
    \]
    where $\sim_0$ denotes $\geq$ and $\sim_1$ denotes $<$. 
    We then proceed as in the proof of Lemma~\ref{lemma:piecewise-functions-reduction}, and relying
    on~\nip (all structures in the statement of Theorem~\ref{thm:approx-fitting-piecewise-reals-and-PA} are~\nip),
    show that for a given sample $S$ and tolerance $\epsilon$, 
    we can compute in polynomial time a formula $\phi'(\vec y)$ of the form 
    \begin{align*}
        \bigvee\nolimits_{T \in \mathcal{T}} \big(&\,\epsilon^2 \geq {\textstyle\sum\nolimits_{(\vec s,t,i,j) \in T} j \cdot (1 - t \cdot \ell_i(\vec s, \vec y))} \land \bigwedge\nolimits_{(\vec s, t, i,j) \in T} (\psi_i(\vec s, \vec y) \land 0 \sim_j 1-t \cdot \ell_i(\vec s, \vec y))\big).
    \end{align*}
    We then proceed as in the proof of Theorem~\ref{thm:approx-fitting-piecewise-reals-and-PA} for the $L^1$ norm.
\end{proof}

\subsection{Approximate fitting problems for distribution classes}\label{app:ltwoloss:distribution-classes}

We show that Proposition~\ref{prop:easydecidabilitydistribution} 
continues to hold for the $L^2$ norm and the hinge loss function. 
This directly implies that the NP upper bound in Corollary~\ref{cor:distclassconceptbuchi}, as well as Theorem~\ref{thm:finitevcandptimeimpliesdistptime} and Corollary~\ref{cor:rqcmodelimpliesdistclassconceptptimefitting} hold for these loss functions as well. We recall the proposition.

\ProprEasyDecidabilityDistribution*

\begin{proof}
    In the proof of Proposition~\ref{prop:easydecidabilitydistribution} for the $L^1$ norm, we 
    discussed how the approximate fitting problem reduces to solving a problem over the real ordered group, of the form 
    \[
        \sum\nolimits_{i=1}^n |t_i - {\textstyle\sum_{\psi \in F_i}} x_\psi| \leq \epsilon 
        \land \sum\nolimits_{\psi \in F} x_\psi = 1
        \land \bigwedge\nolimits_{\psi \in F} x_\psi \geq 0.
    \]
    This problem has exponentially many variables, but by appealing to Carath\'eodory's theorem we can show that it is sufficient to consider only polynomially many variables. 
    These variables can be guessed, which led us to the non-deterministic polynomial-time reduction.

    The argument is analogous for the $L^2$ norm and the hinge loss function. 
    To be more precise, consider a system of constraints of the form 
    \begin{equation}\label{eq:system-with-C}
        C(\ell_1(\vec x),\ \dots,\ \ell_n(\vec x)) 
        \land \sum\nolimits_{\psi \in F} x_\psi = 1
        \land \bigwedge\nolimits_{\psi \in F} x_\psi \geq 0.
    \end{equation}
    where $C(y_1,\dots,y_n)$ is some constraint, and each $\ell_i(\vec x)$ is a linear polynomial. 
    Suppose that this system has a solution $\vec x^*$.
    Define $m_i \coloneqq \ell_i(\vec x^*)$ for all $i \in[1..n]$, and consider then the linear program 
    \[
        \left(\bigwedge_{i=1}^n \ell_i(\vec x) = m_i \right) 
        \land \sum\nolimits_{\psi \in F} x_\psi = 1
        \land \bigwedge\nolimits_{\psi \in F} x_\psi \geq 0.
    \]
    Since $\vec x^*$ is a solution to this program, by Carath\'eodory's theorem there is a solution to this program with only $O(n)$ non-zero variables. This is also a solution to System~\eqref{eq:system-with-C}. Therefore, in non-deterministic polynomial time, we can guess which variables will have a non-zero value in a solution, and simplify System~\eqref{eq:system-with-C} to a system of polynomial size.

    We now show that, for both the $L^2$ norm and the hinge loss function, System~\eqref{eq:system-with-C} can be solved in polynomial time (note that non-deterministic polynomial time would suffice). 
    Afterwards, the proof of Proposition~\ref{prop:easydecidabilitydistribution} proceeds as for
    the case of the $L^1$ norm.

    \myparagraph{Case: $L^2$ norm} In this case, $C(y_1,\dots,y_n) \coloneqq \sum_{i=1}^n y_i^2 \leq \epsilon^2$ and $\ell_i(\vec x) \coloneqq (t_1 - {\textstyle\sum_{\psi \in F_1}} x_\psi)$. As explained in the proof of Theorem~\ref{thm:approx-fitting-piecewise-reals-and-PA-l2-hinge}, this is a convex constraint. 
    Therefore, (after reducing it to polynomially-many variables) System~\eqref{eq:system-with-C} is a convex program over the reals, which can be solved in polynomial 
    time with the ellipsoid method, see~\cite[Chapter~4]{GLS1988}.

    \myparagraph{Case: hinge loss function} Here, $C(y_1,\dots,y_n) \coloneqq \sum_{i=1}^n y_i \leq \epsilon$, and $\ell_i(\vec x) \coloneqq \max(0, 1 - t_i \cdot {\textstyle\sum_{\psi \in F_i}} x_\psi)$. 
    We observe that $\ell_i(\vec x) \coloneqq \max(0, 1 - t_i \cdot {\textstyle\sum_{\psi \in F_i}} x_\psi)$ is a convex constraint, 
    hence solvability in polynomial 
    time follows again by the ellipsoid method. 
    Indeed, $1 - t_i \cdot {\textstyle\sum_{\psi \in F_i}} x_\psi$ is linear, hence convex, 
    and so is the function $x \mapsto \max(0,x)$: 
    given $\lambda \in [0,1]$ and $x_1, x_2$, we have that
    \begin{align*}
        \max(0, \lambda \cdot x_1  + (1-\lambda) \cdot x_2)
        \leq {} & \max(0, \lambda \cdot  x_1) + \max(0, (1-\lambda) \cdot x_2)\\
        = {} & \lambda \cdot \max(0,x_1) + (1-\lambda) \cdot \max(0,x_2).
        \qedhere
    \end{align*}
\end{proof}

\subsection{Approximate fitting problems for expectation classes}\label{app:ltwoloss:expectation-classes}

We conclude by showing that our results on expectation classes also hold for the $L^2$ norm and the hinge loss function, by simple adaptations to the proofs in Appendix~\ref{app:expectation}. 
In that appendix, the first step was establishing 
the following lemma. 

\LemmaReformulationExpectation*

A similar lemma holds for the $L^2$ norm and the hinge loss function, but where System~\eqref{eq:systemExpectReals} is replaced by the following systems, respectively: 
\begin{align}
    \label{eq:systemExpectRealsL2}\tag{8-$L^2$}
    \sum\nolimits_{i=1}^n (t_i - \sum\nolimits_{j=1}^{n+1} p_j \cdot f_i(y_j))^2 \leq \epsilon^2
    \ \land \ \sum\nolimits_{i=1}^{n+1} p_i = 1.\\ 
    \label{eq:systemExpectRealsHinge}\tag{8-hinge}
    \sum\nolimits_{i=1}^n \max(0, 1- t_i \cdot \sum\nolimits_{j=1}^{n+1} p_j \cdot f_i(y_j)) \leq \epsilon
    \ \land \ \sum\nolimits_{i=1}^{n+1} p_i = 1.
\end{align}
The proof is essentially the same as the one for the $L^1$ norm.

We now discuss the adaptation of the proof of Theorem~\ref{thm:approx-fitting-piecewise-reals-and-PA-expect} to the $L^2$ norm and the hinge loss function. 
We recall the statement.

\approxfittingexpect*

\begin{proof}
As with the proof for $L^1$ norm, we split into separate arguments for each structure.

    \myparagraph{Case: real ordered field} 
    It is easy to encode Systems~\eqref{eq:systemExpectRealsL2} and~\eqref{eq:systemExpectRealsHinge} as existential formulas over the real ordered field. Membership in $\exists\R$ follows.

    \myparagraph{Case: Presburger arithmetic}
    The proof follows the one for the $L^1$ norm, 
    using the fact that the functions are bounded, 
    and guessing integer values~$k_{ij}$ for the terms $f_i(y_j)$ in the systems of constraints. 
    After replacing each $f_i(y_j)$ by the guessed value $k_{ij}$, both Systems~\eqref{eq:systemExpectRealsL2} and~\eqref{eq:systemExpectRealsHinge} become 
    convex, and can be solved in polynomial time 
    with the ellipsoid method (as also explained in Appendix~\ref{app:ltwoloss:distribution-classes}).

    \myparagraph{Case: real ordered group}
    Following the same steps as in the proof for the $L^1$ norm, we can reduce the problem of solving Systems~\eqref{eq:systemExpectRealsL2} and~\eqref{eq:systemExpectRealsHinge}
    to the problem of solving, respectively, 
    the following formulas (which are variations of System~\ref{eq:systemExpectLinearRealsThree}):

    \begin{align}
        \label{eq:systemExpectLinearRealsThreeL2}\tag{11-$L^2$}
        \begin{aligned}
            &\sum\nolimits_{i=1}^n (t_i - \sum\nolimits_{j=1}^{n+1} (\sum\nolimits_{k=1}^{r_j} x_{jk} \cdot [h_j(\vec v_{kj})]_i))^2 \leq \epsilon^2\\
            \land&(\sum\nolimits_{j=1}^{n+1}\sum\nolimits_{k=1}^{r_j} x_{jk} = 1)
            \land \bigwedge\nolimits_{j = 1}^{n+1} \bigwedge\nolimits_{k = 1}^{r_j} x_{jk} \geq 0
            \land \bigwedge\nolimits_{j=1}^{n+1}\left((\textstyle\bigvee_{k = 1}^{r_j} x_{jk} = 0) \rightarrow \sum\nolimits_{k=1}^{r_j} x_{jk} = 0\right).
        \end{aligned}
    \end{align}
    \begin{align}
        \label{eq:systemExpectLinearRealsThreeHinge}\tag{11-hinge}
        \begin{aligned}
            &\sum\nolimits_{i=1}^n \max(0, 1- t_i \cdot \sum\nolimits_{j=1}^{n+1} (\sum\nolimits_{k=1}^{r_j} x_{jk} \cdot [h_j(\vec v_{kj})]_i)) \leq \epsilon\\
            \land&(\sum\nolimits_{j=1}^{n+1}\sum\nolimits_{k=1}^{r_j} x_{jk} = 1)
            \land \bigwedge\nolimits_{j = 1}^{n+1} \bigwedge\nolimits_{k = 1}^{r_j} x_{jk} \geq 0
            \land \bigwedge\nolimits_{j=1}^{n+1}\left((\textstyle\bigvee_{k = 1}^{r_j} x_{jk} = 0) \rightarrow \sum\nolimits_{k=1}^{r_j} x_{jk} = 0\right).
        \end{aligned}
    \end{align}
    After guessing which variables $x_{jk}$ are non-zero, 
    these formulas reduce to convex programs over the reals, which (once more)
    can be solved in polynomial time with the ellipsoid method.
\end{proof}

\endgroup

\end{document}